\documentclass[aos,preprint]{imsart}

\RequirePackage[OT1]{fontenc}
\RequirePackage{amsthm,amsmath}
\RequirePackage[colorlinks,citecolor=blue,urlcolor=blue]{hyperref}

\usepackage[comma,authoryear]{natbib}
\usepackage[noend]{algpseudocode}
\usepackage{chngcntr}
\usepackage{appendix}
\usepackage{amsmath}
\usepackage{amssymb,amsfonts,mathabx}
\usepackage{amsthm}
\usepackage{bm,color,tabularx,algorithm,algorithmicx}
\usepackage{subfigure,graphicx,graphics}
\usepackage{tikz,verbatim,bibentry}
\usepackage{mathrsfs,enumerate,epstopdf,lscape}
\usepackage[normalem]{ulem}

\numberwithin{equation}{section}
\newtheorem{theorem}{Theorem}[section]
\newtheorem{lemma}{Lemma}[section]
\newtheorem{proposition}{Proposition}[section]
\newtheorem{cor}{Corollary}[section]

\theoremstyle{definition}
\newtheorem{rmk}{Remark}[section]

\newtheorem{hypothesis}{Hypothesis}

\newtheorem{assumption}{Assumption}

\newcommand{\utwi}[1]{\mbox{\boldmath $ #1$}}
\newcommand{\ba}{{\utwi{a}}}

\newcommand{\lam}{{\lambda}}

\newcommand{\cA}{{\cal A}}
\newcommand{\cB}{{\cal B}}
\newcommand{\cC}{{\cal C}}

\newcommand{\cE}{{\cal E}}
\newcommand{\cF}{{\cal F}}
\newcommand{\cG}{{\cal G}}
\newcommand{\cI}{{\cal I}}
\newcommand{\cL}{{\cal L}}
\newcommand{\cM}{{\cal M}}

\newcommand{\cU}{{\cal U}}
\newcommand{\cV}{{\cal V}}

\newcommand{\cX}{{\cal X}}

\newcommand{\cZ}{{\cal Z}}

\newcommand{\sF}{{\mathscr F}}
\newcommand{\sR}{{\mathscr R}}
\newcommand{\sP}{{\mathscr P}}

\newcommand{\bes}{\begin{eqnarray*}}
\newcommand{\ees}{\end{eqnarray*}}
\newcommand{\bei}{\begin{itemize}}
\newcommand{\eei}{\end{itemize}}
\newcommand{\beiftnt}{\begin{itemize}\footnotesize}

\newcommand{\bel}{\begin{eqnarray}\label}
\newcommand{\eel}{\end{eqnarray}}

\def\Var{\textsf{Var}}
\def\E{\mathbb{E}}
\def\P{\mathbb{P}}
\def\Q{\mathbb{Q}}
\def\RR{\mathbb{R}}
\def\R{\mathbb{R}}

\def\tr{\mathrm{tr}}
\def\TV{\mathsf{TV}}
\DeclareSymbolFont{rmlargesymbols}{OMX}{mdbch}{m}{n}
\DeclareMathSymbol{\rmintop}{\mathop}{rmlargesymbols}{82}
\newcommand{\rmint}{\rmintop\nolimits}
\def\rmintd{\, {\mathrm{d}}}

\DeclareMathOperator*{\mat1}{\text{mat}_1}
\DeclareMathOperator*{\vec1}{\text{vec}}

\def\cano{(\text{\footnotesize cano})}
\def\ideal{(\text{\footnotesize ideal})}
\def\add{(\text{\footnotesize add})}
\def\iter{(\text{\footnotesize iter})}

\def\tipup{(\text{\tiny TIPUP})}
\def\topup{(\text{\tiny TOPUP})}

\begin{document}
\begin{frontmatter}
\title{Tensor Factor Model Estimation by Iterative Projection}

\runtitle{Tensor Factor Models by Iterative Projection}

\begin{aug}
\author[A]{\fnms{Yuefeng} \snm{Han}
\ead[label=e1,mark]{yuefeng.han@nd.edu}},
\author[B]{\fnms{Rong} \snm{Chen}
\ead[label=e2,mark]{rongchen@stat.rutgers.edu}},
\author[C]{\fnms{Dan} \snm{Yang}
\ead[label=e3]{dyanghku@hku.hk}},
\and
\author[B]{\fnms{Cun-Hui} \snm{Zhang}
\ead[label=e4,mark]{czhang@stat.rutgers.edu}}

\address[A]{
Department of Applied and Computational Mathematics and Statistics,
University of Notre Dame,
Notre Dame, IN 46556,
USA,
\printead{e1}}

\address[B]{
Department of Statistics,
Rutgers University,
Piscataway, NJ 08854,
USA,
\printead{e2,e4}}

\address[C]{
Faculty of Business and Economics,
The University of Hong Kong,
Hong Kong,
\printead{e3}}


\runauthor{Y. Han, R. Chen, D. Yang and C. Zhang}


\end{aug}

\begin{abstract}
  Tensor time series, which is a time series consisting of tensorial observations, has become ubiquitous. It typically exhibits high dimensionality. One approach for dimension reduction is to use a factor model structure, in a form similar to Tucker tensor decomposition, except that the time dimension is treated as a dynamic process with a time dependent structure. In this paper we introduce two approaches to estimate such a tensor factor model by using iterative orthogonal projections of the original tensor time series. These approaches extend the existing estimation procedures and improve the estimation accuracy and convergence rate significantly as proven in our theoretical investigation. Our algorithms are similar to the higher order orthogonal projection method for tensor decomposition, but with significant differences due to the need to unfold tensors in the iterations and the use of autocorrelation. Consequently, our analysis is significantly different from the existing ones. Computational and statistical lower bounds are derived to prove the optimality of the sample size requirement and convergence rate for the proposed methods. Simulation study is conducted to further illustrate the statistical properties of these estimators.

\end{abstract}



\end{frontmatter}

\section{Introduction} \label{section:introduction}

Motivated by a diverse range of modern scientific applications, analysis of tensors, or multi-dimensional arrays, has emerged as one of the most important and active research areas in statistics, computer science, and machine learning. Large tensors are encountered in genomics \citep{alter2005, omberg2007}, neuroimaging analysis \citep{zhou2013, sun2017}, recommender systems
\citep{bi2018}, computer vision \citep{liu2012}, community detection
\citep{anandkumar2014}, among others. High-order tensors often bring about high dimensionality and impose significant computational challenges. For example, functional MRI produces a time series of 3-dimensional brain images, typically consisting
of hundreds of thousands of voxels observed over time. Previous work has developed various tensor-based methods for independent and identically distributed (i.i.d.)\,tensor data or tensor data with i.i.d.\,noise.
However, the statistical framework for general tensor time series data is much less studied in the literature.

Factor analysis is one of the most useful tools for understanding common dependence among multi-dimensional outputs. Over the past decades, vector factor
models have been extensively studied in the statistics and economics communities. For instance, \citet{chamberlain1983}, \citet{bai2002}, \citet{stock2002} and \citet{bai2003} developed the static factor model using principal component analysis (PCA). They assumed that the common factors must have impact on most of the time series, and weak serial dependence is allowed for the idiosyncratic noise process. \citet{fan2011, fan2013, fan2018} established large covariance matrix estimation based on the static factor model. The static factor model has been further extended to the dynamic factor model in \citet{forni2000}.
In the dynamic factor model, the latent factors are assumed to follow a time series process, which is commonly taken to be a vector autoregressive process. \citet{fan2016} studied semi-parametric factor
models through projected principal component analysis. \cite{pena1987identifying}, \citet{pan2008}, \citet{lam2011} and \citet{lam2012} adopted another type of factor model. They assumed that the latent factors 
{capture} all dynamics of the observed process, and thus the idiosyncratic noise process has no serial dependence.
We will adopt this approach.
We note that the factor process may have complex dynamic behavior, resulting in complex dynamics of the observed tensor, even with white additive noise process. Of course, when all the dynamics of the observed tensor process are `forced' to be included in the signal process induced by the factor process, situations may arise in which some factors are `weak' (or have impact on a small portion of the observed series in the tensor). {This} leads us to consider the `signal strength' in our investigation.

Although there have been significant efforts in developing methodologies and theories for vector factor models, there is a paucity of literature on
matrix- or tensor-valued time series. \citet{wang2019} proposed a matrix factor model for matrix-valued time series, which explores the matrix structure. \citet{chen2020constrained} established a general framework for incorporating domain and prior knowledge in the matrix factor model through linear constraints. \citet{chen2022modeling} applied the matrix factor model to the dynamic transport network. \citet{chen2023statistical} developed an inferential theory of the matrix factor model under a different setting from that in \citet{wang2019}. \cite{chang2023modelling,han2023tensor,han2021cp} studied factor models with CP type low rank structures.

Recently, \citet{chen2022factor} introduced a factor approach for analyzing
high dimensional dynamic tensor time series in the form
\begin{equation} \label{eq:tensorfactor}
\cX_t=\cM_t+\cE_t,
\end{equation}
where $\cX_1,...,\cX_T\in\RR^{d_1\times\cdots\times d_K}$ are the observed tensor time series, $\cM_t$ and $\cE_t$ are the corresponding signal and noise components of $\cX_t$, respectively. The goal is to estimate the unknown signal tensor $\cM_t$ from the tensor time series data. Following \citet{lam2012}, it is
assumed that the signal tensor accommodates all dynamics, making the idiosyncratic noise $\cE_t$ uncorrelated (white) across time.
%
%
It is further assumed that $\cM_t$ lives in a lower dimensional space and has certain multilinear decomposition. Specifically, we assume that
$\cM_t$ satisfies a Tucker-type decomposition and model \eqref{eq:tensorfactor}
can be written as
\begin{equation} \label{eq:tucker-model}
\cX_t=\cF_t\times_1 A_1\times_2\ldots\times_K A_K+\cE_t,
\end{equation}
where $A_k$ is the deterministic loading matrix of size $d_k\times r_k$ and $r_k\ll d_k$, and the core tensor $\cF_t$ itself is a latent tensor factor process of dimension $r_1\times\ldots \times r_K$.
Here the $k$-mode product of $\cX\in\RR^{d_1\times d_2\times \cdots \times d_K}$ with a matrix $U\in\RR^{d_k'\times d_k}$, denoted as $\cX\times_k U$,
  is an order $K$-tensor of size
  $d_1\times \cdots \times d_{k-1} \times d_k'\times d_{k+1}\times \cdots \times d_K$ such that
$$ (\cX\times_k U)_{i_1,...,i_{k-1},j,i_{k+1},...,i_K}=\sum_{i_k=1}^{d_k} \cX_{i_1,i_2,...,i_K} U_{j,i_k}.$$

The core tensor $\cF_t$ is usually much smaller than $\cX_t$ in dimension. This structure provides an effective dimension reduction, as all the comovements of individual time series in $\cX_t$ are driven by $\cF_t$. Without loss of generality, assume that $A_k$ is of rank $r_k \ll d_k$.
It should be noted that vector and matrix factor models can be viewed as special cases of our model since a vector time series is a tensor time series composed of a single fiber ($K=1$), and a matrix times series is one composed of a single slice ($K=2$).

\citet{chen2022factor} proposed two estimation procedures, namely TOPUP and TIPUP,
for estimating the column space spanned by the loading
matrix $A_k$, for $k=1,\ldots,K$. The two procedures are based on different
auto-cross-product operations of the observed tensors $\cX_t$
to accumulate information, but they both utilize the assumption that the noise $\cE_t$ and $\cE_{t-h}, ~h>0$ are uncorrelated. The convergence rates of
their estimators critically depend on $d=d_1d_2\ldots d_K$, a potentially
very large number as $d_k,~ k=1,\ldots,K$, are large. Often a large $T$, the
length of the time series, is required for accurate estimation of the loading spaces.

In this paper we propose extensions of the TOPUP and TIPUP procedures, motivated by the following observation. Suppose that
the loading matrices $A_k$ are orthonormal with $A_k^\top A_k=I$, and we are given
$A_2,\ldots, A_K$. Let
\[
\cZ_t=\cX_t\times_2 A_2^\top\times_3 \ldots \times_K A_K^\top;
\mbox{\ \ and \ \ }
\cE_t^*=\cE_t\times_2 A_2^\top\times_3 \ldots \times_K A_K^\top;
\]
Then (\ref{eq:tucker-model}) leads to
\begin{equation}
\cZ_t=\cF_t\times_1 A_1+\cE_t^* \label{eq:ideal}
\end{equation}
where $\cZ_t$ is a $d_1\times r_2\times \ldots \times r_{K}$ tensor. Since
$r_k\ll d_k$, $\cZ_t$ is a much smaller tensor than $\cX_t$. Under proper
conditions on the combined noise tensor $\cE_t^*$, the estimation of
the loading space of $A_1$ based on $\cZ_t$ can be made significantly
more accurate, as the convergence rate now depends on
$d_1r_2\ldots r_{K}$ rather than $d_1d_2\ldots d_{K}$.

Of course, in practice we do not know $A_2,\ldots, A_K$. Similar to
backfitting algorithms, we propose an iterative algorithm. With a proper
initial value, we iteratively estimate the loading space of $A_k$ at iteration
$j$ based on
\[
\cZ_{t,k}^{(j)}=\cX_t\times_1\widehat{A}_1^{(j)\top}\times_2 \ldots
\times_{k-1} \widehat{A}_{k-1}^{(j)\top}\times_{k+1} \widehat{A}_{k+1}^{(j-1)\top}
\times_{k+2} \ldots \times_K \widehat{A}_{K}^{(j-1)\top},
\]
using the estimate $\widehat{A}_{k'}^{(j-1)},~ k<k'\le K$ obtained in the previous iteration and the estimate $\widehat{A}_{k'}^{(j)},~ 1\le k'< k$, obtained in the current iteration.
Our theoretical investigation shows that
the iterative procedures for estimating $A_1$
can achieve the convergence rate as if all $A_2,\ldots,A_K$ are known and
we indeed observe $\cZ_t$ that follows model (\ref{eq:ideal}). We call the
procedure iTOPUP and iTIPUP, based on the matrix unfolding mechanism used,
corresponding to TOPUP and TIPUP procedures. To be more specific, our algorithms have two steps:
(i) We first use the estimated column space of factor loading matrices of TOPUP (resp. TIPUP) to construct the initial estimate of factor loading spaces;
(ii) We then iteratively perform matrix unfolding of the auto-cross-moments of much smaller tensors $\cZ_{t,k}^{(j)}$ to obtain the final estimator.


We note that the iterative procedure is related to
higher order orthogonal iteration (HOOI) that has been widely studied in the
literature; see, e.g., \cite{de2000}, \cite{sheehan2007}, \cite{liu2014}, \cite{zhang2018tensor}, among others.
However, most of the existing works
are not designed for tensor time series.
They do not consider
the special role of the time mode nor the covariance structure in the time direction. Typically HOOI treats the signal part as fixed or deterministic.
In this paper we treat the signal as dynamic in the sense that the core tensor $\cF_t$ in (\ref{eq:tucker-model}) is dynamic and the relationship between $\cF_t$ and the lagged $\cF_{t-h}$ is of interest. Our setting requires special treatment although each iteration of our iterative procedures also consists of power up and orthogonal projection operations. While HOOI applies the SVD directly to the matrix unfolding of the iteratively projected data, in our approach the SVD is applied to the matrix unfolding of the outer- and inner-auto-cross-product of the iteratively projected data, respectively in iTOPUP and iTIPUP.
Although the iTOPUP algorithm proposed here can be reformulated as a twist of HOOI on the auto-cross-moment tensor, the iTIPUP algorithm is different
and cannot be recast equivalently as HOOI.
More importantly, the theoretical analysis and theoretical properties of
the estimators are fundamentally different from those of HOOI, due to the
dynamic structure of tensor time series and the need to use the auto-cross-product operation between the SVD and data projection in each iteration.
Different concentration inequalities are derived to study the performance bounds.

In this paper, we establish upper bounds on the estimation errors for both the iTOPUP and the iTIPUP, which are much sharper than the respective theoretical guarantees for TOPUP and TIPUP, demonstrating the benefits of using iterative projection. It is also shown that the number of iterations needed for convergence is of order no greater than $\log(d)$. We mainly focus on the cases where the tensor dimensions are large and of similar order. We also cover the cases where
the ranks of the tensor factor process increase with the dimensions of the tensor time series.

\citet{chen2022factor} showed that the TIPUP has a faster convergence rate in estimation error than the TOPUP, under a mild condition on the level of signal cancellation. In contrast, the theoretically guaranteed rate of convergence for the iTOPUP in this paper is of the same order or even faster than that for the iTIPUP under certain regularity conditions. Our results also suggest an interesting phenomenon. Using the iterative procedures, we find that the increase in either dimension or sample size can improve the estimation of the factor loading space of the tensor factor model with the tensor order $K\ge 2$. We believe that such a super convergence rate is new in the literature. Specifically, under
proper regularity conditions, the convergence rate of the iterative procedures for estimating the space of $A_k$ is $O_{\P}(T^{-1/2}d_{-k}^{-1/2})$, where $d_{-k}=\prod_{j\ne k}d_j$, while the existing rate for non-iterative procedures is $O_{\P}(T^{-1/2})$ for the vector factor model \citep{lam2011}
and the matrix/tensor factor models \citep{wang2019,chen2022factor}.
While the increase in the dimensions $d_k$ ($k=1,\ldots,K$) does
not improve the performance of the non-iterative estimators, it significantly improves that of the proposed iterative estimators.

In addition, we establish the computational lower bound for the estimation of the loading spaces of tensor factor models under the hardness assumption of certain instances of hypergraphic planted clique detection problem. It shows that the sample size requirement (or signal to noise ratio condition) needed for using
the TIPUP estimate as the initial values for the iterative procedures
is unavoidable for any computationally manageable estimation procedure to achieve consistency, although the iterative procedures have faster convergence rates. Furthermore, we provide a statistical lower bound that matches the convergence rates of our iterative procedures under certain conditions, revealing a different effect of the ranks $r_k$ ($k=1,...,K$) compared to tensor Tucker decomposition \citep{zhang2018tensor}.

{\it Related work.} We close this section by highlighting several recent papers on related topics. First, we draw attention to the work of \cite{foster1996time, fan2016} and \cite{chen2024semi}. \cite{chen2024semi} adopts an estimation precedure composed of a spectral initialization followed by an iterative refinement step,
so that our methods are related to theirs. However, due to the differences in problem setting and model assumptions,
their estimation procedures, performance bounds and analytic techniques are all significantly different from ours.  \cite{foster1996time, fan2016}
use the projection to the space spanned by the sieve bases without iteration.
\cite{rogers2013multilinear} assumes the tensor factor model in \eqref{eq:tucker-model}, with an additional specific AR structure on the dynamic of the factor process. The additional model structure in their paper led to an EM type of estimation approach, quite different from the approach we develop here. \cite{wang2024high}
concerns low rank tensor AR model and uses a nuclear norm penalty to enforce the low rank structure and optimization algorithms for estimation, again quite different from our approach.



The paper is organized as follows. Section \ref{section:notation} introduces basic notation and preliminaries of tensor analysis. We present the tensor factor model and the iTOPUP and iTIPUP procedures in Sections  \ref{section:tensormodel} and \ref{section:estimating}. Theoretical properties of the iTOPUP and iTIPUP are investigated in Section \ref{section:thm}. Section \ref{section:summary} provides a brief summary. Numerical comparison of our iterative procedures and other methods, and all technical details are relegated to the Supplementary Material.

\section{Tensor Factor Model by Orthogonal Iteration} \label{section:model}

\subsection{Notation and preliminaries for tensor analysis} \label{section:notation}

Throughout  this  paper, for a vector $x=(x_1,...,x_p)^\top$, define $\|x\|_q = (x_1^q+...+x_p^q)^{1/q}$, $q\ge 1$. For a matrix $A = (a_{ij})\in \RR^{m\times n}$, write the SVD as $A=U\Sigma V^\top$, where $\Sigma=\text{diag}(\sigma_1(A), \sigma_2(A), ..., \sigma_{\min\{m,n\}}(A))$, with
singular values $\sigma_{\max}(A) = $ $\sigma_1(A)\ge\sigma_2(A)\ge \cdots\ge \sigma_{\min\{m,n\}}(A)\ge 0$ in descending order. The matrix 
spectral norm is denoted as
$\|A\|_{\rm S}
=\sigma_1(A).$ Write $\sigma_{\min}(A)$ 
the smallest 
nontrivial singular value of $A$.
For two sequences of real numbers $\{a_n\}$ and $\{b_n\}$, write $a_n=O(b_n)$ (resp. $a_n\asymp b_n$) if there exists a constant $C$ such that $|a_n|\leq C |b_n|$ (resp. $1/C \leq a_n/b_n\leq C$) for all sufficiently large $n$, and write $a_n=o(b_n)$ if $\lim_{n\to\infty} a_n/b_n =0$. Write $a_n\lesssim b_n$ (resp. $a_n\gtrsim b_n$) if there exist a constant $C$ such that $a_n\le Cb_n$ (resp. $a_n\ge Cb_n$). Write $a\wedge b=\min\{a,b\}$ and $a\vee b=\max\{a,b\}$. We use $C, C_1,c,c_1,...$ to denote generic constants, whose actual values may vary from line to line.

For any two $m\times r$ matrices with orthonormal columns, say, $U$ and $\widehat U$, suppose the singular values of $U^\top \widehat U$ are $\sigma_1\ge \sigma_2 \ge \cdots \ge \sigma_r\ge 0$.
A natural measure of distance between the column spaces of $U$ and $\widehat U$ is then
\begin{equation}\label{loss}
\|\widehat U\widehat U^\top - UU^\top\|_{\rm S}=\sqrt{1-\sigma_r^2}, 
\end{equation}
which equals to the sine of the largest principle angle
between the column spaces of $U$ and $\widehat U$.
For any two matrices $A\in\RR^{m_1\times r_1},B\in \RR^{m_2\times r_2}$, denote the Kronecker product $\odot$ as $A\odot B\in \RR^{m_1 m_2 \times r_1 r_2}$. For any two tensors $\cA\in\RR^{m_1\times m_2\times \cdots \times m_K}, \cB\in \RR^{r_1\times r_2\times \cdots \times r_N}$, denote the tensor product $\otimes$ as $\cA\otimes \cB\in \RR^{m_1\times \cdots \times m_K \times r_1\times \cdots \times r_N}$, such that
$$(\cA\otimes\cB)_{i_1,...,i_K,j_1,...,j_N}=(\cA)_{i_1,...,i_K}(\cB)_{j_1,...,j_N} .$$
Let ${\rm{vec}}(\cdot)$ be the vectorization of matrices and tensors. The mode-$k$ unfolding (or matricization) is defined as ${\rm{mat}}_k(\cA)$, which maps a tensor $\cA$ to a matrix
${\rm{mat}}_k(\cA)\in\RR^{m_k\times m_{-k}}$ where $m_{-k}=\prod_{j\neq k}^K m_j$.  For example, if $\cA\in\RR^{m_1\times m_2\times m_3}$, then
$$({\rm{mat}}_1(\cA))_{i,(j+m_2(k-1))}= ({\rm{mat}}_2(\cA))_{j,(k+m_3(i-1))}= ({\rm{mat}}_3(\cA))_{k,(i+m_1(j-1))} =\cA_{ijk}.  $$
For tensor $\cA\in\RR^{m_1\times m_2\times \cdots \times m_K}$, the Hilbert Schmidt norm is defined as
$$ \|\cA\|_{{\rm HS}}=\sqrt{\sum_{i_1=1}^{m_1}\cdots\sum_{i_K=1}^{m_K}(\cA)_{i_1,...,i_K}^2 }. $$
For a matrix, the Hilbert Schmidt norm is
just the Frobenius norm.
Define the tensor operator norm for an order-4 tensor $\cA\in\RR^{m_1\times m_2\times m_3\times m_4}$,
$$ \| \cA\|_{\rm{op}} =\max\left\{ \sum_{i_1,i_2,i_3,i_4} u_{i_1,i_2} \cdot u_{i_3,i_4}\cdot (\cA)_{i_1,i_2,i_3,i_4}:\|U_1\|_{\rm HS}=\|U_2\|_{\rm HS}=1 \right\},$$
where {$U_1=(u_{i_1,i_2})\in\RR^{m_1\times m_2}$ and $U_2=(u_{i_3,i_4})\in\RR^{m_3\times m_4}$.}


\subsection{Tensor factor model} \label{section:tensormodel}

Again, we consider {as in \eqref{eq:tucker-model}}
\[
\cX_t=\cF_t\times_1 A_1\times_2\ldots\times_K A_K+\cE_t.
\]
Without loss of generality, assume that $A_k$ is of rank $r_k$.
$A_k$ is not necessarily orthonormal, which is different from the classical Tucker decomposition \citep{tucker1966}.
Model \eqref{eq:tucker-model} is unchanged if we replace $(A_1,...,A_K, \cF_t)$ by $(A_1H_1,...,A_KH_K, \cF_t\times_{k=1}^K H_k^{-1})$ for any invertible $r_k\times r_k$ matrix $H_k$.
Although $(A_1,...,A_K, \cF_t)$ are not uniquely determined, the factor loading space, that is, the linear space spanned by the columns of $A_k$, is uniquely defined. Denote the orthogonal projection to the column space of $A_k$ as
\begin{equation}\label{P_k}
P_k=P_{A_k}=A_k (A_k^\top A_k)^{-1} A_k^\top = U_kU_k^\top,
\end{equation}
{where $U_k$ is the left singular matrix in the SVD $A_k=U_k\Lambda_k V_k^\top$.}
We use $P_k$ to represent the factor loading space of $A_k$. Thus, our objective is to estimate $P_k$.

The canonical representation of the tensor times series \eqref{eq:tucker-model} is written as
$$
\cX_t=\cF_t^{\cano} \times_{k=1}^K U_k+\cE_t,
$$
where
the diagonal and right singular matrices of $A_k$ are absorbed into the canonical core tensor
$\cF_t^{\cano} = \cF_t\times_{k=1}^K (\Lambda_k V_k^\top)$.
In this canonical form, the loading matrices $U_k$ are identifiable up to a rotation in general and up to a permutation and sign changes of the columns of $U_k$ when the singular values are all distinct in the population version of the  TOPUP or TIPUP methods, as we describe in Section \ref{section:estimating} below.
In what follows, we may identify the tensor time series in
its canonical form, i.e. $A_k=U_k$, without explicit declaration.

We do not impose any specific structure for the dynamics of the core tensor factor process $\cF_t\in\RR^{r_1\times \cdots\times r_K}$
beyond the independence between the core process and the noise process, and we
do not require any additional structure on the
correlation among different time series fibers of the noise process $\cE_t$.
Because of this generality, our estimator is based on the tensor version of the lagged sample cross product
$\widehat \Sigma_h$, $h=1,...,h_0$, where
\begin{equation}
\widehat \Sigma_h
=\widehat \Sigma_h(\cX_{1:T})
= \sum_{t=h+1}^T \frac{\cX_{t-h}\otimes\cX_t}{T-h}
\in\RR^{d_1\times\cdots\times d_K\times d_1\times \cdots\times d_K}
\label{autocov}
\end{equation}
is an order-$2K$ tensor.
The population version of this tensor autocovariance is
\begin{equation*}
\Sigma_h=\E\left( \sum_{t=h+1}^T\frac{\cX_{t-h}\otimes\cX_t}{T-h} \right) = \E\left( \sum_{t=h+1}^T\frac{\cM_{t-h}\otimes\cM_t}{T-h} \right).
\end{equation*}
{Because} $\cM_t=\cM_t\times_{k=1}^K P_k$ for all $t$,
\begin{equation*}
\Sigma_h=\Sigma_h\times_{k=1}^{2K} P_k = \E\left( \sum_{t=h+1}^T\frac{\cF_{t-h}\otimes\cF_t}{T-h} \right)  \times_{k=1}^{2K} P_k A_k,
\end{equation*}
with the notation $A_k=A_{k-K}$ and $P_k=P_{k-K}$ for all $k>K$.

\subsection{Estimating procedures} \label{section:estimating}

In this paper, we consider iterative estimation procedures to achieve sharper convergence rates than the TOPUP and TIPUP procedures proposed in \citet{chen2022factor}. We start with a quick description of their procedures as they serve as the starting point of our proposed iTOPUP and iTIPUP procedures. Note that the procedure in \cite{chen2022modeling} and \cite{wang2019} is the non-iterative TOPUP.

\smallskip
{\bf \noindent (i) Time series Outer-Product Unfolding Procedure (TOPUP)}:

Let $\widehat \Sigma_h$ be the sample autocovariance of the data $\cX_{1:T}=(\cX_1,\ldots,\cX_T)$ as in \eqref{autocov}. Define
\begin{equation} \label{eq:topup:def}
{{\rm{TOPUP}}_k}=\left(
{\rm{mat}}_k\big(\widehat \Sigma_h\big),
\ h=1,...,h_0 \right),
\end{equation}
as a $d_k\times (dd_{-k}h_0)$ matrix,
where $d=\prod_{k=1}^K d_k$, $d_{-k}=d/d_k$ and $h_0$ is a predetermined positive integer. Here we note that
${{\rm{TOPUP}}_k}$ is a function mapping a tensor time series to a matrix. In
${{\rm{TOPUP}}_k}$,
the information from different time lags is accumulated, which is useful especially when the sample size $T$ is small.
A relatively small $h_0$ is typically used, since the autocorrelation is often at its strongest with small time lags. See Remark \ref{rmk:h0}.

The TOPUP method performs SVD
of \eqref{eq:topup:def} to obtain the truncated left singular matrices
\begin{equation}
\text{$\widehat U_k$-TOPUP}(\cX_{1:T},m)
=\text{LSVD}_m\left(
{\rm{mat}}_k\big(\widehat \Sigma_h(\cX_{1:T})\big),
\ h=1,...,h_0 \right),
\label{hat-U-topup}
\end{equation}
where LSVD$_m$ stands for the left singular matrix composed of the first $m$ left singular vectors corresponding to the largest $m$ singular values.
Here $\text{$\widehat U_k$-TOPUP}$ is treated as an operator that maps a noisy tensor time series to a matrix of $m$ columns as an estimate of
the mode-$k$ singular space of the low-rank signal tensor time series.

By \eqref{eq:tucker-model} and \eqref{autocov}, the expectation of  \eqref{eq:topup:def} satisfies
\begin{align}
&\E\left[{{\rm{TOPUP}}_k}\right] \label{eq:topup:svd} \\
&= A_k \text{mat}_k\left( \sum_{t=h+1}^T \E \left(\frac{\cF_{t-h}\otimes \cF_t} {T-h} \right) \times_{l=1}^{k-1}A_l \times_{l=k+1}^{2K}A_l,\ h=1,...,h_0  \right), \notag
\end{align}
so that the TOPUP is expected to be consistent in estimating the column space of $A_k$.

\smallskip
{\bf \noindent (ii) Time series Inner-Product Unfolding Procedure (TIPUP)}:

Similar to \eqref{eq:topup:def}, define a $d_k\times (d_k h_0)$ matrix as
\begin{equation} \label{eq:tipup:def}
{{\rm{TIPUP}}_k} =
\left(\sum_{t=h+1}^T \frac{{\rm{mat}}_k(\cX_{t-h}) {\rm{mat}}_k^\top(\cX_t)} {T-h}, \ h=1,...,h_0 \right),
\end{equation}
which replaces the tensor product by the inner product {through \eqref{autocov} in} (\ref{eq:topup:def}).
The TIPUP method performs SVD:
\begin{equation}
\text{$\widehat U_k$-TIPUP}(\cX_{1:T},m)=\text{LSVD}_m\left(
\sum_{t=h+1}^T \frac{{\rm mat}_k(\cX_{t-h}) {\rm{mat}}_k^\top(\cX_t)}{T-h}, \ h=1,...,h_0
\right),
\label{hat-U-tipup}
\end{equation}
for $k=1,...,K$. Again, $\widehat U_k$-TIPUP is treated as an operator.
We note that the TOPUP method in \eqref{hat-U-topup} utilizes the entire auto-cross product tensor by applying the SVD to its mode $k$ unfolding,
whereas the TIPUP only utilizes a matrix-valued linear mapping of the auto-cross product tensor by first taking the model-$k'$ trace operation for all $k'\neq k$. The trace operation cancels the noise but also possibly some signal.


\smallskip
{\bf \noindent (iii) iTOPUP and iTIPUP}:
Next we describe a generic iterative procedure under the motivation described in Section 1. Its pseudo-code is provided in Algorithm \ref{algorithm:itopup}. It incorporates two estimators/operators \text{$\widehat U_k$-INIT} and \text{$\widehat U_k$-ITER} that map a tensor time series to an estimate of the loading matrix $U_k$.
Respectively they 
stand for the procedures used for initialization and iteration.
The \text{$\widehat U_k$-TOPUP} and \text{$\widehat U_k$-TIPUP} operators in \eqref{hat-U-topup} and \eqref{hat-U-tipup} are examples of such operators.

\begin{algorithm}[ht]
\caption{A generic iterative algorithm}\label{algorithm:itopup}
\begin{algorithmic}[1]
\State Input: $\cX_t\in\RR^{d_1\times\cdots \times d_K}$ for $t=1,...,T$, $r_k$ for all $k=1,..,K$, the tolerance parameter $\epsilon>0$, the maximum number of iterations $J$, and the \text{$\widehat U_k$-INIT} and \text{$\widehat U_k$-ITER} operators.

\State Let $j=0$, initiate via applying \text{$\widehat U_k$-INIT}
on $\{\cX_{1:T}\}$, for $k=1,...,K$, to obtain
\begin{equation*}
\widehat U_{k}^{(0)}=\text{$\widehat U_k$-INIT}_k(\cX_{1:T},r_k).
\end{equation*}

\Repeat

\State Let $j=j+1$. At the $j$-th iteration, for $k=1,...,K$, given previous estimates
$(\widehat U_{k+1}^{(j-1)},\ldots, \widehat U_{K}^{(j-1)})$ and $(\widehat U_{1}^{(j)},\ldots, \widehat U_{k-1}^{(j)})$, sequentially calculate,
\begin{equation*}
\cZ_{t,k}^{(j)}=\cX_t \times_1 (\widehat U_{1}^{(j)})^\top \times_2 \cdots \times_{k-1} (\widehat U_{k-1}^{(j)})^\top \times_{k+1} (\widehat U_{k+1}^{(j-1)})^\top \times_{k+2}\cdots\times_K (\widehat U_{K}^{(j-1)})^\top,
\end{equation*}
 for $t=1,\ldots, T$.
 Perform \text{$\widehat U_k$-ITER} on the new tensor time series $\cZ_{1:T,k}^{(j)}=(\cZ_{1,k}^{(j)},\ldots,\cZ_{T,k}^{(j)})$.
\begin{equation*}
\widehat U_{k}^{(j)}=\text{$\widehat U_k$-ITER}_k(\cZ_{1:T,k}^{(j)},r_k).
\end{equation*}

\Until $j=J$ or
\begin{equation*}
\max_{1\le k\le K}\|\widehat U_{k}^{(j)}(\widehat U_{k}^{(j)})^\top- \widehat U_{k}^{(j-1)}(\widehat U_{k}^{(j-1)})^\top\|_{\rm S}\le \epsilon,
\end{equation*}

\State Estimate and output:
\begin{align*}
\widehat U_k^{\text{iFinal}}&=\widehat U_k^{(j)},\quad k=1,...,K, \\
\widehat{P_{k}}^{\text{iFinal}} &= \widehat U_{k}^{\text{iFinal}}(\widehat U_{k}^{\text{iFinal}})^\top,\quad k=1,...,K,\\
\widehat\cF_t^{\text{iFinal}}&=\cX_t\times_{k=1}^K (\widehat U_{k}^{\text{iFinal}})^\top, \quad t=1,...,T,\\
\widehat\cE_t^{\text{iFinal}} &
{=\cX_t-
\cX_t\times_1\widehat P_1^{\text{iFinal}}  \times_2\cdots\times_K\widehat P_K^{\text{iFinal}}, \quad t=1,...,T.}
\end{align*}
\end{algorithmic}
\end{algorithm}

When we use the \text{$\widehat U_k$-TOPUP} operator \eqref{hat-U-topup} for both \text{$\widehat U_k$-INIT} and \text{$\widehat U_k$-ITER} in Algorithm~\ref{algorithm:itopup}, it will be called iTOPUP procedure. Similarly,
iTIPUP uses \text{$\widehat U_k$-TIPUP} operator \eqref{hat-U-tipup} for both \text{$\widehat U_k$-INIT} and \text{$\widehat U_k$-ITER}.
Besides these two versions, we may also use \text{$\widehat U_k$-TIPUP} for \text{$\widehat U_k$-INIT}
and \text{$\widehat U_k$-TOPUP} for \text{$\widehat U_k$-ITER}, named as TIPUP-iTOPUP. Similarly, TOPUP-iTIPUP
uses \text{$\widehat U_k$-TOPUP} as \text{$\widehat U_k$-INIT} and \text{$\widehat U_k$-TIPUP} as \text{$\widehat U_k$-ITER}.
These variants are sometimes useful, because TOPUP and TIPUP have different theoretical properties as the initializer or for iteration, as we will discuss
in Section \ref{section:thm}. Other estimators of the loading spaces based on the tensor
time series can also be used in place of \text{$\widehat U_k$-INIT} and \text{$\widehat U_k$-ITER}, such as the conventional high order SVD for tensor decomposition, which we refer to as Unfolding Procedure (UP), that simply performs SVD of the matricization along the appropriate mode of the $K+1$ order tensor $(\cX_1,\ldots, \cX_T)$ with time dimension as the additional $(K+1)$-th mode.


\begin{rmk} {While Algorithm \ref{algorithm:itopup} 
resembles an HOOI-type iteration of the orthogonal projection and singular matrix estimation methods, the proposed iTOPUP and iTIPUP are significantly different from HOOI which iterates the operations of
\begin{align*}
\hbox{orthogonal projection $\to$ matrix unfolding $\to$ SVD.}
\end{align*}
In both iTOPUP and iTIPUP, each iteration carries out the operations
\begin{align}\label{alg-1-words}
\hbox{orthogonal projection $\to$
autocovariance $\to$ matrix unfolding $\to$ SVD}.
\end{align}
As the outer product is taken with TOPUP$_k$ in \eqref{eq:topup:def}, its orthogonal projection and autocovariance operations are exchangeable, so that we can write
\begin{align*}
{{\rm{iTOPUP}}} = {\rm HOOI}(\hat{\Sigma}_h, h=1,\ldots,h_0)
\end{align*}
as long as the HOOI is modified by applying $U_\ell^{(j)}$ to both mode $\ell$ and mode $K+\ell, \ell\neq k$ in the projection operation and leaving alone the $(2K+1)$-th mode in the lags $1:h_0$ throughout.
However, for iTIPUP, the orthogonal projection and autocovariance operations in \eqref{alg-1-words} are not exchangeable as the projections are sandwiched inside the autocovariance. Needless to say, the analysis of iTOPUP and iTIPUP is much more difficult than the conventional HOOI with iid assumption} due to the involvement of the autocovarinace operations in the time-axis in the iterations.
\end{rmk}

\begin{rmk}[Rank determination] \label{rmk:rank}
Here the estimators are constructed with given ranks $r_1,\ldots,r_K$, though in theoretical analysis they are allowed to diverge. In practice, existing
procedures for rank determination in the vector factor model, including the information criteria approach \citep{bai2002,bai2007,hallin2007} and ratio of eigenvalues approach \citep{lam2012,ahn2013} can be extended to the tensor factor model by treating $d_1\times\cdots\times d_k$ tensors as $d$-dimensional vectors, $d=\prod_{k=1}^Kd_k$.
\end{rmk}

%

\section{Theoretical Properties} \label{section:thm}

In this section we present some theoretical properties of the iterative procedures. We first present the additional notation needed for the discussion, and then the error bounds for the iterative estimators under a minimum condition on the error process $\cE_t$ in the model. These error bounds are quite general and cover many different models. To help decipher the general results, we present two concrete signal process models (or general sets of assumptions) 
with simpler and more explicit convergence rates.

\subsection{Notation}


Let $\overline\E[\cdot]=\E[\cdot|\{ \cF_1,...,\cF_T\}]$. Define $d=\prod_{k=1}^K d_k$, $d_{-k}=d/d_k$, $r=\prod_{k=1}^K r_k$ and $r_{-k}=r/r_k$. Define order-4 tensors
\begin{align}\label{Theta}
& \Theta_{k,h} =\sum_{t=h+1}^T \frac{\text{mat}_k(\cM_{t-h})\otimes \text{mat}_k(\cM_{t}) }{T-h} \in\RR^{d_k\times d_{-k}\times d_k \times d_{-k}},
\\ \nonumber
& \Phi_{k,h} =\sum_{t=h+1}^T \frac{\text{mat}_k(\cF_{t-h})\otimes \text{mat}_k(\cF_{t}) }{T-h} \in\RR^{r_k\times r_{-k}\times r_k \times r_{-k}},
\\ \nonumber
& \Phi^{\cano}_{k,h} =\sum_{t=h+1}^T \frac{\text{mat}_k(\cM_{t-h}\times_{k=1}^KU_k^\top)
\otimes \text{mat}_k(\cM_{t}\times_{k=1}^KU_k^\top)}{T-h}
\in\RR^{r_k\times r_{-k}\times r_k \times r_{-k}},
\end{align}
with $U_k$ from the SVD $A_k=U_k \Lambda_k V_k^\top$.
We view $\Phi^{\cano}_{k,h}$ as the canonical version of the auto-covariance of the factor process. The noiseless version of the matrix TOPUP$_k$ in
\eqref{eq:topup:def} is
\begin{equation}\label{Theta-2}
{\text{mat}_1(\Theta_{k,1:h_0})
= \overline\E\big[ {{\rm{TOPUP}}_k}\big]
\in \R^{d_k\times(dd_{-k}h_0)}},
\end{equation}
with $\Theta_{k,1:h_0} = (\Theta_{k,h}, h=1,\ldots,h_0)$.
The canonical factor version of \eqref{Theta-2}
is $\text{mat}_1(\Phi^{{\cano}}_{k,1:h_0})\in\R^{r_k\times (rr_{-k}h_0)}$ with
$\Phi^{{\cano}}_{k,1:h_0}=(\Phi^{{\cano}}_{k,h}, h=1,\ldots,h_0)
\in \R^{r_k\times r_{-k}\times r_k \times r_{-k}\times h_0}$.
Similarly define
\begin{align}\label{Theta^*}
\Theta_{k,h}^* &=\sum_{t=h+1}^T \frac{\text{mat}_k(\cM_{t-h}) \text{mat}_k^\top(\cM_{t}) }{T-h} \in\RR^{d_k\times d_k},\\ \nonumber
\Phi_{k,h}^* &=\sum_{t=h+1}^T \frac{\text{mat}_k(\cF_{t-h}) \text{mat}_k^\top(\cF_{t}) }{T-h} \in\RR^{r_k\times r_k } , \\ \nonumber
\Phi_{k,h}^{*\cano} &= U_k^\top\Theta_{k,h}^*U_k
\\ \nonumber & =\sum_{t=h+1}^T \frac{\text{mat}_k(\cM_{t-h}\times_{k=1}^KU_k^\top )
\text{mat}_k^\top(\cM_{t}\times_{k=1}^KU_k^\top ) }{T-h}
\in\RR^{r_k\times r_k}.
\end{align}
The noiseless version of 
\eqref{eq:tipup:def} is
\begin{equation}\label{Theta^*-2}
\Theta^*_{k,1:h_0} = (\Theta_{k,h}^*, h=1,\ldots,h_0) = \overline\E\big[ {{\rm{TIPUP}}_k}\big]
\in \R^{d_k\times (d_kh_0)},
\end{equation}
and its canonical factor version is $\Phi_{k,1:h_0}^{*{\cano}}=(\Phi_{k,h}^{*{\cano}}, h=1,\ldots,h_0)
\in \R^{r_k\times (r_kh_0)}$.
Let $\tau_{k,m}$ be the $m$-th singular value of the noiseless version of the TOPUP$_k$ matrix,
$$\tau_{k,m}=\sigma_{m}\left(\,\overline\E
\big[{{\rm{TOPUP}}_k}\big]
\right)
= \sigma_{m}\big({\text{mat}}_1(\Theta_{k,1:h_0})\big)
= \sigma_{m}\big({\text{mat}}_1(\Phi^{\cano}_{k,1:h_0})\big).
$$
The signal strength for iTOPUP can be characterized as
\begin{equation}\label{lam_k}
\lambda_k=\sqrt{h_0^{-1/2}\tau_{k,r_k}} .
\end{equation}
Similarly, let
$$\tau_{k,m}^*=\sigma_{m}(\overline\E ({{\rm{TIPUP}}_k}))
= \sigma_{m}\big(
\Theta_{k,1:h_0}^*\big)
= \sigma_{m}\big(
\Phi_{k,1:h_0}^{*\cano}\big).
$$
The signal strength for iTIPUP can be characterized as
\begin{equation}\label{lam^*_k}
\lambda_k^*=\sqrt{h_0^{-1/2}\tau_{k,r_k}^*} .
\end{equation}
We note that by \eqref{Theta^*} and the Cauchy-Schwarz inequality,
\begin{align*}
\lambda_k^{*2} \le h_0^{-1/2} \| \Theta_{k,1:h_0}^* \|_{\rm S} \le
\max_{h\le h_0} \| \Theta_{k,h}^* \|_{\rm S} \le \|\Theta_{k,0}^*\|_{\rm S}/(1-h_0/T).
\end{align*}

\subsection{General error bounds}

Our general error bounds {for the proposed iTOPUP and iTIPUP}
are established under the following assumption
for the error process.

\begin{assumption}\label{asmp:error}
The error process $\cE_t$ are independent Gaussian tensors {conditionally} on the factor process $\{\cF_t,t\in\mathbb Z\}$. In addition, there exists some constant $\sigma>0$, such that
\begin{equation*}
\overline\E (u^\top \text{vec}(\cE_t))^2\le \sigma^2 \|u\|_2^2, \quad u\in\RR^d.
\end{equation*}
\end{assumption}

Assumption \ref{asmp:error} is used in \cite{chen2022factor} for the theoretical investigation of the non-iterative TIPUP and TOPUP, and is
similar to those on the noise imposed in \cite{lam2011, lam2012}. The normality assumption, which ensures fast convergence rates in our analysis, is imposed for technical convenience. It accommodates general patterns of dependence among individual time series fibers, but also allows a presentation of the main results with manageable analytical complexity. In fact, direct extension is visible in our analysis under the sub-Gaussian and even more general tail probability conditions. Under Assumption \ref{asmp:error} the magnitude of the noise can be measured by the dimension $d_k$ before the projection and by the {rank} $r_k$ after the projection. The main
theorems (Theorems~\ref{thm:itopup}, \ref{thm:itipup} and \ref{thm:tipup-itopup}) in this section
are based on this assumption on the noise alone, and cover all thereafter discussed settings of the signal $\cM_t$.

Let us first study the behavior of iTOPUP procedure.
By \citet{chen2022factor},
the risk $\overline\E\big[\big\|\widehat U_k^{(0)}\widehat U_k^{(0)\top} - U_kU_k^\top\big\|_{\rm S}\big]$
of the TOPUP estimator for $U_k$, the initialization of iTOPUP, is no larger than a constant times
\begin{align}\label{TOPUP_k-bd}
{R}_k^{(0)}
=&  \lam_k^{-2} \sigma T^{-1/2}
\Big\{\sqrt{d_kd_{-k}r_{-k}}\|\Theta_{k,0}^*\|_{\rm S}^{1/2} + \big(\sqrt{d_k}+\sqrt{d_{-k}r}\big)\|\Theta_{{k},0}\|_{\rm op}^{1/2}
\\
&\quad +\sigma \sqrt{d_k}d_{-k} + \sigma d_{k}\sqrt{d_{-k}}T^{-1/2} \Big\}, \notag
\end{align}
where $d_{-k} = \prod_{j\neq k}d_j$ and $r_{-k}=\prod_{j\neq k}r_j$.
A variation of the \cite{wedin1972} perturbation theory,
stated in Lemma \ref{lm-pertubation}, provides a sharper bound for the TOPUP estimator as follows.

\begin{proposition}\label{prop:topup}
Suppose Assumption \ref{asmp:error} holds. Let $h_0 \le T/4$. Define 
\begin{align}
\sR_{k2}=&  \lam_k^{-2} \sigma T^{-1/2}
\Big\{\sqrt{r_k}r_{-k}\|\Theta_{k,0}^*\|_{\rm S}^{1/2} + \big(\sqrt{d_k}+\sqrt{rr_{-k}}\big)\|\Theta_{{k},0}\|_{\rm op}^{1/2}
\\
&\quad +\sigma (\sqrt{d_k}+\sqrt{rr_{-k}}) + \sigma \sqrt{d_{k}r}T^{-1/2} \Big\}, \notag \\
R_k^{\topup}=&\sR_{k2} + (R_{k}^{(0)})^2. \label{bound:topup0}
\end{align}
If $\max_{1\le k\le K}R_{k}^{(0)}=o(1)$, it holds simultaneously for all $1\le k\le K$ that
\begin{align*}
\overline\E\big\|\widehat P_k^{(0)} - P_k\big\|_{\rm S} \lesssim R_k^{\topup}. 
\end{align*}
\end{proposition}

\begin{rmk}
In the rank one case ($r_k=1$, $1\le k\le K$), Proposition 1 in \cite{Ouyang2022} {\it provides} the sharpness of the above bound.
Additionally, for fixed rank, the error bound for TIPUP in \cite{chen2022rejoinder} was confirmed to be sharp by Proposition 1 in \cite{Ouyang2022}.
\end{rmk}

The aim of iTOPUP is to achieve dimension reduction
by projecting the data in other modes of the tensor time series from $\R^{d_j}$ to $\R^{r_j}$, $j\neq k$.
Ideally (e.g. when the true projection matrices $U_j$ are used),
this would reduce the rate given in \eqref{TOPUP_k-bd} and \eqref{bound:topup0} to
\begin{align}\label{iTOPUP-ideal_k}
{R}_k^{\ideal}
=& {\sR}_{k2} +{\sR}_{k1}^2,
\end{align}
by replacing all $d_j$ in $R_k^{(0)}$ with $r_j$, $j\neq k$, where
\begin{align*}
{\sR}_{k1}
=& \lam_k^{-2} \sigma T^{-1/2}
\Big\{\sqrt{d_k}r_{-k}
\|\Theta_{{k},0}^*\|_{\rm S}^{1/2} + \big(\sqrt{d_k}+\sqrt{r_{-k}r}\big)\|\Theta_{k,0}\|_{\rm op}^{1/2}
\\
&\quad +\sigma \sqrt{d_k}r_{-k}+\sigma d_k\sqrt{r_{-k}}T^{-1/2}\Big\}.
\end{align*}
However, because the iteration uses the estimated $U_j, j\neq k,$ of total dimension
$d^*_{-k} = \sum_{j\neq k} d_jr_j$, our analysis also involves
the following additional error term,
\begin{equation}
{R}_k^{\add}=\lambda_k^{-2}\sigma^2 T^{-1}\Big(d_{-k}^*+\sqrt{d^*_{-k}d_k r_{-k}}\Big).
\label{iTOPUP-add_k}
\end{equation}

The following theorem provides conditions under which the ideal rate is indeed achieved.

\begin{theorem}\label{thm:itopup}
Suppose Assumption \ref{asmp:error} holds. Let $h_0 \le T/4$
and $P_k$, $\Theta_{k,0}$, $\Theta_{k,0}^*$ and $\lambda_k$ be as in
\eqref{P_k}, \eqref{Theta}, \eqref{Theta^*} and \eqref{lam_k} respectively. Let
${R}^{(0)}=\max_{1\le k\le K}{R}^{(0)}_k$ with the ${R}^{(0)}_k$ in \eqref{TOPUP_k-bd}, ${R}^{\topup}=\max_{1\le k\le K}{R}^{\topup}_k$ with the ${R}^{\topup}_k$ in \eqref{bound:topup0}, ${R}^{\ideal}=\max_{1\le k\le K}{R}^{\ideal}_k$ with the ${R}^{\ideal}_k$ in \eqref{iTOPUP-ideal_k},
and ${R}^{\add}=\max_{1\le k\le K}{R}^{\add}_k$ with
the ${R}^{\add}_k$ in \eqref{iTOPUP-add_k}. Let $\widehat P_k^{(m)} =\widehat U_k^{(m)}\widehat U_k^{(m)\top}$ with the $m$-step estimator
$\widehat U_k^{(m)}$ in the \text{iTOPUP} algorithm.
Then, the following statements hold for a certain numerical constant
$C_{1}^{\topup}$ and a constant $C_{1,K}^{\iter}$ depending on $K$ only: When
\begin{equation}\label{condition1n}
C_{1}^{\topup}{R}^{(0)}\le (1-\rho)/4
\ \hbox{ and }\
C_{1,K}^{\iter}({R}^{\ideal}+R^{\add})\le\rho
\end{equation}
with a constant $0<\rho<1$, it holds simultaneously for all $1\le k\le K$ and $m\ge 0$ that
\begin{equation}\label{bound:itopup:cora2}
\big\|\widehat P_k^{(m)} - P_k\big\|_{\rm S}
\le 2C_{1}^{\topup}\left(
(1-\rho^m)(1-\rho)^{-1}{R}^{\ideal}
+ (\rho^m/2){R}^{\topup} \right)
\end{equation}
in an event with probability at least $1 -\sum_{k=1}^Ke^{-d_k}$.
In particular, after at most $J=\lfloor\log(\max_{k}{d_{-k}/r_{-k}})/\log(1/\rho)\rfloor$ iterations,
\begin{align}\label{bound:topup}
\overline\E\left[\max_{1\le k\le K}\big\|\widehat P_k^{(J)} - P_k\big\|_{\rm S}\right]
\le \frac{3C_{1}^{\topup}}{1-\rho}{R}^{\ideal} + \sum_{k=1}^Ke^{-d_k}.
\end{align}
\end{theorem}


\begin{rmk}
The essence of our analysis of iTOPUP is that under \eqref{condition1n},
each iteration is a contraction of the error in the estimation of
$\times_{j\neq k} U_j$ in a small neighborhood of it.
The upper bound \eqref{bound:itopup:cora2} for the error of the $m$-step estimator is comprised of two terms respectively corresponding to the cumulative iteration error and the contracted error of the initial estimator. Of course, after sufficiently large number of iterations, the first term would dominate the second as in \eqref{bound:topup}.
\end{rmk}

\begin{rmk}
The constant $C_{1}^{\topup}$ is taken in
\eqref{condition1n} to guarantee
sufficient accuracy of the initialization of iTOPUP in the following sense:
\bel{init-error-bd}
\max_{k\le K}\overline{\E} \big\|\widehat U_k^{(0)}(\widehat U_k^{(0)})^\top-P_k\big\|_{\rm S}
\le C_{1}^{\topup} R^{(0)}
\eel
with at least probability $1-8^{-1}\sum_{k=1}^K e^{-d_k}$.
The consistency of the non-iterative TOPUP estimator requires $R^{(0)}\rightarrow 0$ \citep{chen2022factor}. However, here we do not require the TOPUP estimator as the initial value to be consistent. For
\eqref{bound:itopup:cora2}
to hold, the TOPUP estimator is only required to be sufficiently close to the ground truth as in \eqref{init-error-bd}.
\end{rmk}

\begin{rmk}
It is relatively easy to verify that the first part of \eqref{condition1n} implies the second part
under many circumstances, including when $d_k$ are of the same order, $r_k$ are of the same order, and $r_k \lesssim d_k^{1-1/K}$ ($K\geq 2$).
In \cite{zhang2018tensor}, condition $\max_k r_k \lesssim \min_k d_k^{1/2}$ is imposed to control the complexity of the estimated $U_j$ in HOOI although their error bound is sharp and their model is very different. In Corollaries \ref{cor:itopup1} and \ref{cor:itopup2} below,
we prove that the second part of \eqref{condition1n} follows from the first part
respectively in a general fixed rank model and a general diverging rank model.
In fact ${R}_k^{\ideal}+R_k^{\add}\ll {R}^{(0)}$ typically so that the second part of \eqref{condition1n} provides a non-asymptotic lower bound for the $\rho$ in \eqref{bound:itopup:cora2}, allowing $\rho = \rho_{T,d_k,d^*_{-k},r_k,r_{-k},\lam_k} \to 0$.
In Corollary \ref{cor:itopup1} below, $\rho=C_{1,K}^{\iter}({R}^{\ideal} + R^{\add})$ is taken in \eqref{condition1n}
to give \eqref{bound:topup} in one iteration when ${R}_k^{\ideal}$ dominates $R_k^{\add}$.
\end{rmk}


\begin{rmk} \label{samplesize}
When the loading matrices $A_k$ and the TOPUP version of the matrix unfolding of the auto-covariance of $\cF_t$ all have bounded condition numbers and average squared entries of magnitude 1, $\lam_k^2$, $\|\Theta_{{k},0}^*\|_{\rm S}$ and $\|\Theta_{k,0}\|_{\rm op}$ are all of
the order $d\times$poly$(r_1,\ldots,r_K)$.
In this case, Theorem \ref{thm:itopup} just requires $T\ge $poly$(r_1,\ldots,r_K)$ for the initialization to achieve through iteration the fast convergence rate  $T^{-1/2}d_{-k}^{-1/2}$poly$(r_1,\ldots,r_K)$. See Corollary \ref{cor:itopup2} for details. This is in sharp contrast to
the results of traditional factor analysis
which requires $T\rightarrow \infty$ to consistently estimate the loading spaces. The main reason is that the
other tensor modes provide additional information and in certain sense serve as additional samples. Roughly speaking, we have totally $dT = d_kd_{-k}T$
observations in the tensor time series to estimate the $d_kr_k$ parameters in the projection to the column space of the loading matrix $A_k$, where $r_k\ll d_{-k}T$
in the above ``regular'' case.
%
%
\end{rmk}


Now, let us consider the statistical performance of iTIPUP procedure.
Again, by \citet{chen2022factor} the TIPUP risk in the estimation of $P_k$ is bounded by
\begin{align}\label{TIPUP_k-bd}
\overline\E\big[\big\|\widehat P_k^{\tipup} - P_k\big\|_{\rm S}\big]
\lesssim  {R}_k^{*(0)}
= (\lam_k^*)^{-2}\sigma T^{-1/2}
\sqrt{d_k} \left(\|\Theta_{k,0}^*\|_{\rm S}^{1/2} + \sigma\sqrt{d_{-k}} \right)
\end{align}
with $d_{-k}=\prod_{j\neq k}d_j$, and the aim of iTIPUP is to achieve the ideal rate
\begin{eqnarray}\label{iTIPUP-ideal_k}
{R}_k^{*\ideal}
= (\lam_k^*)^{-2} \sigma T^{-1/2}
\sqrt{d_k} \left(\|\Theta_{k,0}^*\|_{\rm S}^{1/2} + \sigma\sqrt{r_{-k}} \right)
\end{eqnarray}
through dimension reduction, where $r_{-k}=\prod_{j\neq k}r_j$. As in the case of iTOPUP, our error bound for iTIPUP involves the additional error term
\begin{equation}
{R}_k^{*\add}
=\sqrt{d^*_{-k}/d_k}{R}_k^{*\ideal}. \label{iTIPUP-add_k}
\end{equation}
The following theorem, which allows the ranks $r_k$ to grow to infinity as well as $d_k$ when $T\to\infty$,
provides sufficient conditions to guarantee the ideal convergence rate for iTIPUP.

\begin{theorem}\label{thm:itipup}
Suppose Assumption \ref{asmp:error} holds.
Let $P_k$, $\Theta_{k,0}^*$ and $\lambda^*_k$ be as in \eqref{P_k}, \eqref{Theta^*}
and \eqref{lam^*_k} respectively. Let $h_0 \le T/4$, and
\[
{R}^{*(0)}=\max_{1\le k\le K}{R}^{*(0)}_k; \mbox{\ \ \ }
{R}^{*\ideal}=\max_{1\le k\le K}{R}^{*\ideal}_k, \mbox{\ \ \ } {R}^{*\add}=\max_{1\le k\le K}{R}^{*\add}_k.
\]
with ${R}^{*(0)}_k$ in \eqref{TIPUP_k-bd}, ${R}_k^{*\ideal}$ in \eqref{iTIPUP-ideal_k} and ${R}_k^{*\add}$ in \eqref{iTIPUP-add_k}.
{Let $\widehat P_k^{(m)} =\widehat U_k^{(m)}\widehat U_k^{(m)\top}$ with the $m$-step estimator
$\widehat U_k^{(m)}$ in iTIPUP algorithm.
Then, the following statements hold for a certain numerical constant
$C_{1}^{\tipup}$ and a constant $C_{1,K}^{\iter}$ depending on $K$ only: When
\begin{equation}\label{condition1}
C_{1}^{\tipup}{R}^{*(0)}\le \min_{1\le k\le K}
\frac{(1-\rho)\lambda_k^{*2}}{8\|\Theta_{k,0}^*\|_{\rm S}}
\ \hbox{ and }\
C_{1,K}^{\iter}({R}^{*\ideal}+R^{*\add})\le\rho
\end{equation}
with a constant $0<\rho<1$, it holds simultaneously for all $1\le k\le K$ and $m\ge 0$ that}
\begin{equation}\label{bound:itipup:cora2}
\big\|\widehat P_k^{(m)} - P_k\big\|_{\rm S}
\le 2C_{1}^{\tipup}\left( (1-\rho^m)(1-\rho)^{-1}{R}^{*\ideal}
+ (\rho^m/2){R}^{*(0)}\right)
\end{equation}
in an event with probability at least
$1 -\sum_{k=1}^Ke^{-d_k}$.
In particular, after at most $J=\lfloor\log(\max_{k}d_{-k}/r_{-k})/\log(1/\rho)\rfloor$ iterations,
\begin{align}\label{bound:tipup}
\overline\E\left[\max_{1\le k\le K}\big\|\widehat P_k^{(J)} - P_k\big\|_{\rm S}\right]
\le \frac{3C_{1}^{\tipup}}{1-\rho}{R}^{*\ideal} + \sum_{k=1}^Ke^{-d_k}.
\end{align}
\end{theorem}

We briefly discuss the conditions and conclusions of Theorem \ref{thm:itipup} as the details are parallel to the remarks below Theorme \ref{thm:itopup}.
By \eqref{Theta^*}, \eqref{lam^*_k} and the Cauchy-Schwarz inequality,
$(1-h_0/T)\lam_k^{*2} \le \|\Theta_{k,0}^*\|_{\rm S}$, so that the first condition in \eqref{condition1}
guarantees a sufficiently small $R^{*(0)}$, which implies a sufficiently small error in
the initialization of iTIPUP by \eqref{TIPUP_k-bd}.
The second condition in \eqref{condition1}
again has two terms respectively reflecting the ideal rate after dimension reduction by the true $U_{-k}=\odot_{j\neq k}U_j$ in the estimation of $U_k$ and the extra cost of estimating $U_{-k}$.
The upper bound \eqref{bound:itipup:cora2} for the error of the $m$-step estimator is also comprised of two terms representing the cumulative iteration error and contracted initialization error.
In Corollary \ref{cor:itipup1} below with fixed $r_k$, the smallest  $\rho=C_{1,K}^{\iter}({R}^{*\ideal} + {R}^{*\add})$ is taken in \eqref{condition1} to achieve \eqref{bound:tipup} in one iteration when ${R}_k^{*\ideal}$ dominates $R_k^{*\add}$.
Moreover, Theorem \ref{thm:itipup} allows diverging ranks $r_k$ and convergence rate $T^{-1/2}d_{-k}^{-1/2}$poly$(r_1,\ldots,r_K)$ under proper conditions as discussed in Remark~\ref{samplesize}.

As discussed in Section~\ref{section:estimating}, we can 
mix the TOPUP and TIPUP operations for the initiation and
iterative operations in Algorithm~\ref{algorithm:itopup}.
For example, the proof of Theorems~\ref{thm:itopup} yields
the following error bound for the mixed TIPUP-iTOPUP algorithm.

\begin{theorem}\label{thm:tipup-itopup}
Assumption \ref{asmp:error} holds.
Let ${R}^{(0)}$, ${R}^{\ideal}$ and ${R}^{\add}$
be as in Theorem~\ref{thm:itopup} and $R^{*(0)}$ be as in Theorem \ref{thm:itipup}.
Let $\widehat P_k^{(m)} =\widehat U_k^{(m)}\widehat U_k^{(m)\top}$ with
$\widehat U_k^{(m)}$ being the $m$-step estimator
in the \hbox{\rm TIPUP-iTOPUP} algorithm.
Then, the following statement holds for a certain numerical constant
$C_{1}^{\topup}$ and a constant $C_{1,K}^{\iter}$ depending on $K$ only: When
\begin{equation}
\label{condition1n*}
C_{1}^{\topup}{R}^{*(0)}\le (1-\rho)/4
\ \hbox{ and }\
C_{1,K}^{\iter}({R}^{\ideal}+R^{\add})\le\rho
\end{equation}
with a constant $0<\rho<1$, it holds in an event with probability at least $1 -\sum_{k=1}^Ke^{-d_k}$
that simultaneously for all $1\le k\le K$ and $m\ge 0$
\bes
\big\|\widehat P_k^{(m)} - P_k\big\|_{\rm S}
\le 2C_{1}^{\topup}\left(
(1-\rho^m)(1-\rho)^{-1}{R}^{\ideal}
+ (\rho^m/2){R}^{*(0)}\right).
\ees
\end{theorem}

We omit the statement of an analogous error bound for the \hbox{\rm TOPUP-iTIPUP} algorithm.

%


\subsection{Fixed rank factor process}

In this section we provide the convergence rate when the dimensions of the factors $\cF_t$, or equivalently the ranks of the signal process $\cM_t$, $r_1,\ldots,r_K$, are fixed, and the auto-cross-outer-product of the factor process is ergodic. Formally, we impose the following additional assumption.

\begin{assumption}\label{asmp:factor}
The ranks $r_1,...,r_K$ are fixed.
The factor process $\cF_t$ is weakly stationary
and its auto-cross-outer-product process is ergodic in the sense of
\begin{equation*}
\frac{1}{T-h}\sum_{t=h+1}^T \cF_{t-h}\otimes\cF_t \longrightarrow \E (\cF_{t-h}\otimes\cF_t) \quad \text{in probability},
\end{equation*}
where the elements of $\E (\cF_{t-h}\otimes\cF_t)$ are all finite. In addition, the condition numbers of $A_k^\top A_k$ ($k=1,...,K$) are bounded. Furthermore, assume that $h_0$ is fixed, and  \\  
(i) (TOPUP related):  $\E[\text{mat}_1(\Phi_{k,1:h_0})]$ is of rank $r_k$ for $1\le k\le K$.  \\ 
(ii) (TIPUP related): $\E[\Phi_{k,1:h_0}^{*\cano}]$ is of rank $r_k$ for $1\le k\le K$. 
\end{assumption}

Under Assumption \ref{asmp:factor}, the factor process has a fixed expected auto-cross-moment tensor with fixed dimensions. The assumption that the condition numbers of $A_k^\top A_k$ ($k=1,...,K$) are bounded corresponds to the pervasive condition (e.g., \citet{stock2002}, \citet{bai2003}). It ensures that all the singular values of $A_k$ are of the same order. Such conditions are commonly imposed in factor analysis.

As our methods are based on auto-cross-moment at nonzero lags, we do not need to assume any specific model for the latent process $\cF_t$, except some rank conditions in Assumption~\ref{asmp:factor}(i) and (ii).
Since the columns of  $\Phi_{k,1:h_0}^{*\cano}$ are linear combinations of those of $\text{mat}_1(\Phi_{k,1:h_0}^{\cano})$ and $\E[\text{mat}_1(\Phi_{k,1:h_0})]$
and $\E[\text{mat}_1(\Phi^{\cano}_{k,1:h_0})]$ have the same rank,
Assumption \ref{asmp:factor}(ii) implies Assumption \ref{asmp:factor}(i).

In order to provide a more concrete understanding of Assumption~\ref{asmp:factor}(i) and (ii),
consider the case of $k=1$ and $K=2$. We write the factor process
$\cF_t = (f_{i,j,t})_{r_1\times r_2}$, and the stationary auto-cross-moments
$\phi_{i_1,j_1,i_2,j_2,h} = \E (f_{i_1,j_1,t-h}f_{i_2,j_2,t})$. Hence
$\E[\text{mat}_1(\Phi_{k,1:h_0})]$ is a $r_k\times(r_{-k}r_kr_{-k}h_0)$ matrix, with columns being
$\phi_{\cdot, j_1,i_2,j_2,h}$. Since
$\E[\text{mat}_1(\Phi_{k,1:h_0})]\E[\text{mat}_1(\Phi_{k,1:h_0})]^{\top}$ is a sum of many semi-positive definite
$r_k\times r_k$ matrices, if any one of these matrices is full rank, then $\E[\text{mat}_1(\Phi_{k,1:h_0})]$ is of rank $r_k$. Hence Assumption~\ref{asmp:factor}(i) is relatively easy to fulfill. On the other hand, Assumption~\ref{asmp:factor}(ii) is quite different. First, the condition is imposed on
the canonical form of the model as the inner product in TIPUP related procedures behaves
differently. Let
$\cF_t^{\cano} = U_1^\top \cM_t U_2 = (f_{i,j,t}^{\cano})_{r_1\times r_2}$, and
$\phi^{\cano}_{i_1,j_1,i_2,j_2,h} = \E (f_{i_1,j_1,t-h}^{\cano}f_{i_2,j_2,t}^{\cano}$). Then
$\|\Phi^{*{\cano}}_{1,1:h_0}\|_{\rm HS}^2
=\sum_{h=1}^{h_0}\sum_{i_1,i_2}\big(\sum_{j=1}^{r_2} \phi_{i_1,j,i_2,j,h}^{{\cano}}\big)^2$.
As $\phi_{i_1,j,i_2,j,h}^{{\cano}}$ may be positive or negative for different $i_1,i_2,j,h$, the summation
$\sum_{j=1}^{r_2} \phi_{i_1,j,i_2,j,h}^{{\cano}}$ is subject to potential signal cancellation for $h>0$. Assumption~\ref{asmp:factor}(ii)
ensures that
there is no complete signal cancellation that makes the rank of  $\E[\Phi_{k,1:h_0}^{*{\cano}}]$ less than $r_k$. While the signal
cancellation rarely causes the rank deficiency, the resulting loss of efficiency may still have an impact on the finite sample performance as our simulation results demonstrate. Of course complete signal cancellation is less likely with larger $h_0$.

The following corollary is a simplified version of Theorem \ref{thm:itopup} under Assumption~\ref{asmp:factor}(i). 

\begin{cor}\label{cor:itopup1}
Suppose Assumptions \ref{asmp:error} and \ref{asmp:factor}(i) hold.
Let $\lambda=\prod_{k=1}^K \| A_k \|_{\rm S}$ and $r_{-k}=r/r_k$. Let ${h_0}\le T/4$ and
$\sigma$ fixed. Then, there exist numerical constants $C_{0,K}$ and  $C_{1,K}$ depending on $K$ only such that when
\begin{align} \label{eq:signal:itopup}
\lambda^{2} \ge C_{0,K} \sigma^2\max_{1\le k\le K}\left(\frac{d r_{-k}}{T}+\frac{d}{\sqrt{Td_{k}r_{-k}}} \right),
\end{align}
the 1-step iTOPUP estimator satisfies
\begin{align}\label{bound:itopup:cora}
\E \|\widehat P_{k}^{(1)}-P_k\|_{\rm S}
\le&  C_{1,K} \left(\frac{\sigma}{\lambda\sqrt{T}}\left( \frac{\sqrt{d_k}} {\sqrt{r_{-k}}}+\sqrt{rr_{-k}}\right)+ \frac{\sigma^2}{\lambda^2\sqrt{T}}\left( \frac{\sqrt{d_k}} {\sqrt{r_{-k}}}+\sqrt{r}\right) \right)  \\
&+C_{1,K} \left(\frac{\sigma\sqrt{d_k} r_{-k}}{\lambda\sqrt{T}}+ \frac{\sigma^2\sqrt{d_kr_{-k}}}{\lambda^2\sqrt{T}} \right)^2 + \sum_{k=1}^K e^{-d_k}. \notag
\end{align}
\end{cor}

\begin{rmk}
Under Assumption \ref{asmp:factor} that $r_k$ is fixed, \eqref{eq:signal:itopup}, \eqref{bound:itopup:cora} and \eqref{eq:signal:itipup}, \eqref{bound:itipup:cora} in Corollary \ref{cor:itipup1} can absorb $r_k$'s into the numerical constants. The corollaries are expressed in this form to allow divergent $r_k$ for the purpose of facilitating comparison with the minimax lower bound in Theorem \ref{thm:stat_lowerbdd}. They also represent a specific case of Corollaries \ref{cor:itopup2}-\ref{cor:tipup-itopup2}.
\end{rmk}

Corollary \ref{cor:itopup1} asserts that, in order to recover the factor loading space for $A_k$, the signal to noise ratio needs to satisfy $\lambda/\sigma \ge C_{0,K,r} \max_{k\le K}(d^{1/2}T^{-1/2}+d^{1/2}d_{k}^{-1/4}T^{-1/4})$ as in \eqref{eq:signal:itopup},
and the ideal rate \eqref{bound:itopup:cora} can be achieved in one iteration. Under Assumptions \ref{asmp:error}, \ref{asmp:factor}(i), the error bound in Proposition \ref{prop:topup} yields the convergence rate
\begin{align*}
\E \|\widehat P_{k}^{(0)}-P_k\|_{\rm S}\lesssim \left( \frac{\sigma \sqrt{d_k}} {\lambda \sqrt{T }}+ \frac{\sigma^2\sqrt{d_k}} {\lambda^2 \sqrt{T}} \right) + \left( \frac{\sigma \sqrt{d}} {\lambda \sqrt{T }}+ \frac{\sigma^2\sqrt{d_k}d_{-k}} {\lambda^2 \sqrt{T}} \right)^2 .
\end{align*}
In comparison, the ideal rate is much sharper than the convergence rate of the non-iterative TOPUP in Proposition \ref{prop:topup} when $\lambda^2/\sigma^2\ll \min_{k\le K}\{d^{4/3}/(T^{1/3}d_{k}),d^2/(T^{1/2}d_{k}^{3/2})\}$.




The following corollary is a simplified version of Theorem \ref{thm:itipup} under Assumption~\ref{asmp:factor}(ii), which excludes severe signal cancellation in iTIPUP.

\begin{cor} \label{cor:itipup1}
Suppose Assumptions \ref{asmp:error} and \ref{asmp:factor}(ii) hold.
Let $\lambda=\prod_{k=1}^K \| A_k \|_{\rm S}$ 
and $r_{-k}=r/r_k$. Let ${h_0}\le T/4$ and $\sigma$ fixed.
Then, there exist constants $C_{0,K}$ and $C_{1,K}$ depending on $K$ only such that when
\begin{align} \label{eq:signal:itipup}
\lambda^{2} \ge C_{0,K}\sigma^2\max_{1\le k\le K}\left(\frac{d_{k}}{Tr_{-k}} +\frac{\sqrt{d}}{\sqrt{T}r_{-k}}\right),
\end{align}
the 1-step iTIPUP estimator satisfies
\begin{align}\label{bound:itipup:cora}
\E  \|\widehat P_{k}^{(1)}-P_k\|_{\rm S}
\le  C_{1,K}\left( \frac{\sigma \sqrt{d_k}} {\lambda \sqrt{T r_{-k}}}+ \frac{\sigma^2\sqrt{d_k}} {\lambda^2 \sqrt{T r_{-k}}} \right) + \sum_{k=1}^K e^{-d_k},
\end{align}
and the 1-step TIPUP-iTOPUP estimator satisfies \eqref{bound:itopup:cora}.
\end{cor}

Compared with the results in Corollary \ref{cor:itopup1} for iTOPUP, the achieved ideal rate  \eqref{bound:itipup:cora} is the same. However, the signal-to-noise ratio requirement \eqref{eq:signal:itipup} is weaker but Assumption \ref{asmp:factor}(ii) is stronger in Corollary \ref{cor:itipup1} for iTIPUP.
Again, the ideal rate is much sharper than the convergence rate of the non-iterative TIPUP in  \citet{chen2022factor}.



\subsection{Diverging ranks}

The main theorems in Subsection 3.2 allow for the case where the dimensions of the core factor, $r_1,...,r_K$, diverge as the dimensions of the observed tensor $d_1,...,d_K$ grow to infinity.
The following assumption provides a concrete set of conditions that can be used to provide some insights of the properties of iTOPUP and iTIPUP in such scenarios.

\begin{assumption}\label{asmp:strength}
For a certain $\delta_0\in [0,1]$,
$\|\Theta_{k,0}\|_{\text{op}} \asymp \sigma^2 d^{1-\delta_0}/r$ 
 and $\|\Theta_{k,0}^*\|_{\rm S} \asymp \sigma^2 d^{1-\delta_0}/r_k$ 
with probability approaching one.
For the singular values, two scenarios are considered. \\
(i) (TOPUP related): There exist some constants $\delta_1\in [\delta_0, 1]$ and $c_1>0$ such that with probability approaching one (as $T\rightarrow \infty$) $\lambda_k^2\ge c_1 \sigma^2 d^{1-\delta_1}/\sqrt{rr_k}$, for all $k=1,...,K$. \\
(ii) (TIPUP related):
There exist some constants ${\delta_1}\in [\delta_0, 1]$, $c_2>0$ and ${\delta_2}\ge 0$ such that with probability approaching one (as $T\rightarrow \infty$), $\lam_k^{*2} \ge c_2 {\sigma^2} d^{1-\delta_1}r_k^{-1} r_{-k}^{-\delta_2}$ for all $k=1,...,K$.
\end{assumption}

Assumption \ref{asmp:strength} is similar to the signal strength condition of \citet{lam2012}, and the pervasive condition on the factor loadings (e.g., \citet{stock2002} and \citet{bai2003}). 
It is more general than Assumption \ref{asmp:factor} in the sense that it allows $r_1,...,r_K$ to diverge and the latent process $\cF_t$ does not have to be weakly stationary.

We take $\delta_0,\delta_1$ as measures of the strength of the signal process $\cM_t$. They roughly indicate how much information is contained in the signals compared with the amount of noise, with respect to the dimensions and ranks, $d,r$ and $r_k$. In this sense, they reflect the signal to noise ratio.
When $\delta_0=\delta_1=0$, the factors are called strong factors; otherwise, the factors are called weak factors.

\begin{rmk}[Signal Strength and the index $\delta_0$] \label{rmk:delta0}
We note that $\text{trace}(\Theta_{k,0}) = \text{trace}(\Theta_{k,0}^*) = \sum_{t=1}^T\|\text{vec}(\cM_t)\|_2^2/T$, and that $\text{rank}(\Theta_{k,0})=r$ and $\text{rank}(\Theta_{k,0}^*)=r_k$ when the data is in general position, where $\Theta_{k,0}$ is treated as a $d\times d$ matrix.
Thus, if $\sum_{t=1}^T\|\text{vec}(\cM_t)\|_2^2/(\sigma^2 d\, T) \asymp d^{-\delta_0}$ is the signal-to-noise ratio, then the condition $\|\Theta_{k,0}\|_{\text{op}} \asymp \sigma^2 d^{1-\delta_0}/r$ holds when $r$ is the order of the effective rank of $\Theta_{k,0}$ and the condition $\|\Theta_{k,0}^*\|_{\text{S}} \asymp \sigma^2 d^{1-\delta_0}/r_k$ holds when $r_k$ is the order of the effective rank of $\Theta_{k,0}^*$.
Because the signal $\cM_t$ has $d$ elements at each $t$, the assumption $\sum_{t=1}^T\|\text{vec}(\cM_t)\|_2^2/(\sigma^2 dT) \asymp d^{-\delta_0}$ says that the squared ratio of the elements and the noise level is $d^{-\delta_0}$ averaged over time and space.  Thus, the factor is called strong when $\delta_0=0$. In view of \eqref{eq:tensorfactor} and \eqref{eq:tucker-model}, $\cM_t=\cF_t\times_{k=1}^K A_k$, so that we may have weaker factor with $\delta_0>0$ when the loading matrices $A_k$ are sparse or have some relatively small singular components.
We note that by Cauchy-Schwarz, the signal-to-noise ratio conditions also imply $(1-h/T)^2\|\Theta_{k,h}\|_{\rm HS}^2\le\|\Theta_{k,0}\|_{\rm HS}^2\lesssim r(\sigma^2d^{1-\delta_0}/r)^2$ and $(1-h/T)^2\|\Theta_{k,h}^*\|_{\rm HS}^2\le\|\Theta_{k,0}^*\|_{\rm HS}^2\lesssim r_k(\sigma^2d^{1-\delta_0}/r_k)^2$ respectively.
\end{rmk}

\begin{rmk}[Assumption~\ref{asmp:strength}(i) and the role of $\delta_1$] \label{rmk:asmp:3i}
In fact, for TOPUP, Assumption \ref{asmp:strength}(i) holds when
(a) $\|{{\overline \E}}{[\text{TOPUP}_k]}\|_{\rm HS}^2 = \sum_{h=1}^{h_0} \|\Theta_{k,h}\|_{\rm HS}^2
\asymp h_0 \sigma^4d^{2(1-\delta_1)}/r$ and (b) all the nonzero singular values of
${\overline \E}[{\text{TOPUP}_k}]$ are of the same order.
Because $\|\Theta_{k,h}\|_{\rm HS}^2\lesssim \sigma^4d^{2(1-\delta_0)}/r$ by the condition on the signal-to-noise ratio, we must have $\delta_1\ge \delta_0$, and $d^{\delta_0-\delta_1}$ can be viewed as the order of average auto-correlation over lags $h=1,\ldots,h_0$. For
$k=1$ and $K=2$, the factor process in the canonical form is $\cF_t^{\cano} = U_1^\top \cM_t U_2 = (f_{i,j,t}^{\cano})_{r_1\times r_2}$, and $\phi^{\cano}_{i_1,j_1,i_2,j_2,h} = \sum_{t=h+1}^T f_{i_1,j_1,t-h}^{\cano} f_{i_2,j_2,t}^{\cano} /(T-h)$ is the time average cross product between the factor fibers $f_{i_1,j_1,1:T}^{\cano}$ and $f_{i_2,j_2,1:T}^{\cano}$. Thus, the first condition (a) means $\sum_{h=1}^{h_0}\|\Theta_{1,h}\|_{\rm HS}^2 = \sum_{h=1}^{h_0}\|\Phi^{\cano}_{1,h}\|_{\rm HS}^2 = \sum_{i_1,j_1,i_2,j_2,h}\big(\phi^{\cano}_{i_1,j_1,i_2,j_2,h}\big)^2\asymp h_0 \sigma^4d^{2(1-\delta_1)}/r$.
\end{rmk} 

\begin{rmk}[Assumption~\ref{asmp:strength}(ii), the role of $\delta_2$ and signal cancellation]
The points parallel to those in Remark~\ref{rmk:asmp:3i} are applicable to TIPUP, but with one caveat: Beyond the average auto-correlation, an additional discount $r_{-k}^{-\delta_2}\le 1$ is needed to take into account the impact of possible signal cancellation with TIPUP and its iteration. For $k=1$ and $K=2$,
$\|\Theta_{1,h}^*\|_{\rm HS}^2
= \|\Phi^{*\cano}_{1,h}\|_{\rm HS}^2
=\sum_{i_1,i_2}\big(\sum_{j=1}^{r_2} \phi^{\cano}_{i_1,j,i_2,j,h}\big)^2$,
and the summation inside the square is subject to signal cancellation for $h>0$  since the auto-cross-moment $\phi^{\cano}_{i_1,j,i_2,j,h}$ can have different signs.
The additional parameter $\delta_2$ measures the severity of signal cancellation in the TIPUP related procedures. For example, when the majority of $\phi^{\cano}_{i_1,j,i_2,j,h}$ are of the same sign for most of $(i_1,i_2,h)$,
it would be reasonable to assume $\delta_2=0$. When $\phi^{\cano}_{i_1,j,i_2,j,h}$ behave like independent mean zero variables, {$\delta_2$} would be close to $0.5$. And $\delta_2=\infty$ when all the signals cancel out by the summation $\phi^{\cano}_{i_1,j,i_2,j,h}$ over $j$. In the case of fixed $r_k$, the convergence rate depends on whether $\delta_2=\infty$ (severe signal cancellation) or not.
\end{rmk}

\begin{rmk}[The role of $h_0$] \label{rmk:h0}
The selection of $h_0$ is a relative minor problem in practice though very complex to analyze. Theoretically it suffices to use an $h_0$ with $\lambda_k$ of the right order,
so that choosing a somewhat large $h_0$ would not harm the convergence rate for the proposed methods.
In practice a small $h_0$ (less than 3) is often sufficient. The impact of the choice of $h_0$ on the signal and noise depends on the autocorrelation of the factor process, as well as the loading matrices. For example, if the factor process is of very short memory (e.g. an MA(1) process), including any lag $h>1$ only introduces
noise to TOPUP$_k$ in \eqref{eq:topup:def}
and TIPUP$_k$ in \eqref{eq:tipup:def} without enhancing the signal. On the other hand, including an extra lag is the most simple and effective way to prevent signal cancellation with iTIPUP, as discussed in the previous
remark. Increasing $h_0$ includes more non-negative terms in the
signal strength
$\sum_{i_1,i_2,h}\big(\sum_{j=1}^{r_2}\phi^{\cano}_{i_1,j,i_2,j,h}\big)^2$, hence potentially reducing the chance of severe signal cancellation. The simulation results presented in the supplementary material provide some empirical behavior of choosing different $h_0$. While the choice of $h_0$ will affect the assumptions, {in practice} we may compare the patterns of estimated singular values under different lag values $h_0$ in iTOPUP and iTIPUP to evaluate the benefit of taking a larger $h_0$.
See also the simulation study.
\end{rmk}


We describe below the convergence rate of iTOPUP in terms of $d_k$, $r_k$ and $T$ under Assumption~\ref{asmp:strength}(i) when the dimensions of the core factor $r_1,...,r_K$ are allowed to diverge.




\begin{cor} \label{cor:itopup2}
Suppose Assumptions \ref{asmp:error} and \ref{asmp:strength}(i) hold. Let ${h_0}\le T/4$, $d_{-k}^*=\sum_{j\neq k} d_jr_j$ and $r=\Pi_{k=1}^K r_k$. 
Suppose that for a sufficiently large $C_0$ not depending on $\{\sigma, d_k,r_k, k\le K\}$,
\begin{align} \label{eq:sample:itopup}
T \ge C_0 \max_{1\le k\le K}\left( d^{2\delta_1-\delta_0}r_kr_{-k}^2 + d^{2\delta_1}r_k^2r_{-k}/d_k \right).
\end{align}
Then, after $J=O(\log d)$ 
iterations, we have the following upper bounds for iTOPUP,
\bel{bound:itopup:corb}
&& \max_{1\le k\le K}
\|\widehat P_{k}^{(J)}-P_k\|_{\rm S}
\\ \notag &=&O_{\P}(1)
\max_{1\le k\le K}\Bigg( \frac{d_k^{1/2}
r_k^{1/2}(1+r^{1/2}/d^{(1-\delta_0)/2})+r^{3/2}r_k^{-1/2}(1+r_k^{1/2}/d^{(1-\delta_0)/2})}
{T^{1/2}d^{1/2+\delta_0/2-\delta_1}} \\
\cr &&\qquad\qquad\qquad+ \left(
\frac{d_k^{1/2}r^{3/2}
(1+r_k^{1/2}/d^{(1-\delta_0)/2})}
{T^{1/2}d^{1/2+\delta_0/2-\delta_1} r_k}
\right)^2 \Bigg). \notag
\eel
Moreover, \eqref{bound:itopup:corb} holds after at most $J=O(\log r)$ iterations, if any one of the following three conditions holds in addition to \eqref{eq:sample:itopup}: (i) $d_k$ ($k=1,...,K$) are of the same order, (ii) $\lam_k$ ($k=1,...,K$) are of the same order, (iii) $(\lam_k)^{-2}\sqrt{d_k}$ ($k=1,...,K$) are of the same order.
\end{cor}

Note that the second part of Corollary~\ref{cor:itopup2} says that when the condition is right, iTOPUP algorithm only
needs a small number of iterations to converge, as $O(\log r)$ is typically very small. The noise level $\sigma$ does not appear directly in the rate since it is incorporated in the
signal to noise ratio in the tensor form in Assumption \ref{asmp:strength}.
In Corollary \ref{cor:itopup2}, we show that as long as the sample size $T$ satisfies \eqref{eq:sample:itopup}, the iTOPUP achieves consistent estimation under proper regularity conditions. To digest the condition, we notice that \eqref{eq:sample:itopup} becomes $T\ge C_0 \max_k (r_k r_{-k}^2) $ when the growth rate of $r_k$ is much slower than $d_k$ and the factors are strong with $\delta_0=\delta_1=0$.


The advantage of using index $\delta_0,\delta_1$ is to link the convergence rates of the estimated factor loading space explicitly to the strength of factors. It is clear that the stronger the factors are, the faster the convergence rate is. Equivalently, the stronger the factors are, the smaller the sample size is required.


When the ranks $r_k$ ($k=1,...,K$) also {diverge} and there is no severe signal cancellation in iTIPUP, we have the following convergence rate for iTIPUP under Assumption \ref{asmp:strength}(ii).


\begin{cor} \label{cor:itipup2}
Suppose Assumptions \ref{asmp:error} and \ref{asmp:strength}(ii) hold. Let ${h_0}\le T/4$ and $d_{-k}^*=\sum_{j\neq k} d_jr_j$.
Suppose that for a sufficiently large $C_0$ not depending on $\{\sigma, d_k,r_k, k\le K\}$,
\begin{align}\label{eq:sample:itipup}
T\ge C_0 \max_{1\le k\le K}\left( \frac{(d_kr_k+d^{\delta_0}r_k^2)r_{-k}^{2\delta_2}r^{2\delta_2}}
{d^{1+3\delta_0-4\delta_1}\min_{1\le k\le K} r_k^{2\delta_2}}
+ \frac{d^*_{-k}r_kr_{-k}^{2\delta_2}} {d^{1+\delta_0-2\delta_1}}
 \left(1+\frac{r}{d^{1-\delta_0}}\right)
\right).
\end{align}
Then, after at most $J=O(\log d)$ iterations, 
the iTIPUP estimator satisfies
\begin{align}\label{bound:itipup:corb}
\max_{1\le k\le K}\|\widehat P_{k}^{(J)}-P_k\|_{\rm S}
= O_{\P}(1) \max_{1\le k\le K}
\left(\frac{d_k^{1/2}
r_k^{1/2}r_{-k}^{\delta_2}(1+r^{1/2}/d^{(1-\delta_0)/2})}
{T^{1/2}d^{1/2+\delta_0/2-\delta_1}}\right).
\end{align}
Moreover, \eqref{bound:itipup:corb} holds after at most $J=O(\log r)$ iterations, if any one of the following three conditions holds in addition to condition \eqref{eq:sample:itipup}, (i) $d_k$ ($k=1,...,K$) are of the same order, (ii) $\lam_k^{*}$ ($k=1,...,K$) are of the same order, (iii) $(\lam_k^{*})^{-2}\sqrt{d_k}$ ($k=1,...,K$) are of the same order.
\end{cor}

When the average auto-correlation is of unit order and the signal cancellation for TIPUP has no impact on the order of the signal ($\delta_0=\delta_1$ and $\delta_2=0$ respectively),
Corollary \ref{cor:itipup2} requires the sampling rate
$T\gtrsim h_0+(d_kr_k+d^{\delta_0}r_k^2+d_{-k}^*r_k(1+r/d^{1-\delta_0}))/d^{1-\delta_0}$
and provides the convergence rate $(r_kd_k)^{1/2}(1+r/d^{1-\delta_0})^{1/2}/(Td^{1-\delta_0})^{1/2}$.
For examples, $T\ge 4h_0+C_1$ gives the rate $(r_kd_k)^{1/2}/(Td^{1-\delta_0})^{1/2}$
when $\delta_0\le (K-2)/(2K)$ and $r_k^2\lesssim d_k\asymp d^{1/K}\,\forall k$, and the sample size requirement can be written as
$T\gtrsim h_0+d^{\delta_0}r_k^2/d^{1-\delta_0}$ when $r_k^2\asymp r^{2/K}\lesssim d_k\asymp d^{1/K}\,\forall k$ regardless of $\delta_0\in [1/K,1]$. Thus, the side condition involving $R^{*\add}$
in the second part of \eqref{condition1}
is absorbed into the other components of \eqref{condition1}.


Corollary \ref{cor:itopup2} and Corollary \ref{cor:itipup2} offer comparison of the iTOPUP and iTIPUP when the ranks diverge from two perspectives: sample size requirements and convergence rates. The lower bounds on $T$ in \eqref{eq:sample:itopup} in Corollary \ref{cor:itopup2} and
\eqref{eq:sample:itipup} in Corollary \ref{cor:itipup2} provide the sample complexity of the iTOPUP and iTIPUP respectively.
In the case that the growth rate of $r_k$ is much slower than $d_k$ and the factors are strong with $\delta_0=\delta_1=0$, the required sample size of the iTIPUP reduces to $T\ge 4h_0+C_0\max_{j,k} \left( r_kr_{-k}^{2\delta_2} r_{-j}^{2\delta_2}/d_{-k}+ r_kr_{-k}^{2\delta_2} r_j/d_{-j}\right)$,  
where $r_{-k}=r/r_k$ and $d_{-k}=d/d_k$. By comparing with the comment after Corollary \ref{cor:itopup2}, where the sample size requirement for the iTOPUP is
$T\ge C_0\max_k(r_kr_{-k}^2)$ when $\delta_0=\delta_1=0$, it can be seen that the sample complexity for the iTIPUP is smaller, if $\delta_2$ is a small constant.
From the perspective of convergence rate, let us compare
\eqref{bound:itopup:corb} in Corollary \ref{cor:itopup2} and
\eqref{bound:itipup:corb} in Corollary \ref{cor:itipup2}. When ranks diverge, iTIPUP is slower than iTOPUP if $\delta_2>3/2$, or $\{0\le\delta\le 3/2, d_k\gtrsim rr_{-k}^{2-2\delta_2}, d_{-k}\gtrsim r r_{-k}^{3-2\delta_2}\}$,
and faster if $d_k\lesssim r_kr_{-k}^{2-2\delta_2}$, no matter how strong the factor is or what values $\delta_0,\delta_1$ take.
As expected, the convergence rate is slower in the presence of weak factors. See the simulation for more empirical evidence.

Similar to Corollaries \ref{cor:itopup2} and \ref{cor:itipup2}, we have the following rate for TIPUP-iTOPUP.

\begin{cor} \label{cor:tipup-itopup2}
Suppose Assumptions \ref{asmp:error} and \ref{asmp:strength} hold.
Let ${h_0}\le T/4$ and $d_{-k}^*=\sum_{j\neq k} d_jr_j$.
Suppose that for a sufficiently large $C_0$ not depending on $\{\sigma, d_k,r_k, k\le K\}$,
\begin{align} \label{eq:sample:tipup-itopup}
T \ge C_0
\max_{1\le k\le K}\left(d^{2\delta_1-\delta_0}r_k\bigg(\frac{r_{-k}^{2\delta_2}}{d_{-k}} + \frac{r_{-k}^3}{d_{-k}}\bigg)
+ \frac{d^{2\delta_1}r_k^2}{d_k}\bigg(\frac{r_{-k}^{2\delta_2}}{d_{-k}}
+ \frac{r_{-k}^3}{d_{-k}^2} \bigg) + \frac{d_{-k}^*\sqrt{rr_k}}{d^{1-\delta_1}}
\right).
\end{align}
Then, after at most $J=O(\log d)$ iterations, 
the TIPUP-iTOPUP estimator satisfies \eqref{bound:itopup:corb}.
Moreover, the above error bound holds after at most $J=O(\log r)$ iterations, if any one of the following three conditions holds in addition to condition \eqref{eq:sample:tipup-itopup}, (i) $d_k$ ($k=1,...,K$) are of the same order, (ii) $\lam_k$ ($k=1,...,K$) are of the same order, (iii) $(\lam_k)^{-2}\sqrt{d_k}$ ($k=1,...,K$) are of the same order.
\end{cor}

Compared with Corollary \ref{cor:itopup2}, Corollary  \ref{cor:tipup-itopup2} provides the same error bound for smaller $T$ (possibly with bounded $T\gtrsim h_0$) when $r_{-k}^{2\delta_2} \lesssim r_{-k}d_{-k}$. The side condition involving $R^{\add}$ in the second part of \eqref{condition1n*},
corresponding to the last component of \eqref{eq:sample:tipup-itopup} involving $d_{-k}^*$, is absorbed into the other components of \eqref{condition1n*}
when $r_k^{1/2}\le d^{\delta_1-\delta_0}\big(r_{-k}^{2\delta_2-1} + r_{-k}^2\big)\, \forall k\le K$.



\subsection{Comparisons}


\subsubsection{Comparison between the non-iterative procedures and iterative procedures}
Theorems \ref{thm:itopup} and \ref{thm:itipup}
show that the convergence rates of the non-iterative estimators TOPUP and TIPUP can be improved by their iterative counterparts.
Particularly, when the dimensions $r_k$ for the factor process are fixed and the respective signal strength
conditions are fulfilled, the proposed iTOPUP and iTIPUP just need one-iteration to achieve the much sharper ideal
rate $R^{\ideal}$ in \eqref{iTOPUP-ideal_k} and $R^{*\ideal}$ in \eqref{iTIPUP-ideal_k}, compared with the rate \eqref{bound:topup0} of TOPUP and \eqref{TIPUP_k-bd} of TIPUP derived in \citet{chen2022factor}, respectively. The improvement is achieved through replacing the much larger $d_{-k}$ by $r_{-k}$, via orthogonal projection.
When the factors are strong with $\delta_0=\delta_1=0$ and the factor dimensions are fixed,
the non-iterative TOPUP-based
estimators of \citet{lam2011} for the vector factor model, \citet{wang2019} for the matrix factor and \citet{chen2022factor} for tensor factor models all have the same $O_{\P}(T^{-1/2})$ convergence rate for
estimating the loading space. In comparison, the convergence rate $O_{\P}(T^{-1/2}d_{-k}^{-1/2})$ of both
iterative estimators, iTOPUP and iTIPUP (when there is no severe signal cancellation, with bounded $\delta_2$), is much sharper. Intuitively, when the signal is strong, the orthogonal projection
operation helps to consolidate signals while potentially averaging out the noises, when the projection
reduces {the dimension of} the mode-$k$ unfolded matrix from {$d_k\times d_{-k}$}
for the tensor $\cX_t$ {to $d_k\times r_{-k}$}
for the projected tensor $\cZ_t$,
resulting in the improvement by a factor of $d_{-k}^{-1/2}$
in the convergence rate.

When $r_k$ are allowed to diverge, the
iTOTUP and iTIPUP algorithms converge after at most $O(\log(d))$ iterations
to achieve the ideal rate according to Theorems \ref{thm:itopup} and \ref{thm:itipup}.
The number of iterations needed can be as few as $O(\log(r))$ when the condition is right.

\subsubsection{Comparison between iTIPUP and iTOPUP}
The inner product operation in \eqref{eq:tipup:def} for TIPUP-related procedures enjoys significant amount of noise cancellation comparing to the outer product operation in \eqref{eq:topup:def} for TOPUP-related procedures.
Compared with iTOPUP, the benefit of noise cancellation of the iTIPUP procedure is still visible through the reduction of $r_{-k}$ in \eqref{iTOPUP-ideal_k} to $\sqrt{r_{-k}}$ in \eqref{iTIPUP-ideal_k} in the ideal rates.
However, this post-iteration benefit is much less pronounced compared with the reduction of $d_{-k}$ in \eqref{TOPUP_k-bd} for TOPUP
to $\sqrt{d_{-k}}$ in \eqref{TIPUP_k-bd} for TIPUP in the non-iterative rates.
Meanwhile, the potential {for} signal cancellation in the TIPUP related schemes persists as $\lam^*_k$ and $\lam_k$ are unchanged between the initial and ideal rates. We note that the signal strength can be viewed as $\lam_k$ and $\lam^*_k$ in Theorems \ref{thm:itopup} and \ref{thm:itipup} respectively for TOPUP/iTOPUP and TIPUP/iTIPUP, and that severe signal cancellation can be expressed as $\lam_k^*\ll\lam_k$. When $r_{-k}$ are allowed to diverge to infinity, the impact of signal cancellation is expressed in terms of $\delta_2$ in Assumption \ref{asmp:strength}:
The iTOPUP has a faster rate than the iTIPUP when $\delta_2>3/2$, or $\{0\le\delta\le 3/2, d_k\gtrsim rr_{-k}^{2-2\delta_2}, d_{-k}\gtrsim r r_{-k}^{3-2\delta_2}\}$, and slower rate when $d_k\lesssim r_kr_{-k}^{2-2\delta_2}$, in view of Corollary \ref{cor:itopup2} and \ref{cor:itipup2}. In Corollaries \ref{cor:itopup1} and \ref{cor:itipup1}, iTOPUP and iTIPUP have the same convergence rate because Corollary \ref{cor:itipup1} assumes that signal cancellation does not change convergence rate.


Our results seem to suggest that the mixed TIPUP-iTOPUP procedure would strike a good balance between the benefit of noise cancellation (e.g. smaller $T$ for consistency) and the potential danger of signal cancellation (e.g. $\lam^*_k\ll\lam_k$) for the following four reasons:
(1) The benefit of noise cancellation is much larger in the initialization, in term of $d_{-k}$, in view of the rates $R_k^{(0)}$ in \eqref{TOPUP_k-bd} and $R^{*(0)}$ in \eqref{TIPUP_k-bd}.
(2) The first part of condition  \eqref{condition1n*} for TIPUP-iTOPUP is weaker than the first part of condition \eqref{condition1} for TIPUP-iTIPUP.
(3) The signal strength $\lam_k$ of the stronger TOPUP form is retained in the rate $R^{\ideal}$ after iTOPUP iteration. (4) As we will prove in Section~\ref{section:computation}, the sample size requirement for the TIPUP initialization is optimal in the sense that it matches a computational lower bound under suitable conditions.
Our simulation results support this recommendation, especially for relatively small $r_{-k}$.
Of course if the sample size qualitatively justifies the condition $C_{1}^{\topup}{R}^{(0)}\le (1-\rho)/4$ in \eqref{condition1n} and/or if a possible signal cancellation is a significant concern, the TOPUP initiation should be used.

\subsubsection{Comparison with HOOI}
The signal to noise ratio (SNR) condition, or equivalently the sample size requirement, is mainly used to ensure that the initial estimator has sufficiently small estimation error.
Thus, the performance of iterative procedures is measured by both the SNR requirement
and the error rate achieved.
Consider fixed $h_0$ in the fixed rank case with $K=3$ and $d_{\max}\asymp d^{1/K}$. In the fixed signal model where $\cM_t =\cM$ is fixed and deterministic in \eqref{eq:tensorfactor}, applying HOOI to the average of $\cX_t$ would require SNR $\lambda(T^{1/2}/\sigma)\ge C_0 d^{1/4}$ to achieve the loss of the order $(\sigma/T^{1/2}) d_k^{1/2}/\lam$ according to \citet{zhang2018tensor}, where $\sigma/T^{1/2}$ is viewed as the noise level for HOOI as it is the standard deviation of each element of the average tensor. In terms of the auto-crossproducts, taking the average over $\cX_t$ roughly amounts to taking the average of all $T(T-1)/2$ lagged products between $\cX_{t-h}$ and $\cX_t$, $1\le t-h<t\le T$. However, in the tensor factor model \eqref{eq:tensorfactor} where the signal part is random and serial correlated, the average is taken only over $T-h$ lagged products for each $h$. Thus, while the rate of the average of the signal-by-noise crossproducts in the factor model is heuristically expected to match that of HOOI at noise level $\sigma/T^{1/2}$, the rate of the average of the noise-by-noise crossproducts in the factor model is expected to only match that of HOOI with noise level $\sigma/T^{1/4}$. In Corollary \ref{cor:itipup1}, the contribution of the noise-by-noise crossproducts dominates the initial estimation error as the SNR requirement $\lambda(T^{1/4}/\sigma)\ge C_0 d^{1/4}$ in \eqref{eq:signal:itipup} matches that of HOOI with noise level $\sigma/T^{1/4}$; at the same time the contribution of the signal-by-noise crossproducts dominates the estimation error after iteration as the rate $(\sigma/T^{1/2}) d_k^{1/2}/\lam$ in \eqref{bound:itipup:cora} matches that of HOOI with noise level $\sigma/T^{1/2}$. Thus, if there is no severe signal cancellation, the signal to noise ratio requirement and convergence rate for iTIPUP and TIPIP-iTOPUP in the factor model are both comparable with those of HOOI in the simpler fixed signal setting, but the rate match is achieved in very different and subtle ways. We prove that this insight is intrinsic as the rates in \eqref{eq:signal:itipup} and \eqref{bound:itipup:cora} are both optimal according to the computational and statistical lower bounds in the following subsection.

\subsection{Computational and statistical lower bounds} \label{section:computation}
In this subsection, we focus on the typical factor model setting that the condition numbers of $A_k^\top A_k$ are bounded. 
We shall prove that under the computational hardness assumption, the signal to noise ratio condition \eqref{eq:signal:itipup} imposed on iTIPUP (also TIPUP-iTOPUP) in Corollary \ref{cor:itipup1} is unavoidable for computationally feasible estimators to be consistent. To be specific, we show that, if the signal to noise ratio condition is violated, then any computationally efficient and
consistent estimator of the loading spaces leads to a computationally efficient
and statistically consistent test for the Hypergraphic Planted Clique Detection problem in a regime where it is believed to be computationally intractable.
In addition, we establish a statistical lower bound on the minimax risk of the estimators.

{\it Hypergraphic Planted Clique.} An $m$-hypergraph $G=(V(G),E(G))$ is a natural extension of regular graph, where $V(G)=[N]$ and each hyper-edge is represented by an unordered
group of $m$ different vertices $i_j\in V(G)$ ($j=1,...,m$), denoted as  $e=(i_1,...,i_m)\in E(G)$. Given a $m$-hypergraph 
its adjacency tensor $\cA\in\{0,1\}^{N\times N\times\cdots\times N}$ is defined as
\begin{equation*}
\cA_{i_1,...,i_m}=
\begin{cases}
1, &  \text{if } e=(i_1,...,i_m)\in E(G);\\
0, &  \text{otherwise}.
\end{cases}
\end{equation*}
We denote by
$\cG_m(N, 1/2)$ the Erd\H{o}s–R\'enyi $m$-hypergraph on $N$ vertices where each hyper-edge $e$ is drawn independently with probability $1/2$, by $\cC=\cC(N,\kappa)$ a random clique of size $\kappa$ where the $\kappa$ members are uniformly sampled from $[N]$ and $E(\cC)$ is composed of all $e=(i_1,\ldots,i_m)$ with $i_j\in \cC$,
and by $\cG_m(N, 1/2, \kappa)$ the random graph generated by first sampling independently $\cG_m(N, 1/2)$ and $\cC=\cC(N,\kappa)$
and then adding all the edges in $E(\cC)$ to the set of edges in $\cG_m(N,1/2)$. The Hypergraphic Planted
Clique (HPC) detection problem of parameter $(N, \kappa, m)$ refers to testing the following hypotheses:
\begin{equation}\label{hpc}
H_0^G: \cA\sim \cG_m(N, 1/2) \quad \text{v.s.}\quad H_1^G: \cA\sim \cG_m(N, 1/2,\kappa).
\end{equation}

If $m=2$, the above HPC detection becomes the traditional planted clique (PC) detection problem. When $\kappa \ge c \sqrt{N}$, many computationally efficient algorithms have been developed for PC detection; see, \cite{alon1998finding, feige2000finding, feige2010finding, ames2011nuclear, dekel2014finding, deshpande2015finding, feldman2017statistical}, among others.
However, it has been widely conjectured that when $\kappa=o(\sqrt{N})$, the PC detection problem cannot be solved in randomized polynomial time, which is referred to as the hardness conjecture. Computational lower bounds in several statistical problems have been established by assuming the hardness conjecture of PC detection, including sparse PCA \citep{berthet2013complexity,berthet2013optimal, wang2016statistical}, sparse CCA \citep{gao2017sparse}, submatrix detection \citep{ma2015computational, cai2017computational}, community detection \citep{hajek2015computational}, etc.

Recently, motivated by tensor data analysis, hardness conjecture for HPC detection problem has been proposed; see, for example, \cite{zhang2018tensor, brennan2020reducibility, luo2022tensor, luo2020open, pananjady2022isotonic}. Similar to the PC detection, they hypothesized that when $\kappa = O(N^{1/2-\delta})$ with $\delta>0$, the HPC detection problem \eqref{hpc} cannot be solved by any randomized polynomial-time algorithm. Formally, the conjectured hardness of 
the HPC detection problem can be stated as follows.

\begin{hypothesis}[HPC detection] \label{hypos:hpc}
Consider the HPC detection problem \eqref{hpc} and suppose $m\ge 2$ is a fixed integer. If
\begin{align}
\limsup_{N\to\infty} \frac{\log\kappa}{\log N} \le \frac12 -\delta, \quad\text{for any }\delta>0,
\end{align}
for any sequence of polynomial-time tests $\{\psi\}_N:\cA\to \{0,1\}$,
\begin{align*}
\limsup_{N\to\infty} \big(\P_{H_0^G}(\psi(\cA)=1) + \P_{H_1^G}(\psi(\cA)=0) \big) >1/2 .
\end{align*}
\end{hypothesis}
Evidence supporting this hypothesis has been provided in \cite{zhang2018tensor,luo2022tensor}. This version of the hypothesis is similar to the one in \cite{berthet2013complexity, ma2015computational, gao2017sparse} for the PC detection problem.

For simplicity, we especially consider the factor model \eqref{eq:tucker-model} with each individual series of $\cF_t$ being mean 0 and independent,
\begin{align}\label{eq:tenfm1}
\cX_t=\lambda \cF_t\times_1 U_1\times_2...\times_K U_K+\cE_t,
\end{align}
where $U_k\in\R^{d_k\times r_k}$, $U_k^\top U_k = I$ for $1\le k\le K$, and $0<c_1\le \sigma_{\min}(\E \vec1(\cF_t) \vec1^\top(\cF_t))\le \sigma_{\max}(\E \vec1(\cF_t) \vec1^\top(\cF_t))\le c_2<\infty$. The probability space we consider in this section is
\begin{align}\label{eq:tenfm1sp}
&\sP(T,d_1,...,d_K,\lambda) =\Big\{\cX_1,...,\cX_T: \cX_t \text{ has form }\eqref{eq:tenfm1} \text{ with independent series } \cF_{t,i_1,...,i_K}, \\
&\qquad\qquad \frac{1}{T-1}\sum_{t=2}^T \E \cF_{t,i_1,...,i_K} \cF_{t-1,i_1,...,i_K}=c_0>0, \text{ and } \{\cF_t\}_{t=1}^T \text{ independent of } \{\cE_t\}_{t=1}^T,   \notag\\
&\qquad\qquad \cE_{t,j_1,...,j_K}\overset{i.i.d.}{\sim} N(0,\sigma^2), \text{ for all } 1\le t\le T, 1\le i_k\le r_k, 1\le j_k\le d_k, 1\le k\le K \Big\}. \notag
\end{align}
The computational lower bound over $\sP(T,d_1,...,d_K,\lambda)$ is then presented as below. In the case of $K=1$, the problem reduces to PCA of the spike covariance matrix. For general $K$, the auto-covariance tensor is of order $2K$.

\begin{theorem}\label{lowerbdd}
Suppose that Hypothesis \ref{hypos:hpc} holds for some $0<\delta<1/2$ and $d^{1/K}\asymp d_k\ge T$ and $r_k$ is fixed for all $1\le k\le K$. If, for some $\vartheta>0$,
\begin{align}\label{eq1:lowerbdd}
\liminf_{T\to \infty} \frac{\sigma^2d^{1/2-\vartheta}}{T^{1/2}\lambda^2} >0,
\end{align}
then for any randomized polynomial-time estimators $\widehat U_k=\widehat U_k(\cX_1,...,\cX_T)$, $1\le k\le K$,
\begin{align}\label{eq2:lowerbdd}
\liminf_{T\to \infty} \sup_{\cX_1,...,\cX_T\in \sP(T,d_1,...,d_K,\lambda) } \P \left( \min_{1\le k\le K} \| \widehat P_k -P_k\|_{\rm S}^2 > \frac13  \right)>\frac14 ,
\end{align}
where $\widehat P_k=\widehat U_k \widehat U_k^\top$ and $P_k= U_k U_k^\top$.
\end{theorem}

Comparing \eqref{eq1:lowerbdd} with \eqref{eq:signal:itipup}, we see that the signal to noise ratio condition \eqref{eq:signal:itipup}
cannot be improved upon by a factor of $d^\vartheta$ with polynomial time complexity for any $\vartheta>0$. The condition $d_k \ge T$ is a technical requirement to use the theoretical tools in \cite{ma2015computational} and \cite{brennan2020reducibility} for the reduction from HPC.

\begin{rmk}
Theorem \ref{lowerbdd} illustrates the computational hardness for factor loading spaces estimation under the typical factor model setting that the condition numbers of $A_k^\top A_k$ are bounded and ranks $r_k$ are fixed, and suggests the use of TIPUP initialization with proper fixed $h_0$ as it attains the computational lower bound under the typical factor model setting.
\end{rmk}

\begin{rmk}
In the general $r_k$ case, the optimal signal-to-noise ratio requirement falls between $\lambda^{2}/\sigma^2\gtrsim \max_{1\le k\le K} \sqrt{d}/(\sqrt{T}r_{-k})$ (by Theorem \ref{lowerbdd}) and $\lambda^{2}/\sigma^2\gtrsim \max_{1\le k\le K} d_{k}/(Tr_{-k})$ (by Theorem \ref{thm:stat_lowerbdd} below). It seems possible to unfold tensor into matrix and use the results of \cite{ma2015computational} to narrow the gap.
Anyways, a complete solution to this challenging problem is beyond the scope of our paper.
\end{rmk}

Next, we establish the statistical lower bound for the tensor factor model problem. Again, we consider the probability space \eqref{eq:tenfm1sp}.

\begin{theorem}\label{thm:stat_lowerbdd}
Suppose $\lambda>0$ and $d_k\to\infty$ as $T\to\infty$ for all $1\le k\le K$. Then there exists a universal constant $c>0$ such that for $T$ sufficiently large,
\begin{align}\label{eq1:stat_lowerbdd}
\inf_{\widehat U_k}  \sup_{\cX_1,...,\cX_T\in \sP(T,d_1,...,d_K,\lambda) } \E  \|\widehat P_k -P_k \|_{\rm S} \ge
c\,\min\left(1, (\sigma^2+\sigma\lambda)\sqrt{d_k}\big/(\lambda^2\sqrt{Tr_{-k}})\right)
\end{align}
for all $1\le k\le K$, where $\widehat P_k=\widehat U_k \widehat U_k^\top$
and $P_k= U_k U_k^\top$.
\end{theorem}

\begin{rmk}
The statistical lower bound for high-dimensional tensor factor models is provided in Theorem \ref{thm:stat_lowerbdd}. This bound directly matches the upper bounds in Corollary \ref{cor:itipup1} and also matches the bounds in Corollary \ref{cor:itopup1} when $d_k\gtrsim rr_{-k}^2$ and $\lambda^2/\sigma^2\gtrsim d_kr_{-k}^5/T+d_k^{1/2}r_{-k}^{3/2}/T^{1/2}$. These results demonstrate that the rates obtained by our proposed iterative procedures are minimax-optimal. Moreover, Theorem \ref{thm:stat_lowerbdd} reveals a different effect of the ranks $r_k$ ($k=1,...,K$) compared to tensor Tucker decomposition \citep{zhang2018tensor}, further confirming the distinct nature of tensor factor models from low-rank matrix/tensor problems.
\end{rmk}

\section{A matrix perturbation bound}

In Lemma \ref{lm-pertubation} below,
we provide an improvement of the matrix perturbation bound of \cite{wedin1972}. The lemma, proved in Appendix \ref{section:lemmas} in the supplementary material and used to prove Proposition \ref{prop:topup}, is of independent interest due to wide applications of the \cite{wedin1972} bound.

\begin{lemma}\label{lm-pertubation}
Let $r \le d_1\wedge d_2$, $M$ be a $d_1\times d_2$ matrix,
$U$ and $V$ be, respectively, the left and right singular matrices associated
with the $r$ largest singular values of $M$,
$U_{\perp}$ and $V_{\perp}$ be the orthonormal complements of $U$ and $V$,
and $\lam_r$ be the $r$-th largest singular value of $M$.
Let $\widehat M = M + \Delta$ be a noisy version of $M$,
$\{\widehat U, \widehat V, {\widehat U}_{\perp}, {\widehat V}_{\perp}\}$
be the counterpart of $\{U,V,V_\perp,V_\perp\}$, and
${\widehat \lam}_{r+1}$ be the $(r+1)$-th largest singular value of $\widehat M$.
Let $\|\cdot\|$ be a
matrix norm satisfying $\|ABC\|\le \|A\|_{\rm S}\|C\|_{\rm S}\|B\|$,
$\epsilon_1= \|U^\top \Delta {\widehat V}_{\perp}\|$ and $\epsilon_2 = \|{\widehat U}_{\perp}^\top \Delta V\|$. Then,
\begin{align}\label{wedin+}
\| U_{\perp}^\top \widehat U \|
\le \frac{{\widehat \lam}_{r+1}\epsilon_1+\lam_r\epsilon_2}{\lambda_r^2 - {\widehat \lam}_{r+1}^2}
\le \frac{\epsilon_1\vee\epsilon_2}{\lambda_r - {\widehat \lam}_{r+1}}.
\end{align}
In particular, for the spectral norm $\|\cdot\|=\|\cdot\|_{\rm S}$, $\hbox{\rm error}_1 =\|\Delta\|_{S}/\lambda_r$
and $\hbox{\rm error}_2 =\epsilon_2/\lambda_r$,
\begin{align}\label{wedin-2}
\|\widehat U \widehat U^\top  -U U^\top\|_{\rm S}\le \frac{\hbox{\rm error}_1^2+\hbox{\rm error}_2}{1-\hbox{\rm error}_1^2}.
\end{align}
\end{lemma}

The sharper perturbation bound in the middle of \eqref{wedin+} improves
the commonly used version of the \cite{wedin1972} bound on the right-hand side,
compared with Theorem 1 of \cite{cai2018} and Lemma 1 of \cite{chen2022rejoinder}.
As \cite{cai2018} pointed out, such variations of the \cite{wedin1972} bound
provide sharper convergence rate when $\hbox{\rm error}_2\le \hbox{\rm error}_1$ in \eqref{wedin-2},
typically in the case of $d_1\ll d_2$, as in Proposition \ref{prop:topup}.

\section{Summary}\label{section:summary}

In this paper we propose new estimation procedures for tensor factor model via iterative projection, and focus on two procedures: iTOPUP and iTIPUP. Theoretical analysis shows the asymptotic properties of the estimators. Simulation study presented in the supplementary material illustrates the finite sample properties of the estimators. While theoretical results are obtained under very general conditions, concrete specific cases are considered.
In particular, under the typical factor model setting where the condition numbers of $A_k^\top A_k$ are bounded and the ranks $r_k$ are fixed, the proposed iterative procedures, iTOPUP method and iTIPUP method (with no severe signal cancellation) lead to a convergence rate $O_{\P}((Td_{-k})^{-1/2})$ under strong factors settings due to information pooling of the orthogonal projection of the other $d_{-k}$ dimensions.
This rate is much sharper than the
existing rate $O_{\P}(T^{-1/2})$ in the recent literature for
non-iterative estimators for vector, matrix and
tensor factor models. It implies that the accuracy can be improved by increasing the dimensions, and consistent
estimation of the loading spaces can be achieved even with a fixed finite sample size $T$.
This is in sharp contrast to the folklore based on the existing literature that only the sample size $T$ helps the estimation of the loading matrices in factor models. The proposed iterative estimation methods not only preserve the tensor structure, but also result in sharper convergence rate in the estimation of factor loading space.

The iterative procedure requires two operators, one for initialization and one for iteration. Under certain conditions of
the signal to noise ratio (or the sample size requirement),
we only need the initial estimator to have sufficiently small estimation errors but not the consistency of the initial estimator. Often, one iteration is sufficient. In more complicated general cases, at most $O(\log(d))$ iterations are needed to achieve the ideal rate of convergence. Based on the theoretical
results and empirical evidence, we suggest to use iTOPUP for iteration when the ranks $r_k$ are small.
In terms of initiation, the computational lower bound shows that the signal to noise ratio condition derived from TIPUP initialization is unavoidable for any computationally feasible estimation procedure to achieve consistency, while that from TOPUP initialization is not optimal. Based on this result, we suggest the use of TIPUP initialization. Of course, this should be done with precaution against potential signal cancellation,
for example by using a slightly large $h_0$
as our empirical results show. By examination of the patterns
of estimated singular values under different lag values $h_0$, using iTOPUP and iTIPUP, it is possible to detect signal cancellation, which has significant impact on iTIPUP estimators.

The proposed iterative procedure is similar to HOOI algorithms in spirit, but the detailed operations and the theoretical challenges are significantly different.

\section*{Acknowledgements}
We would like to thank the Editor, the Associate Editor and the anonymous referees for their detailed reviews, which helped to improve the paper substantially.

Yuefeng Han's research is supported in part by National Science Foundation grant IIS-1741390. Rong Chen's research is supported in part by National Science Foundation grants DMS-1737857, IIS-1741390, CCF-1934924, DMS-2027855 and DMS-2319260. Dan Yang's research is supported in part by NSF grant IIS-1741390, Hong Kong grant GRF 17301620, Hong Kong grant CRF C7162-20GF and Shenzhen grant SZRI2023-TBRF-03. Cun-Hui Zhang's research is supported in part by NSF grants
DMS-1721495, IIS-1741390,
CCF-1934924,
DMS-2052949 and DMS-2210850.

\bibliographystyle{imsart-nameyear}
\bibliography{tensorfactor}

\newpage

\appendix

\title{Supplementary Material to ``Tensor Factor Model Estimation by Iterative Projection''}

\begin{aug}
\author{\fnms{Yuefeng} \snm{Han}
},
\author{\fnms{Rong} \snm{Chen}
},
\author{\fnms{Dan} \snm{Yang}
},
\and
\author{\fnms{Cun-Hui} \snm{Zhang}
}

\begin{center}
University of Notre Dame, Rutgers University and The University of Hong Kong
\end{center}

\end{aug}


In this supplementary material, we shall provide simulation studies, the proofs of main results in the paper and some lemmas that are useful in proofs of the paper.

The readers are referred to Appendix \ref{section:simulation} for simulation studies. The proofs of Theorems \ref{thm:itipup}, \ref{thm:itopup}, \ref{lowerbdd} and \ref{thm:stat_lowerbdd} are presented in Appendix \ref{proofth2}, \ref{proofth1}, \ref{proofth3} and \ref{proofth4}, respectively. Appendix \ref{proofcor} includes the proofs of corollaries. All technical lemmas are relegated to Appendix \ref{section:lemmas}.

\section{Simulation Study} \label{section:simulation}

In this section, we compare the empirical performance of different procedures of estimating the
loading matrices of a tensor factor model, under various simulation setups. Specifically, we consider the following procedures: the non-iterative and iterative methods, and the intermediate output from the iterative procedures when the number of iteration is 1 after initialization. If TIPUP is used as \text{$\widehat U_k$-INIT} and \text{$\widehat U_k$-ITER}, the one step procedure will be denoted as 1TIPUP. Similarly for 1UP and 1TOPUP. We consider the following combinations of \text{$\widehat U_k$-INIT} and \text{$\widehat U_k$-ITER}.
\begin{itemize}
\item UP based: (i) UP, (ii) 1UP and (iii) iUP
\item TIPUP based: (iv) TIPUP, (v) 1TIPUP and (vi) iTIPUP
\item TOPUP based: (vii) TOPUP, (viii) 1TOPUP and (ix) iTOPUP
\item mixed iterative: (x) TIPUP-1TOPUP, and (xi) TIPUP-iTOPUP
\item mixed iterative: (xii) TOPUP-1TIPUP, and (xiii) TOPUP-iTIPUP
\end{itemize}
Our empirical results show that (xii) and (xiii) are always inferior to (v) and (vi) respectively. Hence results used (xii) and (xiii) are not shown here.


We demonstrate the performance of all procedures under the setting of a matrix factor model,
\begin{equation}
    X_t=A_1 F_t A_2^\top+ E_t=\lambda U_1 F_t U_2^\top+ E_t, \label{eq:simu}
\end{equation}
and an order 3 tensor factor model
\begin{equation}
\cX_t=
\lambda\cF_t\times_1 U_1\times_2 U_2\times_3 U_3+\cE_t. \label{eq:simu2}
\end{equation}

Under matrix factor model \eqref{eq:simu}, we consider two experimental configurations:
\begin{enumerate}
\item[I.] We set $r_1=r_2=1$ for three purposes: to see the effect of sample size $T$ and signal strength $\lambda$; to check the effect of $h_0$ when there is no signal cancellation; and to verify the theoretical bounds on the sample size.
\item[II.] We consider the rank setup $r_1=1$ and $r_2=2$ while fixing the signal strength $\lambda$. We vary the parameters so that signal cancellation may or may not happen, and under the latter case, the choice of $h_0$ plays an important role.
\end{enumerate}
Under tensor factor model \eqref{eq:simu2}, we implement the following configuration:
\begin{enumerate}
\item[III.] The ranks are given as $r_1=r_2=r_3=2$.
\end{enumerate}
We repeat all the experiments 100 times.

Here, $E_t$ in matrix factor model \eqref{eq:simu} (resp. $\cE_t$ in tensor factor model \eqref{eq:simu2}) is white noise with no autocorrelation, $E_t\perp E_{t+h},h>0$ (resp. $\cE_t\perp \cE_{t+h},h>0$), and generated according to $E_t=\Psi_1^{1/2} Z_t\Psi_2^{1/2}$ (resp. $\cE_t = \cZ_t \times_1\Psi_1^{1/2}\times_2\Psi_2^{1/2}\times_3\Psi_3^{1/2}$), where all of the elements in the $d_1\times d_2$ matrix $Z_t$ (resp. $d_1\times d_2\times d_3$ tensor $\cZ_t$) are iid $N(0,1)$. Furthermore, $\Psi_1, ~\Psi_2$ (resp. $\Psi_3$) are the covariance matrices along each mode with
the diagonal elements being $1$ and all the off-diagonal elements being $\psi_1,~\psi_2$ (resp. $\psi_3$). The elements of the loading matrices $U_j$ of size $d_j\times r_j$, for $j=1,2,3$, are first generated from iid N(0,1), and then
orthonormalized through QR decomposition. We set different $\lambda$ values for different signal-to-noise ratio.

As for the factor time series, under matrix factor model Configuration I, the univariate $f_t$ follows AR(1) with AR coefficient $\phi$; under matrix factor model Configuration II, $F_t$ is a $1\times 2$ matrix and the two univariate time series $f_{it}$ follow AR(1) $f_{it}=\phi_{i}f_{i(t-1)}+\epsilon_{it}$ independently with two AR coefficients $\phi_1$ and $\phi_2$ respectively; under tensor factor model Configuration III, $\cF_t$ is a $2\times 2 \times 2$ tensor with eight independent univariate time series, where three follow AR(1) processes $f_{111t}=0.7f_{111t-1}+e_{111t}, ~f_{211t}=0.6f_{211t-1}+e_{211t}, ~f_{222t}=0.8f_{222t-1}+e_{222t}$, one follows AR(2) process $f_{221t}=0.5f_{221t-1}+0.3f_{221t-2}+e_{221t}$, and four, $f_{112t},f_{121t},f_{122t},f_{212t}$, are white noise. Here, all of the innovations follow iid $N(0,1)$.

We fix dimensions $d_1=d_2$ ($=d_3$) $=16$ and use the estimation error for $A_1$ or $U_1$ as criterion: $\|\hat{P}_{1}-P_{1}\|_{\rm S}$. The $\lambda$ in \eqref{eq:simu} and \eqref{eq:simu2} is $\lambda=\prod_{k=1}^K \| A_k \|_{\rm S}$ in Corollaries \ref{cor:itopup1} and \ref{cor:itipup1}. 

Configuration I satisfies Assumption~\ref{asmp:factor} since the rank $r_1$ and $r_2$ are fixed and the factor process is stationary. We performed three experiments under Configuration I for three purposes.

\noindent
{\bf Experiment 1 under Configuration I.} We set the off-diagonal entries of the covariance matrices of the noise as $\psi_1=\psi_2=0.2$ and the AR coefficient of the factor time series as $\phi=0.8$, and vary the sample size $T=256, ~1024$ and signal strength $\lambda=1,2,~4$. Figure \ref{fig-sim-rank1} shows the boxplot of the logarithm of the estimation errors for methods (i)-(ix). The performance of the mixed algorithms (x)-(xi) is not shown because they are identical to the corresponding methods (v) and (vi), 
under the rank one setting when the same initialization is used. We use
$h_0=1$ and $\hat r_2 =1$ in the process of the estimation. It can be seen easily from Figure~\ref{fig-sim-rank1}
that UP, 1UP, and iUP are always the worst, showing the advantage of the methods that accommodate time series features
and the disadvantage of neglecting temporal correlation. We will exclude UP, 1UP, and iUP from comparison for the rest of the simulation. When the sample size is small and the signal is weak ($T=256$ and $\lambda=1$), none of the methods work well, though procedures using TIPUP work sometimes. When the sample size is not too small or the signal strength is not too weak (shown in all panels except for the top left one), one-step methods (1TIPUP and 1TOPUP) are better than the
noniterative methods (TIPUP and TOPUP), and iterative methods (iTIPUP and iTOPUP) are in turn better than the one-step methods. When the sample size and signal strength increase, all methods perform better, but meanwhile the advantage of iterative methods over one-step methods and the advantage of one-step methods over initialization methods become smaller. When the sample size is large and signal is strong, the one-step methods are similar to the iterative methods after convergence, corroborating Corollaries \ref{cor:itopup1} and \ref{cor:itipup1}.

\begin{figure}[h]
  \centering
  \includegraphics[width=\textwidth]{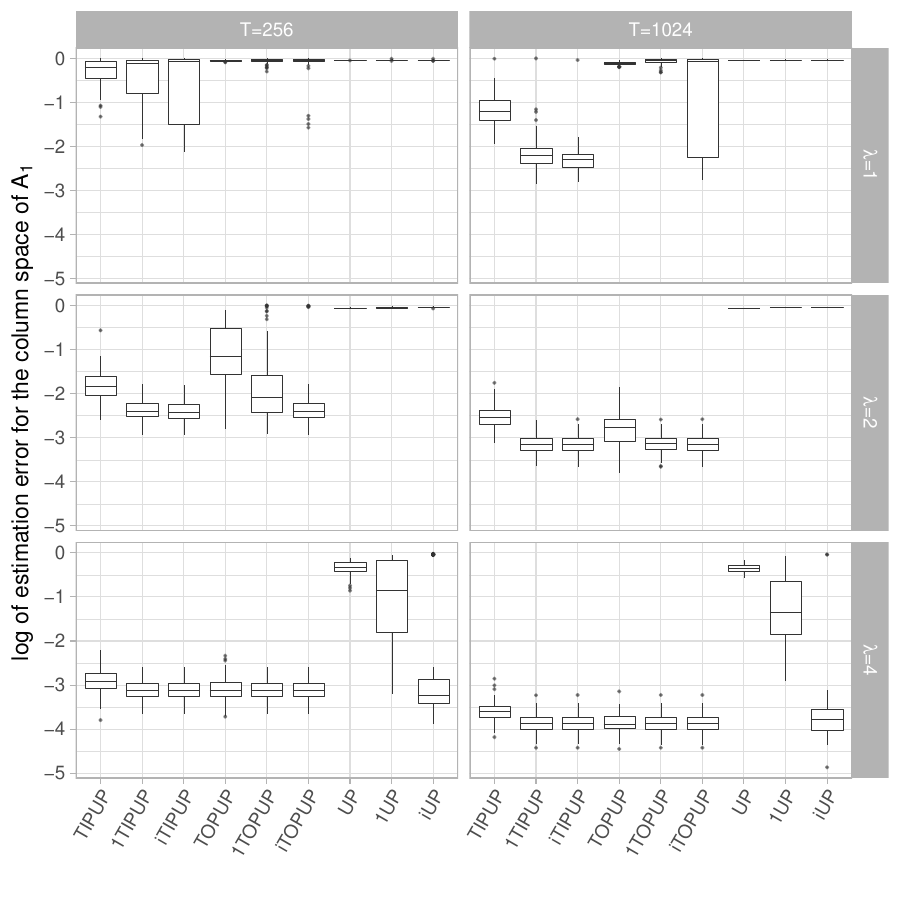}
  \caption{Experiment 1 under Configuration I. Boxplot of the logarithm of the estimation error of $A_1$. Nine methods (i)-(ix) are considered in total. Three rows correspond to three signal-to-noise strengths $\lambda=1,2,4$. Two columns correspond to two sample sizes $T=256,1024$.}
  \label{fig-sim-rank1}
\end{figure}

It is somewhat surprising to observe, from the top left panel in Figure \ref{fig-sim-rank1} ($T=256$ and $\lambda=1$), that in the small sample size and low signal strength case, the median error of iTIPUP is larger than that of 1TIPUP, which in turn is larger than TIPUP, whereas the order is reversed under stronger signal to noise ratio or with larger sample size shown in the other panels. Furthermore, the top right panel in Figure \ref{fig-sim-rank1} shows that, with weak signal to noise ratio, the TIPUP based methods perform better than the TOPUP based methods. This observation coincides with the results in Corollaries \ref{cor:itopup1} and \ref{cor:itipup1}, which together state that iTOPUP requires larger signal-to-noise ratio for consistency than iTIPUP. Figure \ref{fig-sim-rank1-track} produces some deeper insight, where the trajectories of the iterative methods (including initial estimations, estimations after one iteration, and the estimations after final convergence) of the 100 repetitions are connected, for the $T=256$ case. The top two panels show that when signal is weak and the sample size is small, the initial estimates may be poor, and the iterative methods may need certain accuracy in the
initial estimates to produce further improvement. This reemphasizes the condition on the initial estimate in the theorems. The bottom two panels show that when signal is stronger, the relatively more accurate initial estimates
enable the iterative methods to improve the estimates. Again, TOPUP initial estimates are not as accurate as the
TIPUP estimates.

\begin{figure}[h]
  \centering
  \includegraphics[width=\textwidth]{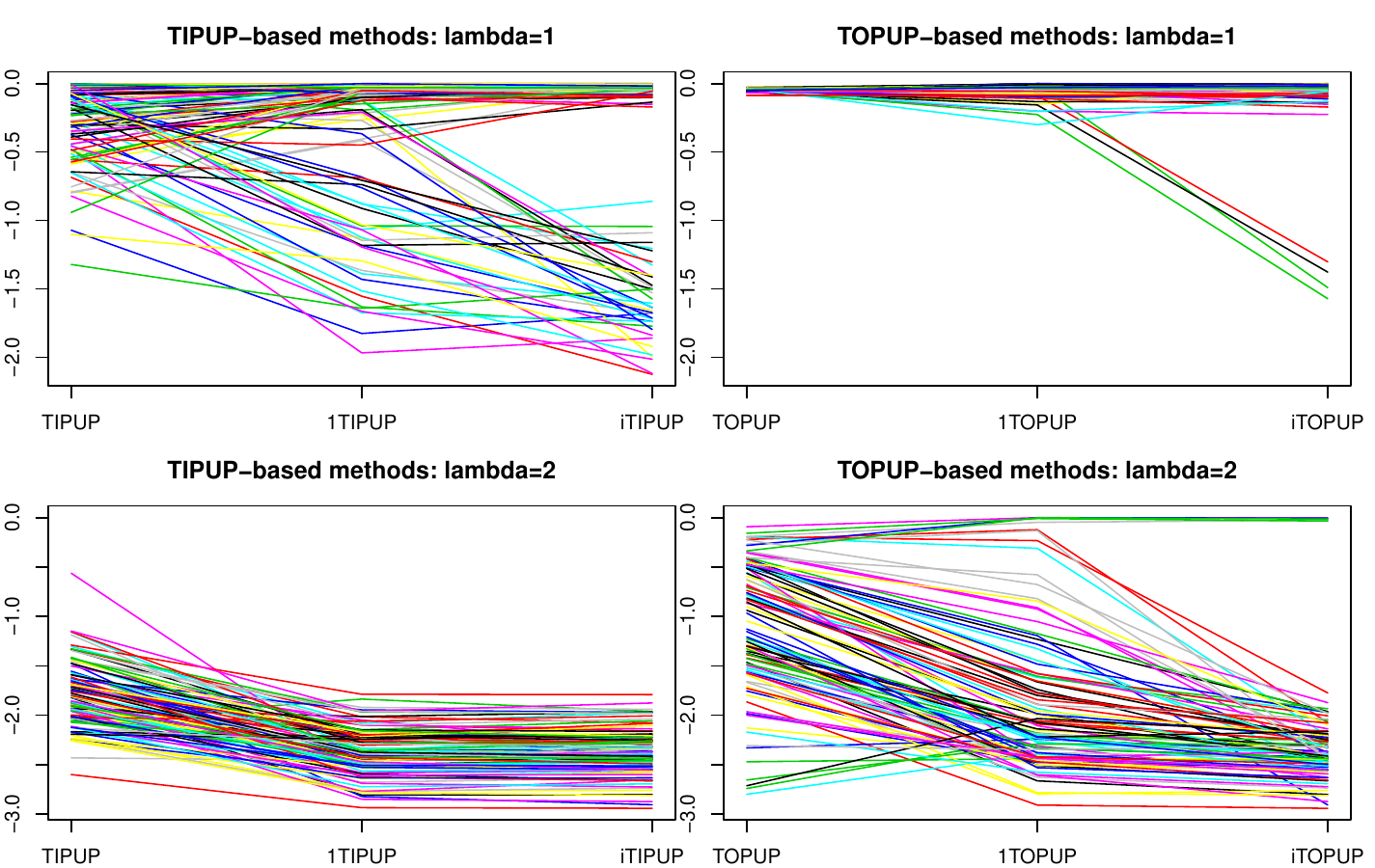}
  \caption{Experiment 1 under Configuration I. Trajectory of the logarithm of the estimation error of $A_1$ with fixed sample size $T=256$. Two rows correspond to two signal-to-noise strengths $\lambda=1,2$. Two columns correspond to TIPUP-based and TOPUP-based methods respectively.}
  \label{fig-sim-rank1-track}
\end{figure}

\noindent
{\bf Experiment 2 under Configuration I.} We use the same setting as above, but vary $h_0=1,2,\ldots,5$, and fix $\hat r_2 =1$ in the process of the estimation. Figure \ref{fig-sim-extra1} provides the results. It can be seen that when there is no signal cancellation, the choice of $h_0$ does not affect the performance dramatically. When $h_0$ increases, the performance of all TIPUP-based and TOPUP-based algorithms becomes slightly worse most of the time.

\begin{figure}[h]
  \centering
  \includegraphics[width=\textwidth]{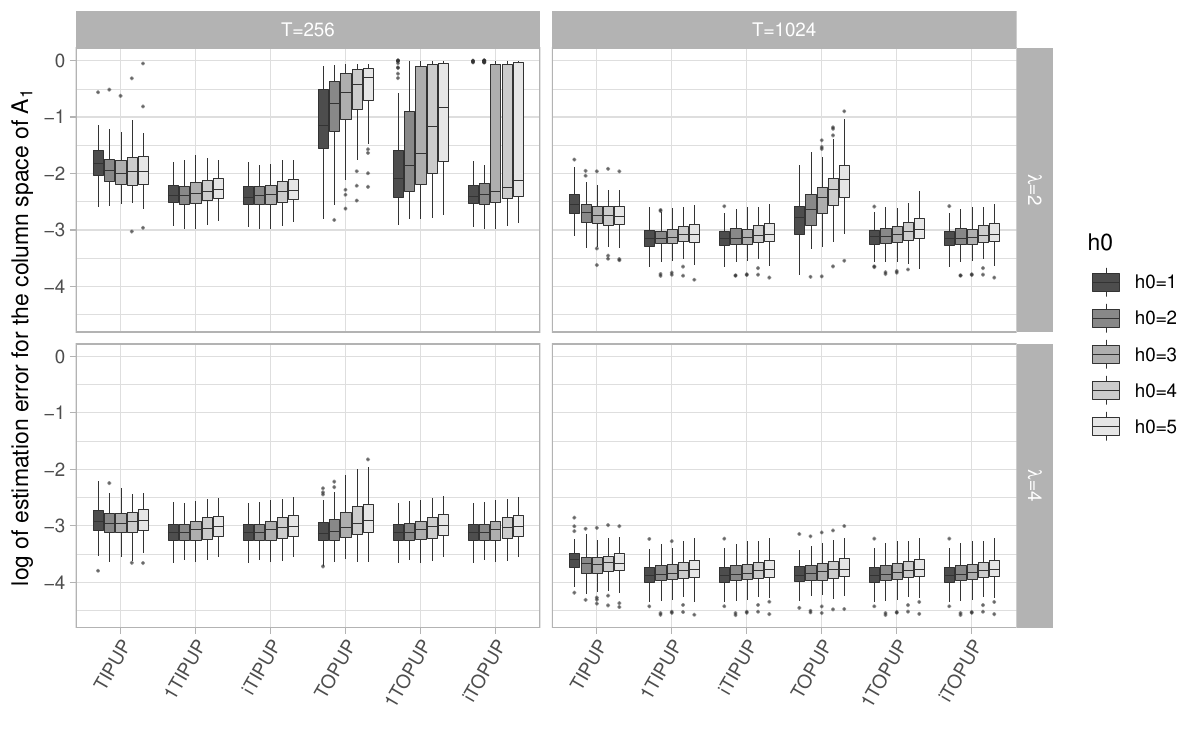}
  \caption{Experiment 2 under Configuration I. Boxplot of the logarithm of the estimation error of $A_1$. Six methods (iv)-(ix) with five choices of $h_0$ are considered in total. Two rows correspond to two signal-to-noise strengths $\lambda=2,4$. Two columns correspond to two sample sizes $T=256,1024$.}
  \label{fig-sim-extra1}
\end{figure}

\noindent
{\bf Experiment 3 under Configuration I.} This experiment is conducted to verify the bounds on the sample size for iTOPUP in Corollary \ref{cor:itopup2} and for iTIPUP in Corollary \ref{cor:itipup2}. We set the off-diagonal entries of the covariance matrices of the noise as $\psi_1=\psi_2=0.4$ and the AR coefficient of the factor time series as $\phi=0.9$, and vary the sample size $T=16, 64, 256, 1024, 4096$ and signal strength $\lambda=1,2,4,8$. Again, we use $h_0=1$ and $\hat r_2 =1$ in the process of the estimation.

Table \ref{tbl:bound-sim} provides the values of $\delta_0,\delta_1$ in Assumption \ref{asmp:strength} as signal strength $\lambda$ varies and the lower bounds on the sample size in \eqref{eq:sample:itopup} and \eqref{eq:sample:itipup} when the values of $\delta_0,\delta_1$ are plugged in. We have $\delta_2=0$ in this case and $d=d_1\times d_2=256$. Figure \ref{fig-sim-extra2} shows the simulation results, where the raw estimation error instead of the logarithm is given. When the raw errors are close to 1, it implies the estimation is not accurate. It is corroborated that the sample size $T$ can be much smaller than $d$ in the strong factor case.

\begin{table}[h]
\begin{tabular}{rrrrr}
  \hline
 $\lambda$ & $\delta_0$ & $\delta_1$ & iTOPUP bound \eqref{eq:sample:itopup}& iTIPUP bound \eqref{eq:sample:itipup}\\\hline
      1&   1.00&   1.00&     4096.00&      256.00\\
      2&   1.00&   1.00&     4096.00&      256.00\\
      4&   0.82&   0.86&      829.44&       81.00\\
      8&   0.57&   0.61&       51.84&        5.06\\
\hline
\end{tabular}
  \caption{Experiment 3 under Configuration I. This table provides the values of $\delta_0,\delta_1$ in Assumption \ref{asmp:strength} as signal strength $\lambda$ varies and the lower bounds on the sample size in \eqref{eq:sample:itopup} and \eqref{eq:sample:itipup} when the values of $\delta_0,\delta_1$ are plugged in. }
  \label{tbl:bound-sim}
\end{table}


\begin{figure}
  \centering
  \includegraphics[width=\textwidth,angle=0]{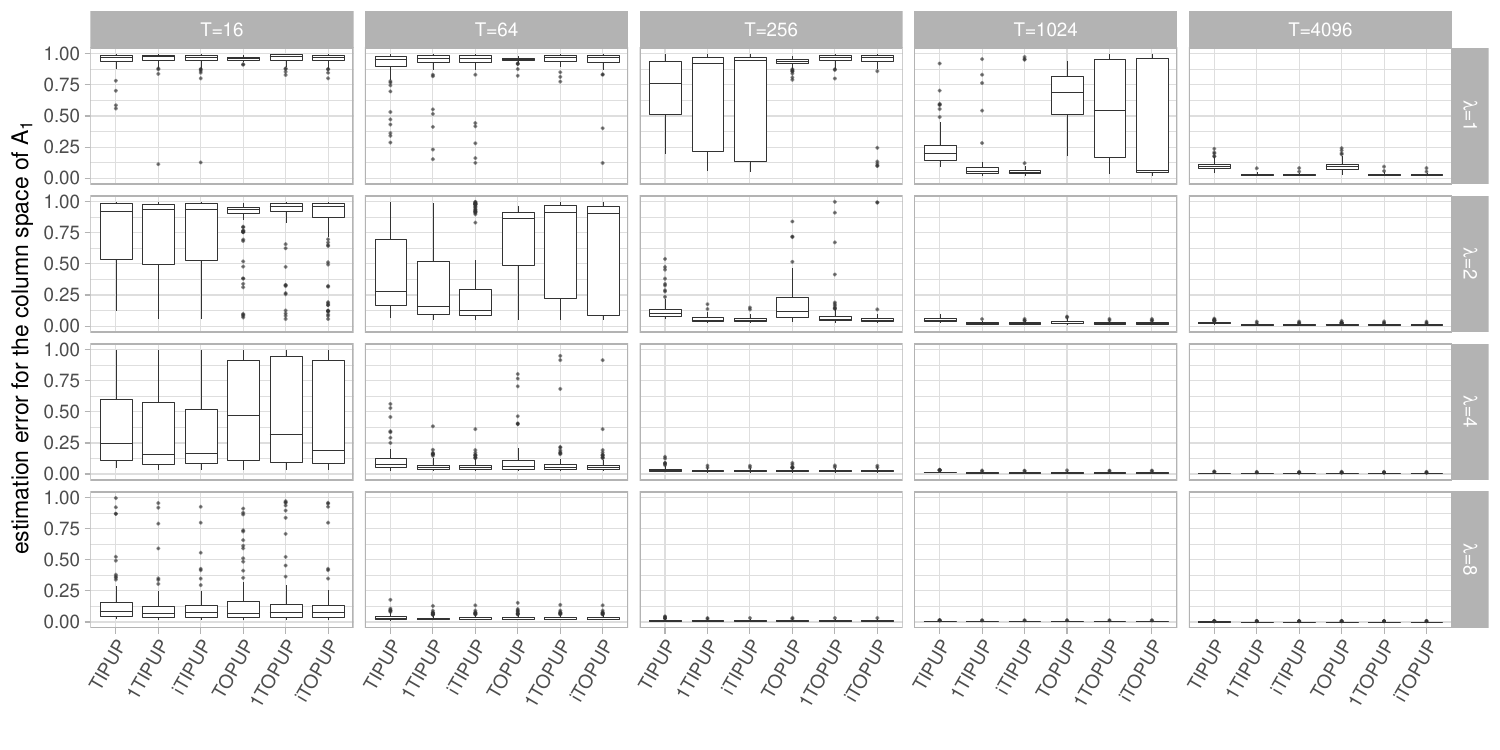}
  \caption{Experiment 3 under Configuration I. Boxplot of the estimation error of $A_1$. Six methods (iv)-(ix) are considered in total. Four rows correspond to four signal-to-noise strengths $\lambda=1,2,4,8$. Five columns correspond to five sample sizes $T=16, 64, 256, 1024, 4096$. This figure corroborates the theoretical lower bounds on the sample size in \eqref{eq:sample:itopup} and \eqref{eq:sample:itipup}.}
  \label{fig-sim-extra2}
\end{figure}
Under Configuration II, we performed two experiments:
the two AR coefficients for the two independent univariate time series $f_{1t}$ and $f_{2t}$ are $\phi_1=0.8$ and $\phi_2=0.6$ in Experiment 1 and $\phi_1=0.8$ and $\phi_2=-0.8$ in Experiment 2. Experiment 1 under configuration II satisfies Assumption~\ref{asmp:factor}(i)-(ii) because there is no signal cancellation; Experiment 2 under configuration II satisfies Assumption~\ref{asmp:factor}(i) for TOPUP related methods. When $h_0=1$, Experiment 2 does not satisfy Assumption~\ref{asmp:factor}(ii) for TIPUP related methods as there is a severe signal cancellation. However, using $h_0=2$ significantly reduces signal cancellation as lag 2 auto-cross-covariance does not cancel each other in TIPUP related methods.

\noindent
{\bf Experiment 1 under Configuration II.}
Figure \ref{fig-sim-rank2-nocancel} shows the boxplot of the logarithm of the estimation errors of 8 methods including (iv)-(ix) and mixed (x)-(xi) with TIPUP initiation and TOPUP iteration. Again, the performance of the
mixed (xii)-(xiii) procedures with iTIPUP iteration is not as good as that of iTIPUP hence not shown.
Here we use
different sample sizes, with the signal strength fixed at $\lambda=1$ and two $h_0$ values: $h_0=1$ and $h_0=2$.
The theoretical $\lambda_1$  defined in \eqref{lam_k} and $\lambda_1^*$ in \eqref{lam^*_k} under the stationary
auto-cross-moments of the factor process are given in the figure. Note that they are different for different $h_0$.
It shows that the mixed TIPUP-1TOPUP method can slightly improve 1TOPUP because of the better initialization. With
larger sample size $T=1024$, TIPUP-1TOPUP also slightly outperforms 1TIPUP. In this case, using the larger $h_0=2$
provides slightly poorer performance than $h_0=1$, as the lag-2 autocorrelation is significantly smaller than that of
lag 1 for the underlying AR(1) process with $\phi_2=0.6$. The extra term adds limited signal, shown by the
small differences in $\lambda_1$ and $\lambda_1^*$, but incorporates extra noise terms in the estimators. To see more
clearly the impact of $h_0$, we show the boxplots of the estimated $\lambda_1^*$ and $\lambda_1$ using
iTIPUP and iTOPUP, respectively, for $h_0=1,2$ and $3$, under different sample sizes in Figure \ref{fig-sim-rank2-nocancel-lambda}. The theoretical values are marked with a diamond. It is seen
that the estimated values are relatively close to the theoretical values. More importantly, they decrease as
$h_0$ increases in this no-signal cancellation case.


\begin{figure}[h]
  \centering
  \includegraphics[width=\textwidth,page=1]{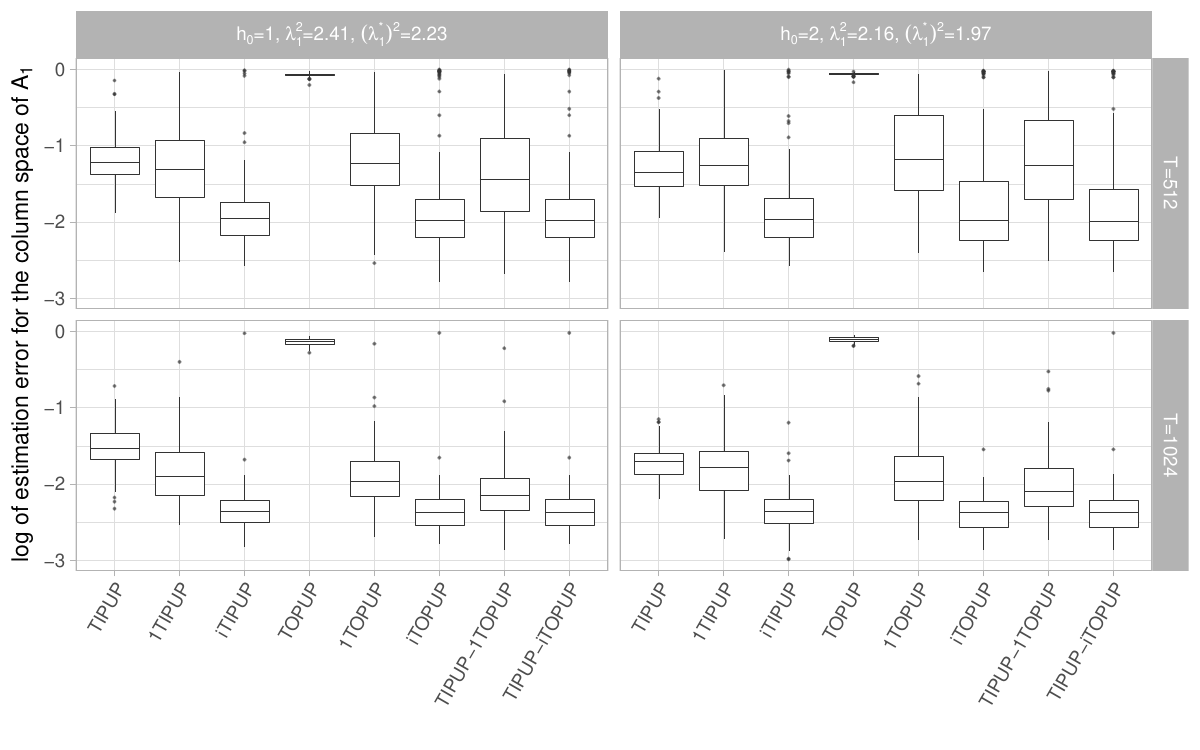}
  \caption{Experiment 1 under Configuration II. Boxplot of the logarithm of the estimation error of $A_1$. Eight methods are considered in total. Two rows correspond to two sample sizes $T=512,1024$. Two columns correspond to two choices of $h_0$.  The population signal strengths $\lambda_1^2$ \eqref{lam_k} and $\lambda_1^{*2}$ \eqref{lam^*_k} for different $h_0$ are provided on the top.}
  \label{fig-sim-rank2-nocancel}
\end{figure}

\begin{figure}[h]
  \centering
  \includegraphics[width=\textwidth,page=1]{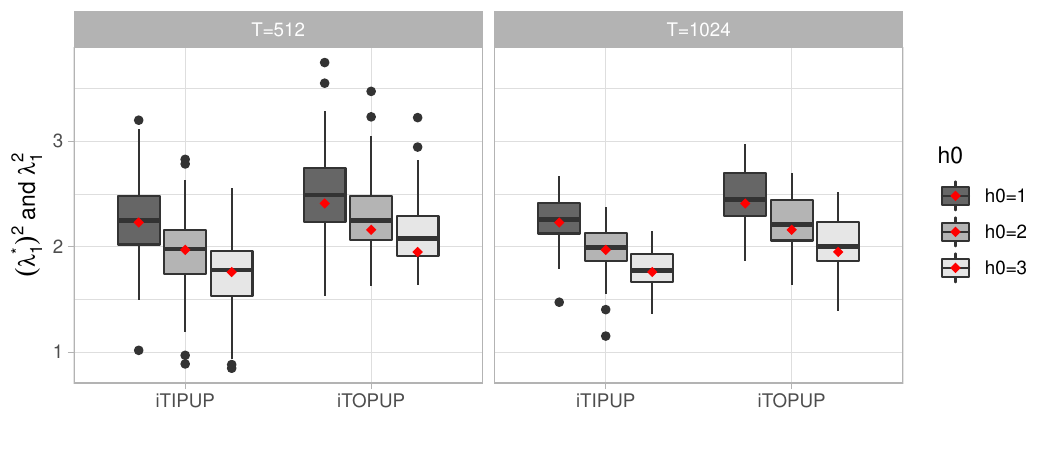}
  \caption{Experiment 1 under Configuration II. Boxplot of the sample estimates of the signal strengths $\lambda_1^2$ \eqref{lam_k} and $\lambda_1^{*2}$ \eqref{lam^*_k} over 100 replications for iTIPUP and iTOPUP with three choices of $h_0$. Two panels correspond to two sample sizes $T=512,1024$. The superimposed red diamonds are the population version of the signal strengths.}
  \label{fig-sim-rank2-nocancel-lambda}
\end{figure}

\noindent
{\bf Experiment 2 under Configuration II.}
When $\phi_1=0.8$ and $\phi_2=-0.8$, we can readily check that $\mathbb{E}(F_t F_{t-1}^\top)=(\phi_1+\phi_2)\sigma^2=0$. Therefore, in the TIPUP-related procedure for estimating $A_1$ with $h_0=1$, the signal completely cancels out. Since the ranks $r_1$ and $r_2$ are fixed, we have $\delta_2=\infty$ for $h_0=1$,
and the corresponding $\lambda_1^*=0$. Figure \ref{fig-sim-rank2-cancel} shows the boxplot of the logarithm of the estimation error of $A_1$ for 8 methods including (iv)-(ix) and mixed (x)-(xi) with two choices of $h_0=1$ and $h_0=2$. We fix the signal strength to be $\lambda=1$ to isolate the effect of $h_0$. When $h_0=1$, both initialization TIPUP and TOPUP do not perform well. But 1TOPUP and iTOPUP improve the performance of TOPUP significantly with TOPUP iteration while 1TIPUP and iTIPUP cannot improve TIPUP. This is because signal cancellation has significant impact on TIPUP based
procedures while having no impact on TOPUP based procedures. To our pleasant surprise, when $h_0=1$,
the mixed TIPUP-1TOPUP is better than both 1TIPUP and 1TOPUP, and the mixed TIPUP-iTOPUP is similar to iTOPUP and much better than iTIPUP.
When using $h_0=2$, the noise cancellation is mild and $(\lambda_1^*)^2=1.78$. Since $r_k$ are fixed, we have
$\delta_2<\infty$. Note that in this case the signal using TIPUP only comes from lag-2 cross product and is weaker
than that using TOPUP related procedures. The difference does not have impact on the convergence rate, but on the
signal to noise ratio. Comparing the left two subfigures with the right ones of Figure~\ref{fig-sim-rank2-cancel}, it is seen that
using $h_0=2$ always boosts the performance of TIPUP-related methods significantly.
Meanwhile, the TOPUP based methods are not sensitive to the choice of $h_0$.
When $h_0=2$, the non-iterative TIPUP performs better than TOPUP, 1TIPUP performs better than 1TOPUP, but after convergence, iTOPUP performs better than iTIPUP. Because the initialization TIPUP is better than TOPUP for $h_0=2$, it is of no surprise to see that TIPUP-1TOPUP behaves better than 1TIPUP and 1TOPUP, and TIPUP-iTOPUP is similar as iTOPUP and slightly better than iTIPUP.

\begin{figure}[h]
  \centering
  \includegraphics[width=\textwidth,page=2]{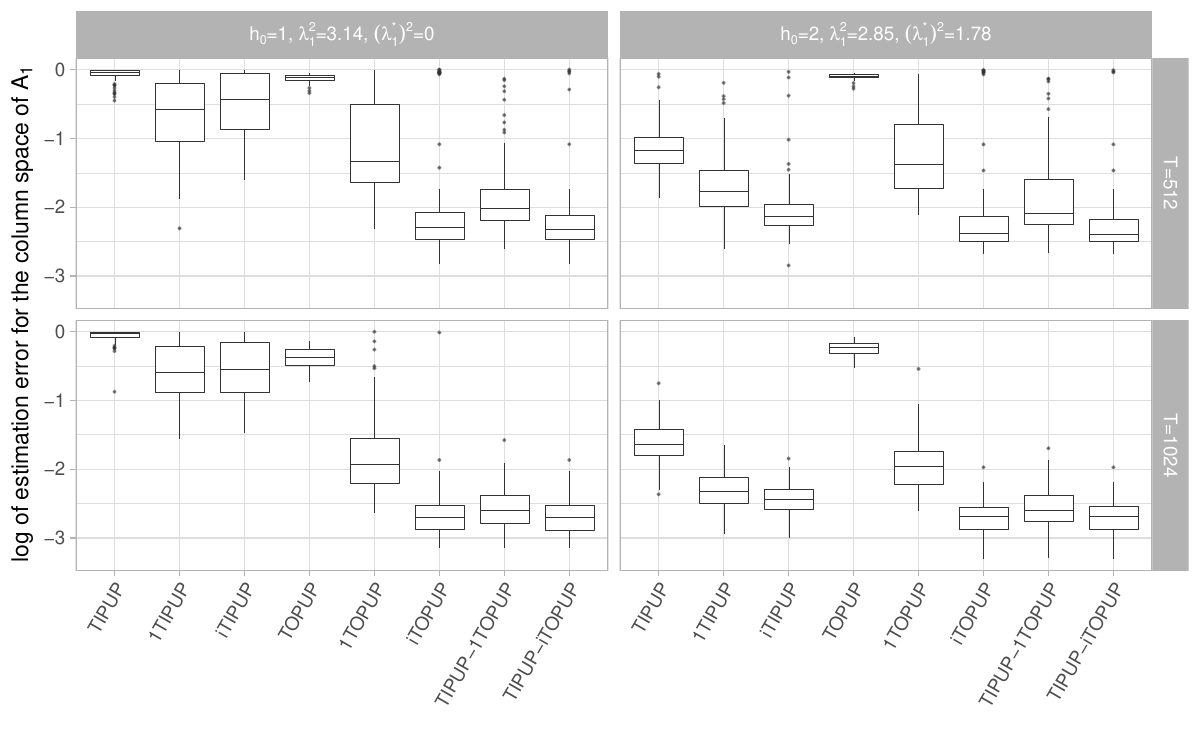}
  \caption{Experiment 2 under Configuration II. Boxplot of the logarithm of the estimation error of $A_1$. 8 methods are considered in total. Two rows correspond to two sample sizes $T=512,1024$. Two columns correspond to two choices of $h_0$.  The population signal strengths $\lambda_1^2$ \eqref{lam_k} and $\lambda_1^{*2}$ \eqref{lam^*_k} for different $h_0$ are provided on the top.}
  \label{fig-sim-rank2-cancel}
\end{figure}

Again, to see more clearly the impact of $h_0$ in this case with noise cancellation, we show the boxplots of the estimated $\lambda_1^*$ and $\lambda_1$ using
iTIPUP and iTOPUP, respectively, for $h_0=1,2$ and $3$ in Figure \ref{fig-sim-rank2-cancel-lambda}. It is seen
that the iTOPUP procedure remains robust in estimating $\lambda_1$ under the noise-cancellation case. And
$\lambda_1$ decreases as $h_0$ increases. However, iTIPUP is very different. Although when using $h_0=1$ the estimated
$\lambda_1^*$ significantly overestimates the theoretical value $\lambda_1^*=0$, they are still much less than
those from using $h_0=2$ and $3$. The reversed order of the magnitude of $\lambda_1^*$ as $h_0$ increases can be potentially used to detect
signal cancellation in practice, though the theoretical property of the estimators of $\lambda_1^*$ (e.g. standard
deviation) is technically challenging to obtain. In practice, when one observes such a reversed order, it is
recommended to use iTOPUP as a conservative estimator. Of course, the behaviors of $\lambda_k$ and $\lambda_k^*$
depend on the auto-cross-moment structure of the underlying factor process. For example, if the factor process
follows a MA(2) model with zero lag-1 autocorrelation ($f_t=e_t+\theta_2e_{t-2}$), then $\lambda_1$ and $\lambda_1^*$
under $h_0=2$ would be larger than those under both $h_0=1$ and $h_0=3$. But we expect that the pattern of
$\lambda_1$ under different $h_0$ would be similar to that of $\lambda_1^*$ under different $h_0$, if there is no severe signal cancellation. Severe signal cancellation would make the patterns different.

\begin{figure}[h]
  \centering
  \includegraphics[width=\textwidth,page=2]{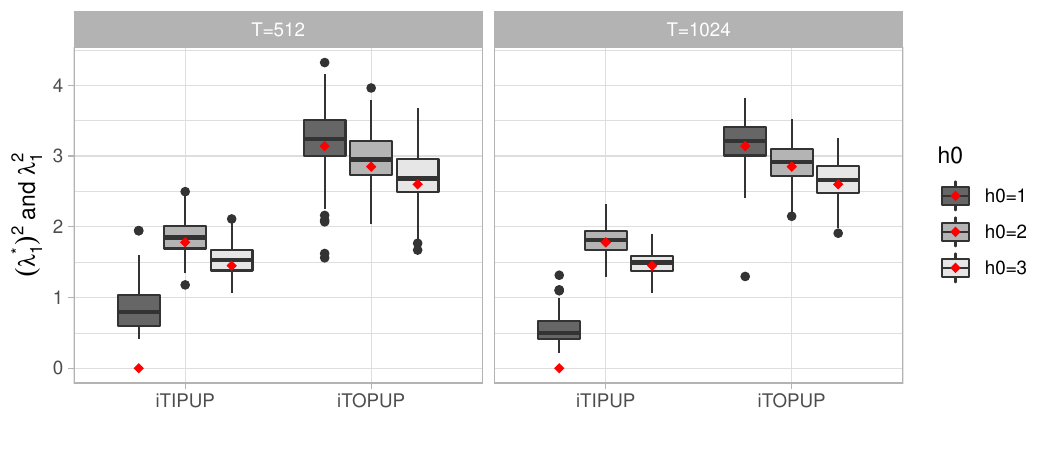}
  \caption{Experiment 2 under Configuration II. Boxplot of the sample estimates of the signal strengths $\lambda_1^2$ \eqref{lam_k} and $\lambda_1^{*2}$ \eqref{lam^*_k} over 100 replications for iTIPUP and iTOPUP with three choices of $h_0$. Two panels correspond to two sample sizes $T=512,1024$. The superimposed red diamonds are the population version of the signal strengths.}
  \label{fig-sim-rank2-cancel-lambda}
\end{figure}

\noindent
{\bf Experiment under Configuration III.} With order-3 tensor factor model \eqref{eq:simu2}, Configuration III satisfies Assumption~\ref{asmp:factor}, and Figure \ref{fig-sim-extra3} shows the results. The message is almost the same as in Figure \ref{fig-sim-extra1} for the matrix factor model. That is, TIPUP offers better initialization than TOPUP, which supports the theoretically smaller requirement on the sample size by TIPUP; iterative methods are better than one-step methods, which are in turn better than the non-iterative methods; when there is no signal cancellation, the iterative methods are in general not sensitive to the choice of $h_0$ and the non-iterative and one-step methods tend to behave slightly worse with larger values of $h_0$; when iTOPUP converges, its performance is better than that of iTIPUP as shown in the bottom right panel, which also verifies the theoretical claim of faster convergence rate of iTOPUP.

\begin{figure}[h]
  \centering
  \includegraphics[width=\textwidth]{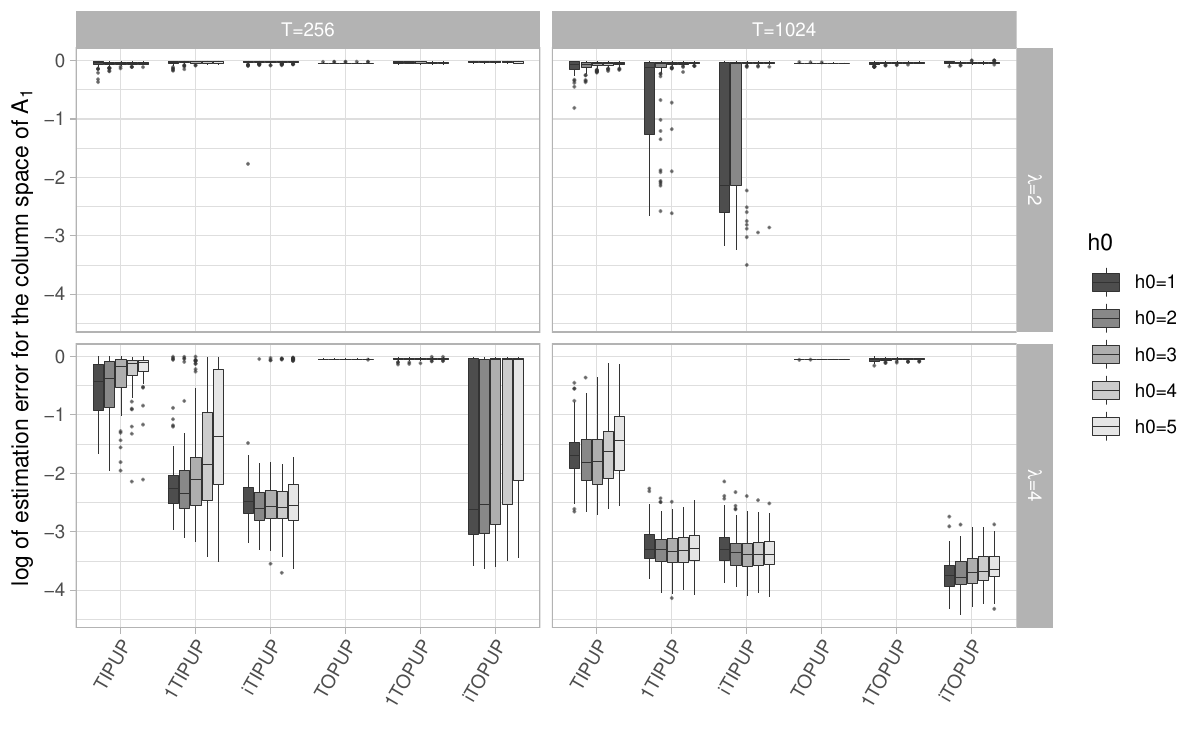}
  \caption{Experiment under Configuration III for order-3 tensor factor model. Boxplot of the logarithm of the estimation error of $A_1$. Six methods (iv)-(ix) with five choices of $h_0$ are considered in total. Two rows correspond to two signal-to-noise strengths $\lambda=2,4$. Two columns correspond to two sample sizes $T=256,1024$.}
  \label{fig-sim-extra3}
\end{figure}

\section{Proof of Theorem \ref{thm:itipup}}\label{proofth2}

We focus on the case of $K=2$ as the iTIPUP begins with mode-$k$ matrix unfolding. In particular, we sometimes give explicit expressions only in the case of $k=1$ and $K=2$. For $K=2$, we observe a matrix time series with $X_t=A_1 F_t A_2^\top + E_t \in \mathbb R^{d_1\times d_2}$. Recall 
that under the conditional expectation $\overline\E$, $F_1,...F_T$ are fixed. Let $U_1$, $U_2$ be the left singular matrices of $A_1$ and $A_2$ respectively with $r_k=$rank$(U_k)=$rank$(A_k)$.

We outline the proof as follows. Let $L^{(m)}_k$ be
the loss \eqref{loss} for
$\widehat U_k^{(m)}$ or equivalently the spectral norm error for
$\widehat P_k^{(m)} = \widehat U_k^{(m)}\widehat U_k^{(m)\top}$, $k=1,...,K$, and $L^{(m)}$ their maximum,
\begin{equation}\label{L^m_k}
L^{(m)}_k = 
\|\widehat P_k^{(m)} - P_k\|_{\rm S},\quad
L^{(m)}=\max_{k=1,2,...,K} L^{(m)}_k.
\end{equation}

From \citet{chen2022factor}, $\overline\E\big[L^{(0)}_k\big]\lesssim {R}^{*(0)}_k$
as we mentioned in \eqref{TIPUP_k-bd}.
By applying the Gaussian concentration inequality for Lipschitz functions and Lemme \ref{lm-GH}
in their analysis, we have
\begin{equation}\label{TIPUP_k-bd-tail}
L^{(0)} \le C_{1}^{\tipup}R^{*(0)}\ \text{ with }\ {R}^{*(0)}=\max_{1\le k\le K}{R}^{*(0)}_k
\end{equation}
in an event $\Omega_0$ with $\overline\P(\Omega_0)\ge 1- 5^{-1} \sum_{k=1}^Ke^{-d_k}$.
This is similar to \eqref{ideal-bd} below.

After the initialization with $\widehat U_k^{(0)}$,
the algorithm iteratively produces estimates $\widehat U_k^{(m)}$ from $m=1$ to $m=J$.
Define the matrix-valued operator $\text{TIPUP}_1(\cdot)$ as
\begin{equation*}
\text{TIPUP}_1(\widetilde U_2) =
\bigg(\sum_{t=h+1}^T \frac{X_{t-h}\widetilde U_2\widetilde U_2^\top X_{t}^\top}{T-h},
h=1,...,h_0\bigg)\in\R^{d_1\times (d_1h_0)}
\end{equation*}
for any matrix-valued variable $\widetilde U_2\in \R^{d_2\times r_2}$.
Given $\widehat U_2^{(m)}$, the $(m+1)$-th iteration produces estimates
$$
\widehat U_1^{(m+1)} = \text{LSVD}_{r_1}\big(\text{TIPUP}_1(\widehat U_2^{(m)})\big),\quad
\widehat P_1^{(m+1)} = \widehat U_1^{(m+1)}\widehat U_1^{(m+1)\top}.
$$
The ``noiseless" version of this update is given by
\begin{equation}\label{error-free-update}
\Theta_{1,h}^*(\widetilde U_2)
=\sum_{t=h+1}^T \frac{A_1 F_{t-h} A_2^\top \widetilde U_2\widetilde U_2^\top A_2  F_t^\top A_1^\top} {T-h},\quad
\overline{\E}[\text{TIPUP}_1](\widetilde U_2)
=\Theta_{1,1:h_0}^*(\widetilde U_2),
\end{equation}
with $\Theta_{1,1:h_0}^*(\widetilde U_2)=\big(\Theta_{1,h}^*(\widetilde U_2),h=1,...,h_0\big)$ as in \eqref{Theta^*-2},
giving error free ``estimates",
$$
U_1 =\text{LSVD}_{r_1}\big(\overline{\E}[\text{TIPUP}_1](\widehat U_2^{(m)})\big),\quad P_1 = U_1U_1^\top,
$$
when $\overline{\E}[\text{TIPUP}_1](\widehat U_2^{(m)})$ is of rank $r_1$.
Thus, by Wedin's theorem (\citet{wedin1972}),
\bes
L_1^{(m+1)} = \big\|\widehat P_1^{(m+1)} - P_1\big\|_{\rm S}
\le \frac{2\| \text{TIPUP}_1(\widehat U_2^{(m)})-\overline{\E}[\text{TIPUP}_1](\widehat U_2^{(m)}) \|_{\rm S}}
{\sigma_{r_1}(\overline{\E}[\text{TIPUP}_1](\widehat U_2^{(m)}))}.
\ees
We note that
$\widehat U_2^{(m)}$ is plugged-in after the conditional
expectation in $\overline{\E}[\text{TIPUP}_1](\cdot)$.
For general $1\le k\le K$, we define $\text{TIPUP}_k(\widetilde U_{-k})$
and $\overline{\E}[\text{TIPUP}_k](\widetilde U_{-k})$
as matrix-valued functions of $\widetilde U_{-k}=\odot_{j\neq k}\widetilde U_j$.
We will prove by induction that
$\sigma_{r_k}(\overline{\E}[\text{TIPUP}_k](\widehat U_{-k}^{(m)}))$,
the denominator in the above inequality, is
no smaller than a half of its ideal version as in \eqref{lam^*_k}, e.g.
\begin{equation}\label{signal-bd}
2 \sigma_{r_1}(\overline{\E}[\text{TIPUP}_1](\widehat U_2^{(m)}))
\ge \sigma_{r_1}(\overline{\E}[\text{TIPUP}_1](U_2)) = h_0^{1/2}\lam_1^{*2},
\end{equation}
in the case of $k=1$ and $K=2$. It would then follow that
\begin{align}\label{Wedin-72}
L_1^{(m+1)} = \big\|\widehat P_1^{(m+1)} - P_1\big\|_{\rm S}
\le \frac{\| \text{TIPUP}_1(\widehat U_2^{(m)})-\overline{\E}[\text{TIPUP}_1](\widehat U_2^{(m)}) \|_{\rm S}}
{ h_0^{1/2}\lam_1^{*2}/4 }.
\end{align}

To bound the numerator on the right-hand side of \eqref{Wedin-72}, we write
\begin{equation}\label{Delta_jh}
\text{TIPUP}_1(\widetilde U_2)-\overline{\E}[\text{TIPUP}_1](\widetilde U_2)
= \sum_{j=1}^3 \Big(\Delta^*_{j,1,h}\big(\widetilde U_2\widetilde U_2^\top\big), h=1,\ldots,h_0\Big)
\in \R^{d_1\times (d_1h_0)}
\end{equation}
as both $\text{TIPUP}_1(\widetilde U_2)$ and $\overline{\E}[\text{TIPUP}_1](\widetilde U_2)$ are linear in
$\widetilde U_2\widetilde U_2^\top$, where for any $\widetilde M_2\in \R^{d_2\times d_2}$
\begin{align*}
& \Delta_{1}^*(\widetilde M_2) := \Delta_{1,1,h}^*(\widetilde M_2)
:= \hbox{$\sum_{t=h+1}^T$}A_1 F_{t-h} A_2^\top \widetilde M_2 E_t^\top/(T-h),
\\ & \Delta_{2}^*(\widetilde M_2) := \Delta_{2,1,h}^*(\widetilde M_2) :=
\hbox{$\sum_{t=h+1}^T$} E_{t-h} \widetilde M_2 A_2 F_t^\top A_1^\top/(T-h),
\\\ & \Delta_{3}^*(\widetilde M_2) := \Delta_{3,1,h}^*(\widetilde M_2) :=
\hbox{$\sum_{t=h+1}^T$}  E_{t-h}\widetilde M_2 E_{t}^\top/(T-h).
\end{align*}
As $\Delta_{j,1,h}^*(\widetilde M_2)$ is linear in $\widetilde M_2$,
the numerator on the right-hand of \eqref{Wedin-72} can be bounded
by
\begin{align}\label{norm-bd}
& \| \text{TIPUP}_1(\widehat U_2^{(m)})-\overline{\E}[\text{TIPUP}_1](\widehat U_2^{(m)})\|_{\rm S}
\\ \notag
\le& \|\text{TIPUP}_1(U_2)-\overline{\E}[\text{TIPUP}_1](U_2)\|_{\rm S}
+ L^{(m)}(2K-2)
 \sum_{j=1}^3 
h_0^{1/2}\max_{h\le h_0}\|\Delta_{j,1,h}^*\|_{1,S,S}
\end{align}
with an application of Cauchy-Schwarz inequality for the sum over $h=1,\ldots,h_0$,
where $\|\Delta_{j,1,h}^*\|_{1,S,S}$ are norms of the
$\R^{d_2\times d_2}\to \R^{d_1\times d_1}$
linear mappings $\Delta_{j,1,h}^*$ defined as
$$
\|\Delta_{j}^*\|_{1,S,S} = \|\Delta_{j,1,h}^*\|_{1,S,S}
:= \max_{\|\widetilde M_2\|_{\rm S}\le 1, \text{rank}(\widetilde M_2)\le r_2}
\big\|\Delta_{j,1,h}^*(\widetilde M_2)\big\|_{\rm S}.
$$
For general $1\le k\le K$, $\Delta_{j,k,h}^*$ is an
$\R^{d_{-k}\times d_{-k}}$ to $\R^{d_k\times d_k}$ mapping,
and \eqref{lm-3-3} of Lemma \ref{lemma:epsilonnet} (iii)
gives the general version of \eqref{norm-bd} with 
$$
\big\|\Delta_{j,k,h}^*\big\|_{k,S,S}
:=\max_{\|\widetilde M_{\ell}\|_{\rm S}\le 1, \text{rank}(\widetilde M_{\ell})\le r_\ell,\forall \ell \neq k}
\big\|\Delta_{j,k,h}^*\big(\odot_{\ell \neq k}\widetilde M_{\ell}\big)\big\|_{\rm S}
$$
because it is applied to
$\widetilde M_{-k} = \odot_{\ell\neq k}\widehat P_\ell^{(m)}
- \odot_{\ell\neq k}P_\ell$ with
$\|\widehat P_\ell^{(m)} - P_\ell\|_{\rm S}\le L^{(m)}$.

We claim that in certain events $\Omega_j, j=1,2,3$, with
$\overline\P(\Omega_j)\ge 1- 5^{-1} \sum_{k=1}^Ke^{-d_k}$,
\begin{equation}\label{Delta-bd}
\big\|\Delta_{j,k,h}^*\big\|_{k,S,S} \le \rho \lam_k^{*2}\big/\big(24(K-1)\big),\quad \forall 1\le k\le K.
\end{equation}
For simplicity, we will only prove this inequality for $k=1$ and $K=2$ in the form of
\begin{equation}\label{Delta-bd-1}
\overline\P\left\{\|\Delta_{j}^*\|_{1,S,S}
= \max_{\|\widetilde M_2\|_{\rm S}\le 1, \text{rank}(\widetilde M_2)\le r_2 }
\big\|\Delta_{j,1,h}^*(\widetilde M_2)\big\|_{\rm S}  \ge \rho \lam_1^{*2}\big/24\right\}
\le 5^{-1}e^{- d_2}
\end{equation}
as the proof of its counter part for general $1\le k\le K$ is similar.

Define the ideal version of the ratio in \eqref{Wedin-72} for general $1\le k\le K$ as
\begin{equation*}
L^{\ideal}_k = \frac{\| \text{TIPUP}_k(U_{-k})-\overline{\E}[\text{TIPUP}_k](U_{-k}) \|_{\rm S}}
{ h_0^{1/2}\lambda_{k}^{*2}/4},
\quad L^{\ideal} = \max_{1\le k\le K}  L^{\ideal}_k.
\end{equation*}
As $U_{-k} = \odot_{j\neq k}U_j$ is true and deterministic, by \eqref{iTIPUP-ideal_k} and \eqref{Delta_jh}, the proof of \eqref{Delta-bd} also implies
\begin{equation}\label{ideal-bd}
L^{\ideal} \le C_{1}^{\tipup}R^{*\ideal}\ \text{ with }\ {R}^{*\ideal}=\max_{1\le k\le K}{R}^{*\ideal}_k
\end{equation}
in an event $\Omega_4$ with $\overline\P(\Omega_4)\le 5^{-1} \sum_{k=1}^Ke^{-d_k}$,
where $R^{*\ideal}_k$ is as in \eqref{iTIPUP-ideal_k}.

Putting together \eqref{L^m_k}, \eqref{Wedin-72}, 
\eqref{norm-bd} and \eqref{Delta-bd}, we find
that in the event $\cap_{j=0}^4 \Omega_j$
\begin{equation}\label{key}
L^{(m+1)}_k
\le L^{\ideal}_k +
L^{(m)}\frac{(6K-6) \max_{j,k,h}\big\|\Delta_{j,k,h}^*\big\|_{k,S,S}}{\lam_k^{*2}/4}
\le L^{\ideal}_k + \rho L^{(m)},
\end{equation}
which implies by induction
\begin{equation}\label{key-2}
L^{(m+1)}\le (1+\cdots+\rho^{m})
L^{\ideal} + \rho^{m+1}L^{(0)},
\end{equation}
and then the conclusions follows from \eqref{TIPUP_k-bd-tail} and \eqref{ideal-bd}.
We divide the rest of the proof into 4 steps to prove
\eqref{Delta-bd-1} for $j=1,2,3$ and \eqref{signal-bd}. 

\smallskip
\noindent\underline{\bf Step 1.}
We prove
\eqref{Delta-bd-1} for the $\Delta_1^*(\widetilde M_2)$ in \eqref{Delta_jh}.
By Lemma \ref{lemma:epsilonnet} (ii), there exist $\widetilde M^{(\ell,\ell')} \in \mathbb R^{d_2\times d_2}$ of the form $W_{\ell}Q_{\ell'}^\top$ with $W_{\ell}\in\R^{d_2\times r_2}$, $Q_{\ell'}\in\R^{d_2\times r_2}$, $1\le \ell,\ell' \le N_{d_2r_2,1/8} := 17^{d_2r_2}$, such that  $\|\widetilde M^{(\ell, \ell')}\|_{\rm S}\le 1$, rank$(\widetilde M^{(\ell, \ell')})\le r_2$ and
\begin{eqnarray*}
\|\Delta_{1}^*\|_{1,S,S}
= \|\Delta_{1,1,h}^*\|_{1,S,S}
\le 2\max_{\ell,\ell'\le N_{d_2r_2,1/8}}
\left\| \sum_{t=h+1}^T \frac{A_1 F_{t-h} A_2^\top \widetilde M^{(\ell, \ell')} E_t^\top} {T-h} \right\|_{\rm S}.
\end{eqnarray*}
To bound $\|\Delta_{1,k,h}^*\|_{1,S,S}$ for general $k\le K$, we just need to replace $\widetilde M^{(\ell, \ell')}$ by $\odot_{j\neq k} \widetilde M_j^{(\ell,\ell')}$
and $N_{d_2r_2,1/8}$ by $N_{d^*_{-k},1/(8K-8)}$ with $d^*_{-k}=\sum_{j\neq k}d_jr_j$
as in Lemma \ref{lemma:epsilonnet} (iii).
We apply the Gaussian concentration inequality to the right-hand side above.
Elementary calculation shows that
\begin{eqnarray*}
&&\left| \left\| \sum_{t=h+1}^T A_1 F_{t-h} A_2^\top \widetilde M^{(\ell, \ell')} E_t^\top \right\|_{\rm S} - \left\| \sum_{t=h+1}^T A_1 F_{t-h} A_2^\top \widetilde M^{(\ell, \ell')} E_t^{*\top} \right\|_{\rm S} \right| \\
&\le& \left\| (A_1F_1A_2^\top,...,A_1 F_{T-h}A_2^\top) \begin{pmatrix}
\widetilde M^{(\ell, \ell')}(E_{h+1}^\top-E_{h+1}^{*\top}) \\
\vdots\\
\widetilde M^{(\ell, \ell')}(E_{T}^\top-E_{T}^{*\top})
\end{pmatrix} \right\|_{\rm S} \\
&\le& \left\| 
(A_1F_1A_2^\top,...,A_1 F_{T-h}A_2^\top)
\right\|_{\rm S}^{1/2}
\left\| \text{diag}\Big(\widetilde M^{(\ell, \ell')}\Big)
\begin{pmatrix} E_{h+1}^\top-E_{h+1}^{*\top} \\
\vdots\\
E_{T}^\top-E_{T}^{*\top}
\end{pmatrix}\right\|_{\rm S} \\
&\le& \sqrt{T} \|\Theta_{1,0}^*\|_{\rm S}^{1/2} \left\|\begin{pmatrix}
E_{h+1}^\top-E_{h+1}^{*\top} \\
\vdots\\
E_{T}^\top-E_{T}^{*\top}
\end{pmatrix} \right\|_{\rm F}.
\end{eqnarray*}
That is, $\left\| \sum_{t=h+1}^T A_1 F_{t-h} A_2^\top \widetilde M^{(\ell, \ell')} E_t^\top \right\|_{\rm S}$ is a $\sqrt{T}\|\Theta_{1,0}^*\|_{\rm S}^{1/2}$ Lipschitz function of $(E_1,\ldots,E_T)$. Employing similar arguments in the proof of Theorem 2 in \citet{chen2022factor}, we have
\begin{eqnarray*}
\bar{\E} \left\| \sum_{t=h+1}^T \frac{A_1 F_{t-h} A_2^\top \widetilde M^{(\ell, \ell')} E_t^\top}{T-h} \right\|_{\rm S} \le \frac{\sigma(8Td_1)^{1/2}}{T-h} \|\Theta_{1,0}^*\|_{\rm S}^{1/2}.
\end{eqnarray*}
Then, by Gaussian concentration inequalities for Lipschitz functions,
\begin{eqnarray*}
\P\left( \left\| \sum_{t=h+1}^T \frac{A_1 F_{t-h} A_2^\top \widetilde M^{(\ell, \ell')} E_t^\top} {T-h} \right\|_{\rm S}
- \frac{\sigma(8Td_1)^{1/2}}{T-h} \|\Theta_{1,0}^*\|_{\rm S}^{1/2}
\ge \frac{\sigma\sqrt{T} }{T-h} \|\Theta_{1,0}^*\|_{\rm S}^{1/2} x\right)\le e^{-\frac{x^2}{2 } }.
\end{eqnarray*}
Hence,
\begin{eqnarray*}
\P\left( \|\Delta_1^*\|_{1,S,S}/2
\ge \frac{\sigma(8Td_1)^{1/2}}{T-h} \|\Theta_{1,0}^*\|_{\rm S}^{1/2}
+ \frac{\sigma\sqrt{T} }{T-h} \|\Theta_{1,0}^*\|_{\rm S}^{1/2} x\right)
\le N_{d_2r_2,1/8}^2 e^{-\frac{x^2}{2} }.
\end{eqnarray*}
As $T\ge 4h_0$ and $K=2$, this implies with $x\asymp \sqrt{d_2r_2}$ that in an event
with at least probability $1-e^{-d_2}/5$,
\begin{eqnarray*}
\|\Delta_1^*\|_{1,S,S}
\le \frac{C_{1,K}^{\iter}\sigma T^{-1/2} \|\Theta_{1,0}^*\|_{\rm S}^{1/2} (\sqrt{d_1}+\sqrt{d_2r_2})}{24(K-1)}
\le \frac{C_{1,K}^{\iter}(1+\sqrt{d_2r_2/d_1})\lambda_1^{*2}R_1^{*\ideal}}{24},
\end{eqnarray*}
with the ${R}_k^{*\ideal}$ in \eqref{iTIPUP-ideal_k} and a constant $C_{1,K}^{\iter}$ depending on $K$ only.
In this event, \eqref{condition1} gives
$24 (\lam^*_k)^{-2}\|\Delta_{1,k,h}^*\|_{k,S,S} \le C_{1,K}^{\iter} ({R}_k^{*\ideal}+R_k^{*\add}) \le \rho $.
Thus, \eqref{Delta-bd-1} holds for $\Delta_1^*(\widetilde M_2)$.

\smallskip
\noindent\underline{\bf Step 2.}
Inequality \eqref{Delta-bd-1} for $\Delta_2^*(\widetilde M_2)$ follow from the same argument as the above step.

\smallskip
\noindent\underline{\bf Step 3.}
Here we prove \eqref{Delta-bd-1} for the $\Delta_3^*(\widetilde M_2)$ in \eqref{Delta_jh}. By Lemma \ref{lemma:epsilonnet} (ii), we can find
$U_2^{(\ell)},U_2^{(\ell')}\in\mathbb R^{d_2\times r_2}$,
$1\le \ell,\ell'\le N_{d_2r_2,1/8}$ such that $\|U_2^{(\ell)}\|_{\rm S}\le 1$, $\|U_2^{(\ell')}\|_{\rm S}\le 1$ and
\begin{equation}\label{pf-th2-step3-1}
\|\Delta_3^*\|_{1,S,S} = \|\Delta_{3,1,h}^*\|_{1,S,S}
\le 2 \max_{1\le \ell, \ell' \le N_{d_2r_2,1/8}}
\left\| \sum_{t=h+1}^T \frac{E_{t-h} U_2^{(\ell)} U_2^{(\ell')\top} E_t^\top} {T-h} \right\|_{\rm S}.
\end{equation}
We split the sum into two terms over the index sets, $S_1= \{ (h,2h] \cup (3h,4h] \cup \cdots \} \cap(h,T]$ and its complement $S_2$ in $(h,T]$, so that $\{E_{t-h},t\in S_a\}$ is independent of $\{E_t, t\in S_a\}$ for each $a=1,2$. Let $n_a=|S_a|r_2$.
Define $G_a=(E_{t-h}U_2^{(\ell)},t\in S_a)\in \R^{d_1\times n_a}$ and
$H_a=(E_{t}U_2^{(\ell')},t\in S_a)\in \R^{d_1\times n_a}$. Then, $G_a$, $H_a$ are two independent Gaussian matrices. Note that
\begin{equation}\label{pf-th2-step3-2}
\left\| \sum_{t\in S_a} \frac{E_{t-h} U_2^{(\ell)} U_2^{(\ell')\top} E_t^\top} {T-h}\right\|_{\rm S} =\left\|\frac{G_a H_a^\top}{T-h} \right\|_{\rm S}.
\end{equation}
Moreover, by Assumption \ref{asmp:error}, $\Var(u^\top \text{vec}(G_a))\le \sigma^2$ and
$\Var(u^\top \text{vec}(H_a))\le \sigma^2$ for all unit vectors $u\in\R^{d_1n_a}$,
so that by Lemme \ref{lm-GH} (i),
\begin{eqnarray*}
\P\Big\{ \|G_aH_a^\top\|_{\rm S}/\sigma^2\ge
d_1+2\sqrt{d_1n_a}
+ x(x+2\sqrt{n_a}+2\sqrt{d_1})
\Big\}
\le e^{-x^2/2} ,\quad x>0.
\end{eqnarray*}
As $\sum_{a=1}^2 n_a=r_2(T-h)$, it follows from \eqref{pf-th2-step3-1}, \eqref{pf-th2-step3-2} and the above inequality that
\begin{eqnarray*}
\P\left\{ \frac{\|\Delta_{3,1,h}^*\|_{1,S,S}}{4\sigma^2} \ge
\frac{(\sqrt{d_1}+x)^2}{T-h} + \frac{\sqrt{2r_2}(\sqrt{d_1}+x)}{\sqrt{T-h}}
\right\}
\le 2 N_{d_2r_2,1/8}^2 e^{-x^2/2}.
\end{eqnarray*}
Thus, with $h_0\le T/4$, $x\asymp \sqrt{d_1}+\sqrt{d^*_{-1}}$ and some constant $C_{1,K}^{\iter}$ depending on $K$ only,
\begin{align}\label{event*}
\|\Delta_{3,1,h}^*\|_{1,S,S}
\le \frac{(C_{1,K}^{\iter}-1) (1-h_0/T)^2 \sigma^2} {24(K-1)} \bigg(\frac{(\sqrt{d_1}+\sqrt{d^*_{-1}})\sqrt{r_{-1}}} {(T-h_0)^{1/2} }
+\frac{(\sqrt{d_1}+\sqrt{d^*_{-1}})^2}{T-h_0}\bigg)
\end{align}
with at least probability $1-e^{- d_1}/5$.
As $\lam_k^*\le \|\Theta_{1,0}^*\|_{\rm S}^{1/2}/(1-h_0/T)^{1/2}$ by \eqref{lam^*_k} and \eqref{Theta^*},
\begin{align*}
{R}_1^{*\ideal}
\ge (\lam_1^*)^{-1} (1-h_0/T)\sigma (T-h_0)^{-1/2}
\sqrt{d_1} + (\lam_1^*)^{-2} \sigma T^{-1/2} \sigma\sqrt{d_1r_{-1}}
\end{align*}
by \eqref{iTIPUP-ideal_k}. Thus, in the event \eqref{event*} and for $k=1$ and $K=2$,
\begin{align*}
& \|\Delta_{3,1,h}^*\|_{1,S,S} \\
\le&\frac{ C_{1,K}^{\iter}-1 }{24} \bigg(\Big(1+\sqrt{d^*_{-1}/d_1}\Big)\lam_1^{*2}{R}_1^{*\ideal}
+ \Big(1+\sqrt{d^*_{-1}/d_1}\Big)^2\lam_1^{*2}\big({R}_1^{*\ideal}\big)^2 \bigg) \\
\le & (\lam_1^{*2}/24)(C_{1,K}^{\iter}-1)(1+{R}_k^{*\ideal}+R_k^{*\add})  ({R}_k^{*\ideal}+R_k^{*\add})
\end{align*}
which is no greater than $\lam_1^{*2}\rho/24$ by the condition
\eqref{condition1}. This yields  \eqref{Delta-bd-1} for $\Delta_3^*(\widetilde M_2)$.

\noindent\underline{\bf Step 4.} Next, we prove \eqref{signal-bd} in the event $\cap_{j=0}^4\Omega_j$. Note that,
\begin{eqnarray*}
&&\|\Theta_{1,h}^*(\widehat U_2^{(m)}) -\Theta_{1,h}^*(U_2)\|_{\rm S} = \left\|\sum_{t=h+1}^T \frac{A_1 F_{t-h} A_2^\top (\widehat U_2^{(m)}\widehat U_2^{(m)\top} -U_2 U_2^\top) A_2  F_t^\top A_1^\top} {T-h} \right\|_{\rm S} \\
&&= \frac{1}{T-h}\left\|(A_1F_1A_2^\top,...,A_1F_{T-h}A_2^\top) \begin{pmatrix}
(\widehat U_2^{(m)}\widehat U_2^{(m)\top} -U_2 U_2^\top) A_2  F_1^\top A_1^\top\\
\vdots\\
(\widehat U_2^{(m)}\widehat U_2^{(m)\top} -U_2 U_2^\top) A_2  F_{T-h}^\top A_1^\top
\end{pmatrix} \right\|_{\rm S} \\
&&\le  \frac{1}{T-h}\left\|(A_1F_1A_2^\top,...,A_1F_{T-h}A_2^\top) \right\|_{\rm S}
\left\| \text{diag}\Big(\widehat U_2^{(m)}\widehat U_2^{(m)\top} -U_2 U_2^\top\Big)
\begin{pmatrix}
 A_2  F_1^\top A_1^\top\\
\vdots\\
A_2  F_{T-h}^\top A_1^\top
\end{pmatrix} \right\|_{\rm S} \\
&&\le \| \widehat U_2^{(m)}\widehat U_2^{(m)\top} -U_2 U_2^\top \|_{\rm S} \| \Theta_{1,0}^* \|_{\rm S} /(1-h/T) \\
&&\le L^{(m)} \| \Theta_{1,0}^* \|_{\rm S}/(1-h_0/T).
\end{eqnarray*}
Hence, by the Cauchy-Schwarz inequality and \eqref{error-free-update},
\begin{eqnarray*}
\| \overline{\E}[\text{TIPUP}_1](\widehat U_2^{(m)}) - \overline{\E}[\text{TIPUP}_1](U_2)\|_{\rm S}
\le \sqrt{h_0}L^{(m)} \| \Theta_{1,0}^* \|_{\rm S}/(1-h_0/T).
\end{eqnarray*}
By \eqref{lam^*_k},
$\lam_1^{*2}h_0^{1/2} =  \sigma_{r_1}\big(\mat1(\Theta_{1,1:h_0}^*)\big)
= \sigma_{r_1}\big(\overline{\E}[\text{TIPUP}_1](U_2)\big)$.
Thus, by Weyl's inequality,
\begin{eqnarray*}
\sigma_{r_1}(\overline{\E}[\text{TIPUP}_1](\widehat U_2^{(m)}))
\ge \lam_1^{*2}h_0^{1/2}  -2\sqrt{h_0}L^{(m)} \| \Theta_{1,0}^* \|_{\rm S}
\ge \sigma_{r_1}(\overline{\E}[\text{TIPUP}_1](U_2))\big/2 = \lam_1^{*2}h_0^{1/2}/2.
\end{eqnarray*}
when $\min_{k\le K} (\lam_k^{*2}/\| \Theta_{k,0}^* \|_{\rm S} )\ge 4L^{(m)} $.
We prove this condition by induction in the event $\cap_{j=0}^4\Omega_j$.
By \eqref{condition1} and \eqref{TIPUP_k-bd-tail},
$4L^{(0)}\|\Theta_{1,0}^* \|_{\rm S} \le 4C_1^{\tipup}R^{*(0)}\|\Theta_{1,0}^* \|_{\rm S}\le \min_{k\le K}\lam_k^{*2}$.
Given the induction assumption $4L^{(m)} \| \Theta_{1,0}^* \|_{\rm S}\le \min_{k\le K}\lam_k^{*2}$,
\eqref{signal-bd} holds for the same $m$, so that \eqref{key-2}, \eqref{ideal-bd} and \eqref{TIPUP_k-bd-tail},
\begin{align*}
L^{(m+1)}
\le C_1^{\tipup}\big\{2(1 +\cdots+ \rho^{m})R^{*\ideal} + \rho^{m+1} R^{*(0)}\big\}
\le  C_1^{\tipup}2(1- \rho)^{-1} R^{*(0)}.
\end{align*}
It then follows from \eqref{condition1} that
$4L^{(m+1)}\| \Theta_{1,0}^* \|_{\rm S} \le C_1^{\tipup}8(1- \rho)^{-1}R^{*(0)}
\| \Theta_{1,0}^* \|_{\rm S}\le \min_{k\le K}\lam_k^{*2}$.
This completes the induction and the proof of the entire theorem.

\section{Proof of Theorem \ref{thm:itopup}} \label{proofth1}

Again, we focus on the case of $K=2$ as the iTOPUP also begins with mode-$k$ matrix unfolding. In particular, we sometimes give explicit expressions only in the case of $k=1$ and $K=2$. For $K=2$, we observe a matrix time series with $X_t=A_1 F_t A_2^\top + E_t \in \mathbb R^{d_1\times d_2}$.
Recall $\overline\E(\cdot)=\E(\cdot|\{ \cF_1,...,\cF_T\})$.
Let $U_1$, $U_2$ be the left singular matrices of $A_1$ and $A_2$ respectively with $r_k=$rank$(U_k)=$rank$(A_k)$. Recall $\odot$ is kronecker product and $\otimes$ is tensor product.

We outline the proof as follows, which has exactly the same structure as the proofs of Theorem \ref{thm:itipup}.

Let $L^{(m)}_k=\|\widehat P_k^{(m)} - P_k\|_{\rm S}$ and $L^{(m)}=\max_{k\le K}L^{(m)}_k$ as in \eqref{L^m_k}. From Proposition \ref{prop:topup}, $\overline\E\big[L^{(0)}_k\big]\lesssim {R}^{\topup}_k$
as we mentioned in \eqref{bound:topup0}. 
By applying the Gaussian concentration inequality for Lipschitz functions and Lemme \ref{lm-GH}
in their analysis, we have
\begin{equation}\label{TOPUP_k-bd-tail}
L^{(0)} \le C_{1}^{\topup}R^{\topup}/2 \le C_{1}^{\topup}R^{(0)} \ \text{ with }\ {R}^{\topup}=\max_{1\le k\le K}{R}^{\topup}_k, 
\end{equation}
in an event $\Omega_0$ with $\overline\P(\Omega_0)\ge 1- 8^{-1} \sum_{k=1}^Ke^{-d_k}$.
This is similar to \eqref{ideal-bd:top} below.

For any matrices ${\widetilde U}_j\in \R^{d_j\times r_j}$ with $r_{K+j}=r_j$, define
\begin{align} \label{eq:itopup:def}
& {\rm{TOPUP}}_k({\widetilde U}_j, 1\le j\le 2K, j\neq k, j\neq K+K) \notag\\
=& \left(
{\rm{mat}}_k\Big(\widehat \Sigma_h\times_{j=1}^{k-1}{\widetilde U}_j \times_{j=k+1}^{K+k-1}{\widetilde U}_j
\times_{j=K+k+1}^{2K}{\widetilde U}_j\Big)
\ h=1,...,h_0 \right)
\end{align}
and its noiseless version $\overline{\E}[{\rm{TOPUP}}_k]$. Here both
${\rm{TOPUP}}_k$ and $\overline{\E}[{\rm{TOPUP}}_k]$ are viewed as
$\R^{d_k\times (d_kr_{-k}^2h_0)}$-valued $(2K-2)$-linear mappings
with input ${\widetilde U}_j$. When ${\widetilde U}_{K+j} = {\widetilde U}_j$ for all $j\neq k$,
we write \eqref{eq:itopup:def} as ${\rm{TOPUP}}_k({\widetilde U}_{-k})$ and
its noiseless version $\overline{\E}[{\rm{TOPUP}}_k]({\widetilde U}_{-k})$.
After the initialization with $\widehat U_k^{(0)}$,
the algorithm iteratively produces estimates $\widehat U_k^{(m+1)}$
as the rank-$r_k$ left singular matrix of
${\rm{TOPUP}}_k\big(\widehat U_{1:(k-1)}^{(m+1)},\widehat U_{(k+1):K}^{(m)}\big)$.
For $K=2$ and $k=1$,
\begin{equation*}
\text{TOPUP}_1(\widetilde U_2, \widetilde U_4)
= \text{mat}_1\bigg(\sum_{t=h+1}^T \frac{X_{t-h}\widetilde U_2 \otimes X_{t} \widetilde U_4}{T-h}, h=1,\ldots,h_0\bigg)
\in\R^{d_1\times (d_1r_2^2h_0)}.
\end{equation*}
for any $\widetilde U_j\in \R^{d_2\times r_2}, j=2,4$.
Given $\widehat U_2^{(m)}$, the $(m+1)$-th iteration produces estimates
\begin{align*}
\widehat U_1^{(m+1)}
= \text{LSVD}_{r_1}\big(\text{TOPUP}_1 \big(\widehat U_2^{(m)}\big)\big),\quad
\widehat P_1^{(m+1)} = \widehat U_1^{(m+1)}\widehat U_1^{(m+1)\top}.
\end{align*}
When $\text{rank}\big(\overline{\E}[\text{TOPUP}_1 ](\widehat U_2^{(m)})\big)=r_1$, the ``noiseless" version of this update is error free,
\begin{align*}
U_1 =\text{LSVD}_{r_1}\big(\overline{\E}[ \text{TOPUP}_1 ](\widehat U_2^{(m)})\big),\quad P_1 = U_1U_1^\top.
\end{align*}
Let $V\in \R^{d_1r_2^2h_0 \times r_1}$ be the right singular matrices associated with the top $r_1$ singular values of $\overline{\E}[ \text{TOPUP}_1 ](\widehat U_2^{(m)})$. Write
\begin{align*}
{\rm error}_1 =&  \frac{2\big\| \text{TOPUP}_1 (\widehat U_2^{(m)})-\overline{\E}[ \text{TOPUP}_1 ](\widehat U_2^{(m)}) \big\|_{\rm S}}
{\sigma_{r_1}(\overline{\E}[ \text{TOPUP}_1 ](\widehat U_2^{(m)}))} ,\\
{\rm error}_2 =&  \frac{2\Big\| \left(\text{TOPUP}_1 (\widehat U_2^{(m)})-\overline{\E}[ \text{TOPUP}_1 ](\widehat U_2^{(m)})\right)V \Big\|_{\rm S}}
{\sigma_{r_1}(\overline{\E}[ \text{TOPUP}_1 ](\widehat U_2^{(m)}))} .
\end{align*}
Thus, as $\text{TOPUP}_1(\widehat U_2^{(m)})$ is a rectangular matrix, by Lemma \ref{lm-pertubation}, 
\begin{align*}
L_1^{(m+1)} = \big\|\widehat P_1^{(m+1)} - P_1\big\|_{\rm S}
\le {\rm error}_2 + {\rm error}_1^2.
\end{align*}
We will prove by induction that
$\sigma_{r_k}\big(\overline{\E}[\text{TOPUP}_k]\big(\widehat U_{1:(k-1)}^{(m)},\widehat U_{(k+1):K}^{(m-1)}\big)\big)$,
the general version of the denominator on the right-hand side above,
is no smaller than a half of its ideal version as in \eqref{lam_k}, e.g.
\begin{equation}\label{signal-bd:top}
2 \sigma_{r_1}(\overline{\E}[ \text{TOPUP}_1 ](\widehat U_2^{(m)}))
\ge \sigma_{r_1}(\overline{\E}[ \text{TOPUP}_1 ](U_2)) = h_0^{1/2}\lam_1^{2},
\end{equation}
in the case of $k=1$ and $K=2$.
It would then follow that
\begin{align}\label{Wedin-72:top}
L_1^{(m+1)} = \big\|\widehat P_1^{(m+1)} - P_1\big\|_{\rm S}
\le & \frac{\Big\| \left(\text{TOPUP}_1 (\widehat U_2^{(m)})-\overline{\E}[ \text{TOPUP}_1 ](\widehat U_2^{(m)})\right)V \Big\|_{\rm S}}
{ h_0^{1/2}\lam_1^{2}/4 }  \\
&+ \left( \frac{\| \text{TOPUP}_1 (\widehat U_2^{(m)})
- \overline{\E}[ \text{TOPUP}_1 ](\widehat U_2^{(m)}) \|_{\rm S}}
{ h_0^{1/2}\lam_1^{2}/4 } \right)^2. \notag
\end{align}

To bound the numerator on the right-hand side of \eqref{Wedin-72:top},
we write
\begin{align*}
\Delta\big(\widetilde U_2, \widetilde U_4\big) = \text{TOPUP}_1\big(\widetilde U_2, \widetilde U_4\big)
- \overline{\E}\big[\text{TOPUP}_1\big]\big(\widetilde U_2, \widetilde U_4\big)
\end{align*}
and notice that for any orthonormal matrices $Q_j \in \R^{d_2\times r_2}$
\begin{align*}
\big\|\Delta\big(Q_2 Q_2^\top\widetilde U_2, Q_4Q_4^\top\widetilde U_4\big)\big\|_{\rm S}
\le \big\|\Delta\big(Q_2, Q_4\big)\big\|_{\rm S}\|Q_2^\top \widetilde U_2\|_{\rm S}\|Q_4^\top\widetilde U_4\|_{\rm S}
\end{align*}
as the outer product is taken in \eqref{eq:itopup:def}. Moreover, because
$\Delta\big(\widetilde U_2, \widetilde U_4\big)$ is bilinear,
\begin{align*}
\big\|\Delta\big(\widetilde U_2,\widetilde U_2\big)\big\|_{\rm S}
&\le \big\|\Delta\big(P_2\widetilde U_2, P_2\widetilde U_2\big)\big\|_{\rm S}
+ \big\|\Delta\big(P_2\widetilde U_2, P_2^\perp\widetilde U_2\big)\big\|_{\rm S}
+ \big\|\Delta\big(P_2^\perp\widetilde U_2, \widetilde U_2\big)\big\|_{\rm S} \\
&\le \big\|\Delta\big(U_2, U_2\big)\big\|_{\rm S}
+ \big\|\Delta\big(U_2, Q_2\big)\big\|_{\rm S}\|P_2^\perp\widetilde U_2\|_{\rm S}
+ \big\|\Delta\big(Q_2, \widetilde U_2\big)\big\|_{\rm S}\|P_2^\perp\widetilde U_2\|_{\rm S},
\end{align*}
where $Q_2 =$LSVD$_{r_k}(P_2^\perp\widetilde U_2)$ and $P_2^\perp = I_{d_2}-P_2$.
Thus, due to $\|P_2^\perp \widehat U_2^{(m)}\|_{\rm S}\le \|P_2 - \widehat P_2^{(m)}\|_{\rm S}\le L^{(m)}$,
the numerator of the second term on the right-hand of \eqref{Wedin-72:top} can be bounded
by
\begin{align}\label{norm-bd:top}
& \| \text{TOPUP}_1 (\widehat U_2^{(m)})-\overline{\E}[ \text{TOPUP}_1 ](\widehat U_2^{(m)})\|_{\rm S}
\\ \notag
\le& \| \text{TOPUP}_1 (U_2)-\overline{\E}[ \text{TOPUP}_1 ](U_2)\|_{\rm S}
+ L^{(m)} (2K - 2)
 \sum_{j=1}^3
h_0^{1/2}\max_{h\le h_0}\|\Delta_{j,1,h}\|_{1,S,S},
\end{align}
where $\|\Delta\|_{1,S,S} := \max_{\|\widetilde U_j\|_{\rm S}\le 1, \text{rank}(\widetilde U_j)= r_2, j=2,4}
\big\|\Delta(\widetilde U_2,\widetilde U_4)\big\|_{\rm S}$ for any bilinear $\Delta$ and
\begin{align*}
\Delta_{1}\big( \widetilde U_2,\widetilde U_4 \big) := \Delta_{1,1,h}\big( \widetilde U_2,\widetilde U_4 \big)
&:= \frac{1}{T-h}\sum_{t=h+1}^T \mat1( A_1 F_{t-h} A_2^\top \widetilde U_2 \otimes E_t  \widetilde U_4 ),
\\
\Delta_{2}\big( \widetilde U_2,\widetilde U_4 \big) := \Delta_{2,1,h}\big( \widetilde U_2,\widetilde U_4 \big)
&:=\frac{1}{T-h}\sum_{t=h+1}^T \mat1( E_{t-h} \widetilde U_2 \otimes A_1 F_t^\top A_2^\top  \widetilde U_4 ),
\\
\Delta_{3}\big( \widetilde U_2,\widetilde U_4 \big) := \Delta_{3,1,h}\big( \widetilde U_2,\widetilde U_4 \big)
&:= \frac{1}{T-h}\sum_{t=h+1}^T  \mat1( E_{t-h}\widetilde U_2 \otimes E_{t}  \widetilde U_4 ).
\end{align*}
We claim that in certain events $\Omega_j, j=1,2,3$, with
$\overline\P(\Omega_j)\ge 1- 8^{-1} \sum_{k=1}^Ke^{-d_k}$, $\rho<1$,
\begin{equation}\label{Delta-bd:top}
\big\|\Delta_{j,k,h}\big\|_{k,S,S} \le \rho \lam_k^{2}\big/\big( 48(K-1) \big),\quad \forall 1\le k\le K,
\end{equation}
in \eqref{norm-bd:top}.
Similarly, we claim that in certain events $\Omega_{j+3}, j=1,2,3$, with $\overline\P(\Omega_{j+3})\ge 1- 8^{-1} \sum_{k=1}^Ke^{-d_k}$, $\rho<1$,
\begin{equation}\label{Delta-bd:top2}
\big\|\Delta_{j,k,h}V\big\|_{k,S,S} \le \rho \lam_k^{2}\big/\big( 48(K-1) \big),\quad \forall 1\le k\le K.
\end{equation}
For simplicity, we will only prove this inequality for $k=1$ and  $K=2$,
\begin{equation}\label{Delta-bd-1:top}
\overline\P\left\{\|\Delta_{j}\|_{1,S,S}
= \max_{\|\widetilde M_2\|_{\rm S}\le 1}
\big\|\Delta_{j,1,h}\big( \widetilde U_2,\widetilde U_4 \big)\big\|_{\rm S}  \ge \rho \lam_1^{2}\big/ 48 \right\}
\le 8^{-1}e^{- d_2}.
\end{equation}

Define the ideal version of the ratio in \eqref{Wedin-72:top} for general $1\le k\le K$ as
\begin{align*}
L^{\ideal}_k  =& \frac{ 4 \left\| \left(\text{TOPUP}_k (U_{-k}) - \overline{\E}[ \text{TOPUP}_k ]( U_{-k} ) \right) V \right\|_{\rm S}}
{\sigma_{r_k}(\overline{\E}[ \text{TOPUP}_k ](U_{-k}))}
+ \left (\frac{ 4 \| \text{TOPUP}_k (U_{-k})  - \overline{\E}[ \text{TOPUP}_k ]( U_{-k} ) \|_{\rm S}} {\sigma_{r_k}(\overline{\E}[ \text{TOPUP}_k ](U_{-k}))} \right)^2, \\
L^{\ideal} =& \max_{1\le k\le K}  L^{\ideal}_k.
\end{align*}
Note that $\sigma_{r_k}(\overline{\E}[ \text{TOPUP}_k ]( U_{-k} )) = h_0^{1/2}\lambda_k^2$. By \eqref{iTOPUP-ideal_k}, the proof of \eqref{Delta-bd:top} and \eqref{Delta-bd:top2} also implies
\begin{equation}\label{ideal-bd:top}
L^{\ideal} \le C_{1}^{\topup}R^{\ideal}\ \text{ with }\ {R}^{\ideal}=\max_{1\le k\le K}{R}^{\ideal}_k
\end{equation}
in an event $\Omega_7$ with $\overline\P(\Omega_7)\le 8^{-1} \sum_{k=1}^Ke^{-d_k}$,
where $R^{\ideal}_k$ is as in \eqref{iTOPUP-ideal_k}.

It follows from \eqref{L^m_k}, \eqref{Wedin-72:top}, 
\eqref{norm-bd:top}, \eqref{Delta-bd:top} and \eqref{Delta-bd:top2} that in the event $\bigcap_{j=0}^7 \Omega_j$
\begin{align}\label{key:top}
L^{(m+1)}_k
\le&  L^{\ideal}_k + \frac{L^{(m)} \max_{j,k,h}\big\|\Delta_{j,k,h}V\big\|_{k,S,S}}{\lam_k^{2} /(24(K-1))}
+ \left( \frac{L^{(m)} \max_{j,k,h}\big\|\Delta_{j,k,h}\big\|_{k,S,S} }{\lam_k^{2} /(24(K-1) )}\right)^2  \\
\le& L^{\ideal}_k + (\rho L^{(m)}/2)^2 +\rho L^{(m)}/2 \le L^{\ideal}_k + \rho L^{(m)} , \notag
\end{align}
which implies by induction
\begin{equation}\label{key-2:top}
L^{(m+1)}\le (1+\cdots+\rho^{m}) L^{\ideal} + \rho^{m+1}L^{(0)},
\end{equation}
and then the conclusions follows from \eqref{TOPUP_k-bd-tail} and \eqref{ideal-bd:top}.
%
Again, we divide the rest of the proof into 4 steps to prove
\eqref{Delta-bd-1:top} for $j=1,2,3$ and \eqref{signal-bd:top}. 

\smallskip
\noindent\underline{\bf Step 1.}
We prove
\eqref{Delta-bd-1:top} for the $\Delta_1\big( \widetilde U_2,\widetilde U_4 \big)$.
By Lemma \ref{lemma:epsilonnet} (ii), there exist $U_2^{(\ell)},U_2^{(\ell')}\in\mathbb R^{d_2\times r_2}$, $1\le \ell,\ell' \le N_{d_2r_2,1/8} := 17^{d_2r_2}$, such that  $\| U_2^{(\ell)}\|_{\rm S}\le 1$, $\| U_2^{(\ell')}\|_{\rm S}\le 1$ and
\begin{eqnarray*}
\|\Delta_{1}\|_{1,S,S}
= \|\Delta_{1,1,h}\|_{1,S,S}
\le 2\max_{\ell,\ell'\le N_{d_2r_2,1/8}}
\left\| \sum_{t=h+1}^T \frac{ \mat1(A_1 F_{t-h} A_2^\top U_2^{(\ell)} \otimes E_t U_2^{(\ell')} )} {T-h} \right\|_{\rm S}.
\end{eqnarray*}
We apply the Gaussian concentration inequality to the right-hand side above. Elementary calculation shows that
\begin{eqnarray*}
&&\left| \left\| \sum_{t=h+1}^T \mat1(A_1 F_{t-h} A_2^\top U_2^{(\ell)} \otimes E_t U_2^{(\ell')} ) \right\|_{\rm S} - \left\| \sum_{t=h+1}^T \mat1(A_1 F_{t-h} A_2^\top U_2^{(\ell)} \otimes E_t^* U_2^{(\ell')} )  \right\|_{\rm S} \right| \\
&\le& \left\| \sum_{t=h+1}^T \mat1(A_1 F_{t-h}A_2^\top U_2^{(\ell)} \otimes (E_t-E_t^*) U_2^{(\ell')} )  \right\|_{\rm S} \\
&\le& \left\| (\mat1(A_1 F_1 A_2^\top\otimes I_{d_1}),...,\mat1(A_1 F_{T-h} A_2^\top \otimes I_{d_1})) \begin{pmatrix}
U_2^{(\ell)}\odot I_{d_1}\odot (E_{h+1}-E_{h+1}^{*})U_2^{(\ell')} \\
\vdots\\
U_2^{(\ell)}\odot I_{d_1}\odot(E_{T}-E_{T}^{*})U_2^{(\ell')}
\end{pmatrix} \right\|_{\rm S} \\
&\le& \sqrt{T} \|\Theta_{1,0}^*\|_{\rm S}^{1/2} \|U_2^{(\ell)}\|_{\rm S} \|U_2^{(\ell')}\|_{\rm S} \left\|\begin{pmatrix}
E_{h+1}-E_{h+1}^{*} \\
\vdots\\
E_{T}-E_{T}^{*}
\end{pmatrix} \right\|_{\rm F}.
\end{eqnarray*}
That is, $\left\| \sum_{t=h+1}^T \mat1(A_1 F_{t-h}A_2^\top U_2^{(\ell)} \otimes E_t U_2^{(\ell')} ) \right\|_{\rm S}$ is a $\sqrt{T}\|\Theta_{1,0}^*\|_{\rm S}^{1/2}$ Lipschitz function in $(E_1,...,E_T)$. Employing similar arguments in the proof of Theorem 1 in \citet{chen2022factor}, we have
\begin{eqnarray*}
\overline\E \left\| \sum_{t=h+1}^T \frac{\mat1(A_1 F_{t-h} A_2^\top U_2^{(\ell)} \otimes E_t U_2^{(\ell')} )}{T-h} \right\|_{\rm S} \le  \frac{\sigma(2T)^{1/2}(\sqrt{d_1}+\sqrt{d_1r_2^2})}{T-h} \|\Theta_{1,0}^*\|_{\rm S}^{1/2}
\end{eqnarray*}
Then, by Gaussian concentration inequalities for Lipschitz functions,
\begin{align*}
&\P\left( \left\| \sum_{t=h+1}^T \frac{\mat1(A_1 F_{t-h}A_2^\top U_2^{(\ell)} \otimes E_t U_2^{(\ell')} )}{T-h} \right\|_{\rm S} - \frac{\sigma(2T)^{1/2}(\sqrt{d_1}+\sqrt{d_1r_2^2})}{T-h} \|\Theta_{1,0}^*\|_{\rm S}^{1/2} \ge \frac{\sigma\sqrt{T} }{T-h} \|\Theta_{1,0}^*\|_{\rm S}^{1/2} x\right) \\
&\le 2e^{-\frac{x^2}{2 } }.
\end{align*}
Hence,
\begin{align*}
&\P\left( \|\Delta_{1}\|_{1,S,S}/2\ge \frac{\sigma(2T)^{1/2}(\sqrt{d_1}+\sqrt{d_1r_2^2})}{T-h} \|\Theta_{1,0}^*\|_{\rm S}^{1/2} + \frac{\sigma\sqrt{T} }{T-h} \|\Theta_{1,0}^*\|_{\rm S}^{1/2} x\right) \\ &\le 2N_{d_2r_2,1/8}^2 e^{-\frac{x^2}{2} }.
\end{align*}
As $T\ge 4h_0$, this implies with $x\asymp \sqrt{d_2r_2}$ that in an event
with at least probability $1-e^{-d_2}/8$,
\begin{align*}
\|\Delta_1\|_{1,S,S}
&\le \frac{C_{1,K}^{\iter}\sigma T^{-1/2} \|\Theta_{1,0}^*\|_{\rm S}^{1/2} (\sqrt{d_1}r_2+\sqrt{d_2r_2})}{ 96(K-1) } \\
&\le \frac{C_{1,K}^{\iter}(\lambda_1^{2}R_1^{\ideal}+\sigma T^{-1/2}\|\Theta_{1,0}^*\|_{\rm S}^{1/2}\sqrt{d_2r_2})}{ 96 }
\end{align*}
with the ${R}_k^{\ideal}$ in \eqref{iTOPUP-ideal_k} and a constant $C_{1,K}^{\iter}$ depending on $K$ only.
In this event, \eqref{condition1n} gives
$48 (\lam_k)^{-2}\|\Delta_{1,k,h}\|_{k,S,S} \le C_{1,K}^{\iter} ({R}_k^{\ideal}+R_k^{\add}) \le \rho$.
Thus, \eqref{Delta-bd-1:top} holds for $\Delta_1^*(\widetilde M_2)$.

\smallskip
\noindent\underline{\bf Step 2.}
Note that
\begin{align*}
\|\Delta_2\|_{1,S,S}= \max_{\substack{ \widetilde U_2\in \mathbb R^{d_2\times r_2}, \widecheck U_2\in \mathbb R^{d_2\times r_2},\\ \|\widetilde U_2\|_{\rm S}\le 1, \|\widecheck U_2\|_{\rm S}\le 1 } } \left\| \sum_{t=h+1}^T \frac{\mat1(E_{t-h} \widetilde U_2\otimes U_1^\top A_1 F_t A_2^\top \widecheck U_2 )} {T-h} \right\|_{\rm S} .
\end{align*}
Then, inequality \eqref{Delta-bd-1:top} for
$\Delta_2\big( \widetilde U_2,\widetilde U_4 \big)$ follow from the same argument as the above step.

\smallskip
\noindent\underline{\bf Step 3.}
Now we prove \eqref{Delta-bd-1:top} for the $\Delta_3\big( \widetilde U_2,\widetilde U_4 \big)$. We split the sum into two terms over the index sets, $S_1=\{(h,2h]\cup(3h,4h]\cup\cdots\} \cap(h,T]$ and its complement $S_2$ in $(h,T]$, so that $\{E_{t-h},t\in S_a\}$ is independent of $\{E_t, t\in S_a\}$ for each $a=1,2$. Let $n_a=|S_a|$.

By Lemma \ref{lemma:epsilonnet} (ii), we can find
$U_2^{(\ell)},U_2^{(\ell')}\in\mathbb R^{d_2\times r_2}$,
$1\le \ell,\ell'\le N_{d_2r_2,1/8}$ such that $\|U_2^{(\ell)}\|_{\rm S}\le 1$, $\|U_2^{(\ell')}\|_{\rm S}\le 1$. In this case,
\begin{equation}\label{pf-th1-step3-1}
\|\Delta_3\|_{1,S,S} = \|\Delta_{3,1,h}\|_{1,S,S}
\le 2 \max_{1\le \ell, \ell' \le N_{d_2r_2,1/8}}
\left\| \sum_{t=h+1}^T \frac{\mat1(E_{t-h} U_2^{(\ell)} \otimes E_t U_2^{(\ell')} )} {T-h} \right\|_{\rm S}.
\end{equation}
Define $G_a=(E_{t-h}U_2^{(\ell)},t\in S_a)$ and $H_a=(E_{t}U_2^{(\ell')},t\in S_a)$. Then, $G_a$, $H_a$ are two independent Gaussian matrices. By Lemma \ref{lm-GH}(ii), for any $x>0$,
\begin{align*}
&\P\left(\left\| \sum_{t\in S_a} \mat1(E_{t-h} U_2^{(\ell)} \otimes E_t U_2^{(\ell')} ) \right\|_{\rm S}  \ge d_1\sqrt{r_2}+2r_2 \sqrt{d_1 n_a}+x^2+\sqrt{n_a}x+3\sqrt{d_1r_2}x          \right) \\
&\le   2e ^{-x^2/2}.
\end{align*}
As in the derivation of $\|\Delta_3^*\|_{1,S,S}$ in the proof of Theorem \ref{thm:itipup}, we have, with $x\asymp \sqrt{d_2r_2}$ and some constant $C_{1,K}^{\iter}$ depending on $K$ only,
\begin{align*}
\P\left( \|\Delta_{3,1,h}\|_{1,S,S}
\ge \frac{C_{1,K}^{\iter}\sigma^2}{ 96 } \left(\frac{r_2\sqrt{d_1}+\sqrt{d_2 r_2}}{T^{1/2}} +\frac{d_1\sqrt{r_2}+d_2r_2+r_2\sqrt{d_1d_2}}{T} \right)
\right)
\le e^{- d_2}/8.
\end{align*}
This yields \eqref{Delta-bd-1:top} for $\Delta_3\big( \widetilde U_2,\widetilde U_4 \big)$
as in the end of Step 1 for $\Delta_1\big( \widetilde U_2,\widetilde U_4 \big)$.

\smallskip
\noindent\underline{\bf Step 4.}
Next, we consider the $r_1$-th largest singular value of $\sigma_{r_1}(\overline{\E} [ \text{TOPUP}_1 ](\widehat U_2^{(m)}) )$ in the event $\cap_{j=0}^7\Omega_j$. By definition, the left singular subspace of $\overline{\E} [ \text{TOPUP}_1 ](\widehat U_2^{(m)})$ is $U_1$. Then,
\begin{align*}
&\sigma_{r_1}(\overline{\E} [ \text{TOPUP}_1 ](\widehat U_2^{(m)}) ) \\
=& \sigma_{r_1}\left(\mat1\left(\sum_{t=h+1}^T \frac{A_1 F_{t-h} A_2^\top \widehat U_2^{(m)}\otimes A_1 F_t A_2^\top\widehat U_2^{(m)} } {T-h},h=1,...,h_0\right) \right)  \\
=& \sigma_{r_1}\left(\mat1\left(\sum_{t=h+1}^T \frac{A_1 F_{t-h} A_2^\top U_2U_2^\top \widehat U_2^{(m)}\otimes A_1 F_t A_2^\top U_2 U_2^\top \widehat U_2^{(m)} } {T-h},h=1,...,h_0\right) \right)  \\
=& \sigma_{r_1}\left(\mat1\left(\sum_{t=h+1}^T \frac{A_1 F_{t-h} A_2^\top U_2 \otimes A_1 F_t A_2^\top U_2 }{T-h},h=1,...,h_0 \right)\cdot\left(U_2^\top \widehat U_2^{(m)} \odot I_{d_1} \odot U_2^\top \widehat U_2^{(m)} \odot I_{h_0} \right) \right)  \\
\ge & \sigma_{r_1}\left(\mat1\left(\sum_{t=h+1}^T \frac{A_1 F_{t-h} A_2^\top U_2 \otimes A_1 F_t A_2^\top U_2 }{T-h},h=1,...,h_0 \right)\right)\cdot \sigma_{\min}\left(U_2^\top \widehat U_2^{(m)} \odot I_{d_1} \odot U_2^\top \widehat U_2^{(m)} \odot I_{h_0} \right)  \\
\ge & \sigma_{r_1}\left(\mat1\left(\sum_{t=h+1}^T \frac{A_1 F_{t-h} A_2^\top \otimes A_1 F_t A_2^\top }{T-h},h=1,...,h_0 \right)\right)\cdot \sigma_{\min}\left(U_2^\top \widehat U_2^{(m)} \right) \cdot \sigma_{\min}\left( U_2^\top \widehat U_2^{(m)} \right)  \\
\ge & \sqrt{h_0} \lambda_1^2 (1-L^{(m)2}). 
\end{align*}
The last step follows from
the definitions in \eqref{loss} and \eqref{L^m_k}.
If $L^{(m)}\le 1/2$, then $$\sigma_{r_1}(\overline{\E} [ \text{TOPUP}_1 ](\widehat U_2^{(m)}) ) \ge \sqrt{h_0}\lambda_1^2/2.$$
By \eqref{lam_k},
$\lam_1^{2}h_0^{1/2} =  \sigma_{r_1}\big({\text{mat}}_1(\Theta_{1,1:h_0})\big)
= \sigma_{r_1}\big({\text{mat}}_1(\overline{\E}[\text{TOPUP}_1](U_2)\big)$.
We prove this condition in the event $\cap_{j=0}^7\Omega_j$.
By \eqref{condition1n} and \eqref{TOPUP_k-bd-tail},
$L^{(0)}\le C_1^{\topup}R^{(0)} \le 1/2$.
By induction, given $L^{(m)} \le 1/2$,
\eqref{signal-bd:top} holds for the same $m$. Applying \eqref{TOPUP_k-bd-tail}, \eqref{ideal-bd:top} and \eqref{key-2:top},
\begin{align*}
L^{(m+1)}
&\le C_1^{\topup}\big\{2(1 +\cdots+ \rho^{m})R^{\ideal} + \rho^{m+1} R^{(0)}\big\}
\le  C_1^{\topup}2(1- \rho)^{-1} R^{(0)} \\
&\le 1/2.
\end{align*}
This completes the induction and the proof of the entire theorem.

\section{Proof of Theorem \ref{lowerbdd}} \label{proofth3}

Without loss of generality, we can assume $\sigma=1$. Within the probability space \eqref{eq:tenfm1sp} $\sP(T,d_1,...,d_K,\lambda)$, we study a specific model with $\cF_{2t}=\cF_{2t-1}$ and $\cF_{t,i_1,...,i_K}\overset{i.i.d.}{\sim} N(0,1)$ for all $1\le t\le \lfloor T/2 \rfloor,1\le i_k\le r_k$. Taking an average for each $\cX_{2t}$ and $\cX_{2t-1}$, $1\le t\le \lfloor T/2 \rfloor$, we reduce by sufficiency the model to
\begin{align*}
&\widetilde\sP(\lfloor T/2 \rfloor,d_1,...,d_K,\lambda) =\Big\{\cX_1,...,\cX_{\lfloor T/2 \rfloor}:\cX_t=\lambda \widetilde \cF_t\times_1  U_1\times_2...\times_K U_K+\widetilde \cE_t, \text{ with }  U_k\in\R^{d_k\times r_k},   \\
&\qquad\qquad U_k^\top U_k=I_{r_k}, 1\le k\le K, \widetilde \cF_{t,i_1,...,i_K} \overset{i.i.d.}{\sim} N(0,1/2), \{\widetilde \cF_t\}_{t=1}^{\lfloor T/2 \rfloor} \text{ independent of } \{\widetilde\cE_t\}_{t=1}^{\lfloor T/2 \rfloor},   \notag\\
&\qquad\qquad \widetilde \cE_{t,j_1,...,j_K}\overset{i.i.d.}{\sim} N(0,1/2) \text{ for all } 1\le t\le \lfloor T/2 \rfloor, 1\le i_k \le r_k, 1\le j_k\le d_k, 1\le k\le K \Big\}. \notag
\end{align*}
For notation convenience, in the following of this section, we assume $r_k=1$ for all $1\le k\le K$ and study the probability space
\begin{align}\label{eq:tenfm2sp}
&\widetilde\sP(T,d_1,...,d_K,\lambda) =\left\{\cX_1,...,\cX_{T}:\cX_t=\lambda f_t\times_1  a_1\times_2...\times_K a_K + \cE_t, \text{ with }  a_k\in\R^{d_k}, \right.  \\
&\qquad\qquad \| a_k\|_2=1, 1\le k\le K, f_t \overset{i.i.d.}{\sim} N(0,1), \{f_t\}_{t=1}^T \text{ independent of } \{\cE_t\}_{t=1}^T,   \notag\\
&\qquad\qquad \left. \cE_{t,j_1,...,j_K}\overset{i.i.d.}{\sim} N(0,1) \text{ for all } 1\le t\le T, 1\le j_k\le d_k, 1\le k\le K \right\}. \notag
\end{align}

We first introduce some additional notation. For any probability distributions $\P$ and $\Q$, define total variation distance as $\TV(\P,\Q)=\sup_B|\P(B)-\Q(B)|$. We also write $\TV(p,q)$ if $p,q$ are the densities of $\P$ and $\Q$, respectively. Define
\begin{align*}
\tau_N=\frac{\kappa}{N}, \qquad \eta_N =\frac{\kappa}{(K+2)^4N(\log N)^2},
\end{align*}
where $\kappa$ is the size of the clique. For any $\mu\in\R$, denote $\phi_\mu$ as the density function of the $N(\mu,1)$ distribution, and let
\begin{align*}
\widebar\phi_{\mu}=\frac12 (\phi_{\mu}+\phi_{-\mu})
\end{align*}
be the density function of the Gaussian mixture $\frac12 N(\mu,1) + \frac12 N(-\mu,1)$. We also define $\widetilde\Xi_0$ as a truncated normal distribution by the $N(0,1)$ distribution restricted on the interval $[-(K+2)\sqrt{\log N}, (K+2)\sqrt{\log N} ]$. For any $|\mu|\le (K+2)\sqrt{\eta_N\log N}$, define two distributions $\sF_{\mu,0}$ and $\sF_{\mu,1}$ with density functions
\begin{align*}
& h_{\mu,0}(x)= J_0 (\phi_0(x) -\tau_N^{-1}[\widebar\phi_{\mu}(x)-\phi_0(x)]) \mathbf{1}_{\{|x|\le (K+2)\sqrt{\log N} \}}, \\
& h_{\mu,1}(x)= J_1 (\phi_0(x) +\tau_N^{-1}[\widebar\phi_{\mu}(x)-\phi_0(x)]) \mathbf{1}_{\{|x|\le (K+2)\sqrt{\log N} \}},
\end{align*}
where $J_0,J_1$ are normalizing constants.

Suppose we have a collection of estimators $\widehat\ba=(\widehat a_1,...,\widehat a_K)$ with $\widehat a_j=\widehat a_j(\cX_1,...,\cX_T)$ being the estimated factor loading $ a_j$ ($1\le j\le K$). Our main technique is based on a reduction scheme which maps any order $K+1$ graph adjacent tensor $\cA\in \{0,1\}^{N^{\otimes (K+1)}}$ with dimension $N^{\otimes (K+1)}=N\times N\times\cdots\times N$ (multiply $N$ by $K+1$ times), $N\ge 2T$, and $\widehat\ba$, to a test for the Hypergraphic Planted Clique detection problem \eqref{hpc}. The specific technique was developed by  \cite{ma2015computational,gao2017sparse}. Based on the asymmetry of the tensor PCA instance, our proof is relatively more straightforward in establishing the computational lower bound, in contrast to the approach taken by \cite{brennan2020reducibility} in Section 11, which requires considerable effort in the symmetric-signal case. We refer the readers to \cite{wang2016statistical, cai2017computational} for other related methods. We provide a detailed description of the mapping as follows.


\begin{enumerate}[(1)]
\item (Initialization). Generate i.i.d. random variable $\xi_1,...,\xi_{2T}\sim\widetilde\Xi_0$. Set
\begin{align}\label{lowerbdd:initial}
\mu_t=\eta_N^{1/2}\xi_t, \qquad t=1,...,2T.
\end{align}

\item (Gaussianization). Generate two order $K+1$ tensors $\cB_0,\cB_1 \in \R^{2T^{\otimes (K+1)}}$, where conditioning on the 
$\mu_t$'s, all the entries are mutually independent satisfying
\begin{align}
\cL\big( (\cB_{0})_{t,j_1,...,j_K}|\mu_t\big)=\sF_{\mu_t,0}  \quad\text{and}\quad \cL\big( (\cB_{1})_{t,j_1,...,j_K}|\mu_t\big)=\sF_{\mu_t,1}.
\end{align}
Let $\cA_0\in\{0,1\}^{2T^{\otimes (K+1)}}$ be the lower-left corner block of the order $K+1$ tensor $\cA$. Generate an order $K+1$ tensor $\cX=(\cX_1,\cX_2,...,\cX_{2T})\in\R^{2T\times d_1\times d_2\times\cdots\times d_K}$, where for each $t\le 2T$, if $1\le j_1,...,j_K\le 2T$, set
\begin{align}
\cX_{t,j_1,...,j_K}= \big(1- (\cA_{0})_{t,j_1,...,j_K} \big) (\cB_{0})_{t,j_1,...,j_K} + (\cA_{0})_{t,j_1,...,j_K} (\cB_{1})_{t,j_1,...,j_K};
\end{align}
otherwise, let $\cX_{t,j_1,...,j_K}$ be an independent copy from N$(0,1)$.

\item (Test Construction). Let $\widehat a_k=\widehat a_k(\cX_1,...,\cX_T)$ be the estimator of the factor loading $ a_k$, for all $1\le k\le K$, by treating $\cX_1,...,\cX_T$ as data. All $\widehat a_k$ are normalized to be a unit vector. We reject $H_0^G$ if
\begin{align}\label{eq:testhpc}
\left(\frac1T\sum_{t=T+1}^{2T} \cX_t \otimes \cX_t \right) \times_{k=1}^{2K} \widehat a_k^\top \ge 1+ \frac12\left(\frac{\kappa}{2}\right)^K \eta_N,
\end{align}
with $\widehat a_{k+K}:= \widehat a_k$ for all $1\le k\le K$.

\end{enumerate}

\subsection{Lemmas}
To prove Theorem \ref{lowerbdd}, we need to state lemmas which characterize the distribution of $\cX_t$ under $H_0^G$ and $H_1^G$. Let $\cL(\{\cX_t\}_{t=1}^{2T})$ be the joint distribution of $\cX_1,...,\cX_{2T}$. Under probability space \eqref{eq:tenfm2sp}, we denote the distribution of $\cX_t$ as $\P_{\lambda,\ba}$ with $\ba=( a_1,..., a_K)$, and the joint distribution of $(\cX_1,...,\cX_{2T})$ as $\P_{\lambda,\ba}^{2T}$. Note that $\P_{\lambda,\ba}$ is a tensor (array) normal distribution with covariance tensor $\lambda\times_{k=1}^{2K} a_k+\cI_{d_1\times\cdots\times d_K\times d_1\times\cdots\times d_K }$, where $ a_{K+k}= a_k$ for $1\le k\le K$, $\cI_{j_1,...,j_K,j_1,...,j_K}=1$ and 0 elsewhere. If $\lambda=0$, all entries of $\cX_t$ are i.i.d. N(0,1), thus we denote $\cL(\cX_t)$ as $\P_0$. We also write $\cL(\cX|\beta)$ as the conditional distribution of $\cX | \beta$, and $\rmint \cL(\cX| \beta) \rmintd \xi(\beta)$ as the marginal distribution of $\cX$ after integrating $\beta$ out.

\begin{lemma}\label{lowerbdd:lemma1}
There exists some absolute constant $C>0$, such that for any integers $K\ge1$, $\kappa<N$, $N\ge 3$, and for all $|\mu|\le(K+2)\sqrt{\eta_N\log N}$,
\begin{equation*}
\TV(g_{\mu,0},\phi_0)\le CN^{-K-2}\qquad\text{and }\qquad \TV(g_{\mu,1},\bar\phi_\mu)\le CN^{-K-2}   ,
\end{equation*}
where $g_{\mu,0}=\frac12(h_{\mu,0}+h_{\mu,1})$ and $g_{\mu,1} =\tau_N h_{\mu,1}+ (1-\tau_N)\frac12(h_{\mu,0}+h_{\mu,1})$.
\end{lemma}

\begin{lemma}\label{lowerbdd:lemma2}
Suppose $\cA\sim\cG_{K+1}(N,1/2)$. There exists some constant $C_K>0$ depending on $K$ only, such that
\begin{equation*}
\TV(\cL(\{\cX_t\})_{t=1}^{2T},\P_0^{2T}) \le C_KN^{-1}.
\end{equation*}

\end{lemma}

The proofs of Lemma \ref{lowerbdd:lemma1} and \ref{lowerbdd:lemma2} are analogous to Lemma 7.1 and 7.2 in \cite{gao2017sparse}, thus are skipped here.

\begin{lemma}\label{lowerbdd:lemma3}
Suppose $\cA\sim\cG_{K+1}(N,1/2,\kappa)$. There exists a distribution $\pi$ supported on the set
\begin{align*}
\{(\lambda,\ba): \| a_k\|_2=1, |\text{supp}( a_k)|\le 3\kappa/2, 1\le k\le K, \eta_N^{1/2}(\kappa/2)^{K/2} \le \lambda\le \eta_N^{1/2}(3\kappa/2)^{K/2} \},
\end{align*}
such that for some positive constants $C_{1K},C_{2K}$ depending on $K$ only,
\begin{equation*}
\TV(\cL(\{\cX_t\})_{t=1}^{2T},\P_{\pi}) \le  C_{1K}\cdot\kappa \left(\frac{2T}{N} \right)^{\kappa} + \frac{C_{2K}}{N} + \frac{4(K+1)T}{N} ,
\end{equation*}
where $\P_{\pi}=\rmint\P_{\lambda,\ba}^{2T}\rmintd \pi(\lambda,\ba)$.
\end{lemma}

\begin{proof}
Let $\xi$ be $N(0,\eta_N)$, and $\bar\xi$ be a truncated normal distribution obtained by restricting $\xi$ on the set $[-(K+2)\sqrt{\eta_N\log N}, (K+2)\sqrt{\eta_N\log N}]$. Then the $\mu_i$'s in \eqref{lowerbdd:initial} are i.i.d. following $\bar\xi$. Elementary calculation shows that $\rmint\phi_0(x) \rmintd \xi(\mu)=\phi_0(x)$ gives the density function of $N(0,1)$, and $\rmint \bar\phi_{\mu}(x)\rmintd\xi(\mu)$ is the density function of $N(0,1+\eta_N)$.

We first consider the case $d_1=d_2=\cdots=d_K=2T$. Let $(\alpha_1,...,\alpha_{2T})$ be the 0-1 indicators of the first tensor mode of $\cA_0$ whether the corresponding vertices belong to the planted clique or not. Similarly, define $(\beta_{k1},...,\beta_{k,2T})$ as the corresponding indicators of the $(k+1)$-th tensor mode of $\cA_0$, for all $1\le k\le K$. Let $(\widetilde\alpha_1, ..., \widetilde\alpha_{2T})$ and $(\widetilde\beta_{k1},...,\widetilde\beta_{k,2T})$, $1\le k\le K$, as i.i.d. Bernoulli random variables with mean $\tau_N=\kappa/N$. Define a new tensor $\widetilde\cA_0$ with $\widetilde\cA_{t,j_1,...,j_K}=1$ if $\widetilde\alpha_{t}=\widetilde\beta_{1,j_1}=\cdots=\widetilde\beta_{K,j_K}=1$ and is an independent instantiation of the Bernoulli(1/2) distribution otherwise. Define $\widetilde\cX$ as
\begin{align*}
\widetilde\cX_{t,j_1,...,j_K}= \big(1- (\widetilde\cA_{0})_{t,j_1,...,j_K} \big) (\cB_{0})_{t,j_1,...,j_K} + (\widetilde\cA_{0})_{t,j_1,...,j_K} (\cB_{1})_{t,j_1,...,j_K} .
\end{align*}
By Theorem 4 in \cite{diaconis1980finite} and the data-processing inequality, we have
\begin{align*}
\TV(\cL(\cX),\cL(\widetilde\cX)) \le \TV(\cL(\alpha, \beta_1,..., \beta_K), \cL(\widetilde\alpha, \widetilde\beta_1,..., \widetilde\beta_K)) \le \frac{4(K+1)T}{N} ,
\end{align*}
where $\alpha=(\alpha_1,...,\alpha_{2T})$, $\beta_k=(\beta_{k1},...,\beta_{k,2T})$, $1\le k\le K$, and $\widetilde\alpha, \widetilde\beta_k$ are similarly defined. Note that, conditioning on $\mu_t$ and $\widetilde\beta_{1,j_k}=0$ for some $1\le k\le K$, $\widetilde\cX_{t,j_1,...,j_K}\sim g_{\mu_t,0}$. And conditioning on $\mu_t$ and $\widetilde\beta_{1,j_1}=\cdots=\widetilde\beta_{1,j_K}=1$, $\widetilde\cX_{t,j_1,...,j_K}\sim g_{\mu_t,1}$.

Next, define $\widebar\cX$ with entries
\begin{align*}
&\widebar\cX_{t,j_1,...,j_K} | (\widetilde\beta_{1,j_k}=0,\text{ for some } 1\le k\le K,\mu_t)\sim \phi_0 , \\
&\widebar\cX_{t,j_1,...,j_K} | (\widetilde\beta_{1,j_1} = \cdots = \widetilde\beta_{K,j_K} =1, \mu_t)\sim \bar\phi_{\mu_t} .
\end{align*}
By Lemma \ref{lowerbdd:lemma1} and Lemma 7 in \cite{ma2015computational}, we have, uniformly over $\max_t|\mu_t|\le (K+2)\sqrt{\eta_N\log N}$,
\begin{align*}
& \TV(\cL(\widetilde\cX|\widetilde\beta_1, ..., \widetilde\beta_K,\mu_t), \cL(\widebar\cX|\widetilde\beta_1, ..., \widetilde\beta_K,\mu_t)) \\
\le&  \sum_{t=1}^{2T}\sum_{j_1=1}^{d_1}\cdots\sum_{j_K=1}^{d_K} \TV(\cL(\widetilde\cX_{t,j_1,...,j_K}|\widetilde\beta_{1,j_1}, ..., \widetilde\beta_{K,j_K},\mu_t), \cL(\widebar\cX_{t,j_1,...,j_K}|\widetilde\beta_{1,j_1}, ..., \widetilde\beta_{K,j_K},\mu_t)) \\
\le& \frac{C_K}{N},
\end{align*}
for some constant $C_K>0$. Let $\rmint \cL(\widetilde\cX|\widetilde\beta_1,...,\widetilde\beta_K,\mu) \rmintd \bar\xi(\mu)$ (resp. $\rmint \cL(\widetilde\cX|\widetilde\beta_1,...,\widetilde\beta_K,\mu) \rmintd \xi(\mu)$) be the conditional distribution of $\widetilde\cX| \widetilde\beta_1,...,\widetilde\beta_K$ if the elements of $\mu=(\mu_1,...,\mu_{2T})$ are i.i.d. following $\bar\xi$ (resp. $\xi$). And $\rmint \cL(\widebar\cX|\widetilde\beta_1,...,\widetilde\beta_K,\mu) \rmintd \bar\xi(\mu)$, $\rmint \cL(\widebar\cX | \widetilde\beta_1, ..., \widetilde\beta_K, \mu) \rmintd \xi(\mu)$ are similarly defined. Note that
\begin{align*}
\TV(\xi,\bar\xi)=\int_{|\mu|>(K+2)\sqrt{\eta_N \log N} } \rmintd \xi(\mu) = \int_{|x|>(K+2)\sqrt{\log N}} \phi_0(x) \rmintd x \le CN^{-K-3}.
\end{align*}
Then, we can obtain
\begin{align*}
&\TV (\int\cL(\widetilde\cX|\widetilde\beta_1,...,\widetilde\beta_K,\mu)\rmintd \bar\xi(\mu), \int\cL(\widebar\cX|\widetilde\beta_1,...,\widetilde\beta_K,\mu) \rmintd \xi(\mu) ) \\
\le& \TV (\int\cL(\widetilde\cX|\widetilde\beta_1,...,\widetilde\beta_K,\mu)\rmintd \bar\xi(\mu), \int\cL(\widebar\cX|\widetilde\beta_1,...,\widetilde\beta_K,\mu) \rmintd \bar\xi(\mu) ) \\
&\quad + \TV (\int\cL(\widebar\cX|\widetilde\beta_1,...,\widetilde\beta_K,\mu)\rmintd \bar\xi(\mu), \int\cL(\widebar\cX|\widetilde\beta_1,...,\widetilde\beta_K,\mu) \rmintd \xi(\mu) ) \\
\le& \sup_{\max_t|\mu_t|\le (K+2)\sqrt{\eta_N\log N}}  \TV(\cL(\widetilde\cX|\widetilde\beta_1, ..., \widetilde\beta_K,\mu_t), \cL(\widebar\cX|\widetilde\beta_1, ..., \widetilde\beta_K,\mu_t))   + C_0 (2T)^{K} \TV(\bar\xi,\xi) \\
\le& C_K N^{-1}.
\end{align*}
Define a set
\begin{align*}
S_T := \{(j_1,j_2,...,j_K): \widetilde \beta_{1,j_1} = \cdots= \widetilde \beta_{K,j_K}=1, 1\le j_k\le d_k, 1\le k\le K \} .
\end{align*}
Then, for each given $(\widetilde\beta_1, ..., \widetilde\beta_K)$, we can define $s_k=\sum_{j_k\in S_T} \widetilde \beta_{k,j_k}=\sum_{j_k\in S_T} \widetilde \beta_{k,j_k}^2$, $ a_k=s_k^{-1/2}(\widetilde \beta_{k,j_k}\mathbf{1}_{\{j_k\in S_T\} })$, for all $1\le k\le K$, and $\lambda=\eta_N^{1/2} \prod_{k=1}^K s_k^{1/2} $. Obviously, 
there exists one-to-one identification between $( a_1,..., a_K,\lambda)$ and $(\widetilde\beta_1, ..., \widetilde\beta_K)$.
Note that $\rmint \cL(\widebar\cX| \widetilde\beta_1,..., \widetilde\beta_K, \mu) \rmintd \xi(\mu) = \P_{\lambda,\ba}^{2T}$. As $\cL(\widetilde\cX|\widetilde\beta_1, ..., \widetilde\beta_K) = \rmint \cL(\widetilde \cX| \widetilde\beta_1,..., \widetilde\beta_K, \mu) \rmintd \bar\xi(\mu)$, we have
\begin{align*}
\TV(\cL(\widetilde\cX|\widetilde\beta_1, ..., \widetilde\beta_K), \P_{\lambda,\ba}^{2T}) \le C_K N^{-1}.
\end{align*}

Define an event $Q=\{\widetilde\beta_1,...,\widetilde\beta_K: |s_k-\kappa|\le \kappa/2, 1\le k\le K\}$. By Lemma \ref{lowerbdd:lemma4}, $\P(Q^c)\le K\kappa (2T/N)^{\kappa}$. Let $\widetilde\pi$ be the joint distribution of $(\ba,\lambda)$, and $\pi$ be the distribution by restricting $\widetilde\pi$ on $\{ a_1(\widetilde\beta_1), ...,  a_K(\widetilde\beta_K), \lambda(\widetilde\beta_1,..., \widetilde\beta_K) : \widetilde\beta_1,...,\widetilde\beta_K \in Q\}$. It follows that $\TV(\widetilde\pi,\pi) \le C\P(Q^c) \le CK\kappa (2T/N)^{\kappa}$. As $\cL(\widetilde\cX| \widetilde\beta_1,..., \widetilde\beta_K)= \cL(\widetilde\cX|  a_1,..., a_K,\lambda)= \cL(\widetilde\cX| \ba,\lambda)$,
\begin{align*}
\TV(\cL(\widetilde\cX),\P_{\pi}) &\le \TV(\cL(\widetilde\cX), \int \cL(\widetilde\cX| \ba,\lambda) \rmintd \pi(\ba,\lambda) )  \\
&\qquad+ \TV( \int \cL(\widetilde\cX| \ba,\lambda) \rmintd \pi(\ba,\lambda), \int\P_{\lambda,\ba}^{2T}\rmintd \pi(\lambda,\ba)) \\
&\le \TV(\widetilde\pi,\pi) + \sup_{\ba,\lambda} \TV( \cL(\widetilde\cX| \ba,\lambda), \P_{\lambda,\ba}^{2T}) \\
&\le CK\kappa \left(\frac{2T}{N} \right)^{\kappa} + \frac{C_K}{N}.
\end{align*}
Hence, we have
\begin{align*}
\TV(\cL(\{\cX_t\})_{t=1}^{2T}, \int\P_{\lambda,\ba}^{2T}\rmintd \pi(\lambda,\ba) ) \le C_{1K}\cdot\kappa \left(\frac{2T}{N} \right)^{\kappa} + \frac{C_{2K}}{N} + \frac{4(K+1)T}{N}.
\end{align*}

When $d_k\ge 2T$, $1\le k\le K$, we first use the above arguments to analyze the distribution of the first $2T$ coordinates. Then, as the remaining $2T$ coordinates are exact, the total variation bound is zero.
\end{proof}

\begin{lemma}\label{lowerbdd:lemma4}
Let $s_k=\sum_{j_k\in S_T} \widetilde \beta_{k,j_k}=\sum_{j_k\in S_T} \widetilde \beta_{k,j_k}^2$, $d_k=2T< N$ for all $1\le k\le K$. Define an event $Q=\{\widetilde\beta_1,...,\widetilde\beta_K: |s_k-\kappa|\le \kappa/2, 1\le k\le K\}$, then
\begin{align*}
\P(Q)\ge 1-\frac{K(\kappa+1)}{2} \left(\frac{2T}{N} \right)^{\kappa}.
\end{align*}
\end{lemma}

\begin{proof}
Recall that $\kappa<\sqrt{N}$, Let $\cC_k=\{\widetilde\beta_{k,j_k}:j_k \in S_T \}$, $1\le k\le K$. Then
\begin{align*}
\P\left(|\cC_k|\le \kappa/2 \right) \le& \frac{\sum_{i=0}^{\kappa/2} \binom{d_k}{i}} {\binom{N}{\kappa}} \le \frac{\kappa+1}{2}\cdot \frac{\binom{d_k}{\kappa/2}} {\binom{N}{\kappa}} =\frac{\kappa+1}{2}\cdot \frac{d_k(d_k-1)\cdots(d_k-\kappa/2+1)\cdot \kappa !}{(\kappa/2)! \cdot N(N-1)\cdots (N-\kappa+1) } \\
\le& \frac{\kappa+1}{2}\cdot \left(\frac{d_k}{N} \right)^{\kappa} .
\end{align*}
Therefore, by Bonferroni inequality, we have the desired result.
\end{proof}

\subsection{Proof of Theorem \ref{lowerbdd}}
Write $\widehat\Sigma=\frac1T \sum_{t=T+1}^{2T}\cX_t\otimes\cX_t$. Then the test \eqref{eq:testhpc} can be rewritten as
\begin{align*}
\psi=\psi(\cX_1,...,\cX_{2T})=\psi(\cA,\mu,\cB_0,\cB_1):=\mathbf{1}\left\{ \widehat\Sigma \times_{k=1}^{2K} \widehat a_k^\top \ge 1+ \frac12 \left( \frac{\kappa}{2} \right)^{K} \eta_N \right\}.
\end{align*}
Note that $\psi$ is a test for the Hypergraphic Planted Clique detection problem \eqref{hpc}. Recall the probability space \eqref{eq:tenfm2sp}. For any $(\lambda,\ba)$ in the support of $\pi$, we have
\begin{align*}
\P_{\lambda,\ba}^T \in \widetilde\sP\left(T,d_1,...,d_K, \lambda \right)
\end{align*}
with $\eta_N^{1/2}(\kappa/2)^{K/2} \le \lambda \le \eta_N^{1/2}(3\kappa/2)^{K/2}$.

We first bound the Type-I error of the test $\psi$. By Lemma \ref{lowerbdd:lemma2},
\begin{align*}
\P_{H_0^G}(\psi=1) \le \P_{0}^T(\psi=1)+C_K N^{-1}.
\end{align*}
Under $\P_0^T$, $\widehat\ba=(\widehat a_1,...,\widehat a_K)$ and $\widehat\Sigma$ are independent. Conditioning on $\widehat\ba$, applying Bernstein's inequality, we have
\begin{align*}
\widehat\Sigma \times_{k=1}^{2K} \widehat a_k^\top= 1+\frac1T\sum_{t=T+1}^{2T} \left( \left| \cX_t\times_{k=1}^{K} \widehat a_k^\top \right|^2 - |\vec1(\otimes_{k=1}^K\widehat a_k)|^2 \right) > 1+ \frac12 \left( \frac{\kappa}{2} \right)^{K} \eta_N,
\end{align*}
with probability at most $\exp\left(-\frac{C_K T\kappa^{K+2}}{N^2 (\log N)^4 }\right)$. Integrating over $\widehat\ba$, we have
\begin{align}\label{lowerbdd:error1}
\P_{H_0^G}(\psi=1) \le \exp\left(-\frac{C_K T\kappa^{K+2}}{N^2 (\log N)^4 }\right)+C_K N^{-1}.
\end{align}

Next, we bound the Type-II error. By Lemma \ref{lowerbdd:lemma3},
\begin{align*}
\P_{H_1^G}(\psi=0) \le \P_{\pi}(\psi=0) + C_{1K}\cdot\kappa \left(\frac{2T}{N} \right)^{\kappa} + \frac{C_{2K}}{N} + \frac{4(K+1)T}{N}    .
\end{align*}
Recall that under the probability space \eqref{eq:tenfm2sp},
\begin{align*}
\cX_t=\lambda f_t\times_1  a_1\times_2...\times_K a_K+\cE_t,
\end{align*}
and $f_t \overset{i.i.d}{\sim} N(0,1)$ and the elements of $\cE_t$ follow that $\cE_{t,j_1,...,j_K}\overset{i.i.d.}{\sim} N(0,1)$ for all $1\le t\le T, 1\le j_k\le d_k, 1\le k\le K$. Write $\cE_t(\widehat\ba)=\cE_t\times_{k=1}^{K} \widehat a_k^\top$. Thus,
\begin{align*}
\widehat\Sigma \times_{k=1}^{2K} \widehat a_k^\top &= \lambda^2 \prod_{k=1}^K |\widehat a_k^\top a_k|^2 \left(\frac1T \sum_{t=T+1}^{2T}f_t^2 \right) + \frac{2\lambda}{T} \sum_{t=T+1}^{2T} \prod_{k=1}^K (\widehat a_k^\top a_k)\cE_t(\widehat\ba) + \frac1T \sum_{t=T+1}^{2T} |\cE_t(\widehat\ba)|^2 \\
&=\lambda^2+1 + \lambda^2 \prod_{k=1}^K |\widehat a_k^\top a_k|^2 \left(\frac1T \sum_{t=T+1}^{2T}f_t^2 -1 \right) + \lambda^2 \left( \prod_{k=1}^K |\widehat a_k^\top a_k|^2 -1 \right) \\
&\qquad + \frac{2\lambda}{T} \sum_{t=T+1}^{2T} \prod_{k=1}^K (\widehat a_k^\top a_k)\cE_t(\widehat\ba) + \left( \frac1T \sum_{t=T+1}^{2T} |\cE_t(\widehat\ba)|^2-1 \right).
\end{align*}
After rearrangement, we have
\begin{align*}
\left| \widehat\Sigma \times_{k=1}^{2K} \widehat a_k^\top -(\lambda^2+1) \right| &\le \lambda^2  \left|\frac1T \sum_{t=T+1}^{2T}f_t^2 -1 \right| + \lambda^2 \left( 1-\prod_{k=1}^K |\widehat a_k^\top a_k|^2  \right) \\
&\qquad + \left| \frac{2\lambda}{T} \sum_{t=T+1}^{2T} \cE_t(\widehat\ba)  \right| + \left| \frac1T \sum_{t=T+1}^{2T} |\cE_t(\widehat\ba)|^2-1 \right|.
\end{align*}
By Bernstein's inequality, we can obtain
\begin{align*}
&\P_{\lambda,\ba}^T \left( \left|\frac1T \sum_{t=T+1}^{2T}f_t^2 -1 \right| +  \left| \frac{1}{T} \sum_{t=T+1}^{2T} \cE_t(\widehat\ba)  \right| + \left| \frac1T \sum_{t=T+1}^{2T} |\cE_t(\widehat\ba)|^2-1 \right|\ge C\sqrt{\frac{\log T}{T}} \right) \\
&\le T^{-C'} .
\end{align*}
Note that
\begin{align*}
&\lambda^2\left( 1-\prod_{k=1}^K |\widehat a_k^\top a_k|^2 \right) \ge \lambda^2- \lambda^2\max_k |\widehat a_k^\top a_k|^{2} = \lambda^2\min_k\|P_{\widehat a_k}- P_{ a_k}\|_{\rm S}^2.
\end{align*}
Hence, as $\lambda^2\ge \eta_N (\kappa/2)^{K}$, the Type-II error is upper bounded by
\begin{align}\label{lowerbdd:error2}
\P_{H_1^G}(\psi=0) \le& \P_{\lambda,\ba}^T \left(\min_k\|P_{\widehat a_k}- P_{ a_k}\|_{\rm S}^2 >\frac13 \right) + T^{-C'} + C_{1K}\cdot\kappa \left(\frac{2T}{N} \right)^{\kappa} \\
&\quad+ \frac{C_{2K}}{N} + \frac{4(K+1)T}{N}    .   \notag
\end{align}

Combining \eqref{lowerbdd:error1} and \eqref{lowerbdd:error2}, we have
\begin{align}\label{lowerbdd:error}
&\P_{H_0^G}(\psi=1)+\P_{H_1^G}(\psi=0) \\
\le& \P_{\lambda,\ba}^T \left(\min_k\|P_{\widehat a_k}- P_{ a_k}\|_{\rm S}^2 >\frac13 \right) + \exp\left(-\frac{C_K T\kappa^{K+2}}{N^2 (\log N)^4 }\right)+C_K N^{-1}    \notag\\
&\quad + T^{-C'} + C_{1K}\cdot\kappa \left(\frac{2T}{N} \right)^{\kappa} + \frac{C_{2K}}{N} + \frac{4(K+1)T}{N}    .   \notag
\end{align}

Now, we are ready to prove Theorem \ref{lowerbdd}. Consider the Hypergraphic Planted Clique detection problem \eqref{hpc} with $N=20(K+1)T$, $\kappa=\lfloor N^{1/2-\delta} \rfloor$ and $\delta<1/2-1/(K+2)$. Then we have
\begin{align}\label{lowerbdd:asmp1}
\frac{N(\log N)^5}{\kappa^{K+2}} \le c_0 ,
\end{align}
for some sufficient small constant $c_0>0$. We can also obtain
\begin{align*}
\lambda^2 \le \eta_N\left(\frac{3\kappa}{2} \right)^K \le \frac{C_{K} d^{1/2-\delta(K+1)/K}}{T^{1/2}(\log T)^2}  .
\end{align*}
Obviously, \eqref{eq1:lowerbdd} holds with $\vartheta=\delta(K+1)/K < (K+1)/(2K+4)$. On the contradictory, suppose that the claim of Theorem \ref{lowerbdd} does not hold. It means that
\begin{align}\label{eq0:lowerbdd}
\liminf_{T\to \infty} \sup_{\cX_1,...,\cX_T\in \sP(T,d_1,...,d_K,\lambda) } \P \left( \min_{1\le k\le K} \| P_{\widehat a_k} -P_{ a_k} \|_{\rm S}^2 > \frac13  \right)\le\frac14 .
\end{align}
Substituting \eqref{lowerbdd:asmp1} and \eqref{eq0:lowerbdd} into \eqref{lowerbdd:error}, we have
\begin{align*}
\P_{H_0^G}(\psi=1)+\P_{H_1^G}(\psi=0) \le& \frac14 \left(1+o(1) \right) +\frac15+  N^{-C_{0K}} +C_K N^{-1} + T^{-C'}  \\
&\quad  + C_{1K}\cdot\kappa \left(\frac{1}{10K} \right)^{\kappa} + \frac{C_{2K}}{N}     .
\end{align*}
It follows that
\begin{align*}
\limsup_{N\to\infty} \left( \P_{H_0^G}(\psi=1)+\P_{H_1^G}(\psi=0)   \right) <\frac12,
\end{align*}
which contradicts the Hypothesis \ref{hypos:hpc}. We complete the proof.

\section{Proof of Theorem \ref{thm:stat_lowerbdd}} \label{proofth4}
Without loss of generality, we can assume $\sigma=1$. Applying the same reduction in Appendix \ref{proofth3}, we map the probability space \eqref{eq:tenfm1sp} to
\begin{align*}
&\widetilde\sP(\lfloor T/2 \rfloor,d_1,...,d_K,\lambda) =\Big\{\cX_1,...,\cX_{\lfloor T/2 \rfloor}:\cX_t=\lambda \widetilde \cF_t\times_1  U_1\times_2...\times_K U_K+\widetilde \cE_t, \text{ with }  U_k\in\R^{d_k\times r_k},   \\
&\qquad\qquad U_k^\top U_k=I_{r_k}, 1\le k\le K, \widetilde \cF_{t,i_1,...,i_K} \overset{i.i.d.}{\sim} N(0,1/2), \{\widetilde \cF_t\}_{t=1}^{\lfloor T/2 \rfloor} \text{ independent of } \{\widetilde\cE_t\}_{t=1}^{\lfloor T/2 \rfloor},   \notag\\
&\qquad\qquad \widetilde \cE_{t,j_1,...,j_K}\overset{i.i.d.}{\sim} N(0,1/2) \text{ for all } 1\le t\le \lfloor T/2 \rfloor, 1\le i_k \le r_k, 1\le j_k\le d_k, 1\le k\le K \Big\}. \notag
\end{align*}
For notation convenience, in the following of this section, we study the probability space
\begin{align}\label{eq:tenfm3sp}
&\widetilde\sP(T,d_1,...,d_K,\lambda) =\Big\{\cX_1,...,\cX_{T}:\cX_t=\lambda \cF_t\times_1  U_1\times_2...\times_K U_K+ \cE_t, \text{ with }  U_k\in\R^{d_k\times r_k},   \\
&\qquad\qquad U_k^\top U_k=I_{r_k}, 1\le k\le K, \cF_{t,i_1,...,i_K} \overset{i.i.d.}{\sim} N(0,1), \{\cF_t\}_{t=1}^{T} \text{ independent of } \{\cE_t\}_{t=1}^{T},   \notag\\
&\qquad\qquad \cE_{t,j_1,...,j_K}\overset{i.i.d.}{\sim} N(0,1) \text{ for all } 1\le t\le T, 1\le i_k \le r_k, 1\le j_k\le d_k, 1\le k\le K \Big\}. \notag
\end{align}
Under probability space \eqref{eq:tenfm3sp}, we denote the distribution of $\cX_t$ as $\P_{\lambda,\cU}$ with $\cU=( U_1,..., U_K)$, and the joint distribution of $(\cX_1,...,\cX_{T})$ as $\P_{\lambda,\cU}^{T}$. Let $G(k, r)$ denote the Grassmannian manifold consisting of all $r$-dimensional linear subspace of $\R^k$, $O(p, r) =\{ U\in \R^{p\times r}: U^\top U = I_r\}$. Let $B(\theta,\epsilon)=\{\theta' \in\Theta: \rho(\theta',\theta)\le \epsilon\}$ for some metric $\rho$ and parameter space $\Theta$.

To prove Theorem \ref{thm:stat_lowerbdd}, we use two technical lemmas in \cite{cai2013sparsepca} (Proposition 3 and Lemma 1), regarding a minimax lower bound via the local metric entropy and the metric entropy of the Grassmannian manifold $G(k,r)$.

\begin{lemma}\label{lem1:stat_lowerbdd}
Let $(\Theta,\rho)$ be a totally bounded metric space and $\{\P_{\theta}: \theta\in\Theta\}$ a collection of probability measures. For any $E\subset\Theta$, denote by $N(E,\epsilon)$ the $\epsilon$-covering number of $E$, that is, the minimal number of balls of radius $\epsilon$ whose union contains $E$. Denote by $M(E,\epsilon)$ the $\epsilon$-packing number of $E$, that is, the maximal number of points in $E$ whose pairwise distance is at least $\epsilon$. Put
\begin{align*}
L:=\sup_{\theta\neq\theta'} \frac{D(\P_{\theta} || \P_{\theta' })}{\rho^2(\theta,\theta')}.
\end{align*}
If there exist $0 < c_0 < c_1 <\infty$ and $p\ge 1$ such that
\begin{align*}
\left(\frac{c_0}{\epsilon}\right)^p \le N(\Theta,\epsilon) \le   \left(\frac{c_1}{\epsilon}\right)^p
\end{align*}
for all $0<\epsilon <\epsilon_0$. Then
\begin{align}
\inf_{\widehat\theta}\sup_{\theta\in\Theta} \E_{\theta} [\rho^2(\widehat\theta(X),\theta)] \ge \frac{c_0^2}{840c_1^2}  \left( \frac{p}{L}  \wedge \epsilon_0^2 \right)   .
\end{align}
\end{lemma}

\begin{lemma}\label{lem2:stat_lowerbdd}
For any $V\in O(k, r)$, identifying the subspace span$(V)$ with its projection matrix $VV^\top$, define the metric on $G(k, r)$ by $\rho(VV^\top,UU^\top)=\|VV^\top-UU^\top\|_{\rm F}$. Then for any $\epsilon\in(0,\sqrt{2(r\wedge (k-r))}]$,
\begin{align*}
\left(\frac{c_0}{\epsilon}\right)^{r(k-r)} \le N(G(k,r),\epsilon) \le   \left(\frac{c_1}{\epsilon}\right)^{r(k-r)}     ,
\end{align*}
where $c_0, c_1$ are absolute constants. Moreover, for any $V\in O(k, r)$ and any $\alpha\in(0, 1)$,
\begin{align*}
M(B(V,\epsilon),\alpha\epsilon) \ge \left( \frac{c_0}{\alpha c_1} \right)^{r(k-r)}.
\end{align*}
\end{lemma}

\begin{proof}[Proof of Theorem \ref{thm:stat_lowerbdd}]
Note that the Kullback-Leibler divergence between normal distributions is given by $D(N(0,\Sigma_1)||N(0,\Sigma_0))=\frac12 ( \tr(\Sigma_0^{-1}\Sigma_1 -I) -\log\det \Sigma_0^{-1}\Sigma_1 )$. Under probability space \eqref{eq:tenfm3sp}, $\E\vec1(\cF_t)\vec1(\cF_t)^\top =I_r$. Consider $k\in \{1,2,...,K\}$ and fix the loading space of the other $K-1$ tensor modes, i.e. $U_{\ell}$ for all $\ell\neq k$. Let $\cU=( U_1,..., U_K)$, $\cV=(V_1,..., V_K)$, and $V_\ell=U_\ell$ for all $\ell\neq k$. Then the Kullback-Leibler divergence of $\P_{\lambda,\cU}^{T}$ with respect to $\P_{\lambda,\cV}^{T}$ is given by
\begin{align}
&D(\P_{\lambda,\cV}^{T} || \P_{\lambda,\cU}^{T}) \\
=& \frac{T}{2} \tr\left( -\frac{\lambda^2}{\lambda^2+1} (V_K \odot\cdots\odot V_{1})(V_K \odot\cdots\odot V_{1})^\top  \right)  \notag \\
&+\frac{T}{2} \tr\left( \lambda^2 (U_K \odot\cdots\odot U_{1})(U_K \odot\cdots\odot U_{1})^\top  \right) \notag \\
&-\frac{T}{2} \tr\left( -\frac{\lambda^4}{\lambda^2+1} (V_K \odot\cdots\odot V_{1})(V_K \odot\cdots\odot V_{1})^\top (U_K \odot\cdots\odot U_{1})(U_K \odot\cdots\odot U_{1})^\top  \right) \notag \\
=& \frac{T\lambda^4}{2(1+\lambda^2)} \left(r -r_{-k} \|U_k^\top V_k\|_{\rm F}^2 \right) \notag \\
=& \frac{Tr_{-k}\lambda^4}{4(1+\lambda^2)} \|V_k V_k^\top - U_k U_k^\top \|_{\rm F}^2 , \notag
\end{align}
where the first and second inequalities are by the matrix inversion property and
the facts that $\tr(U_{\ell} U_{\ell}^\top) = \tr(U_{\ell}^\top U_{\ell}) = r_{\ell}$, $\tr(V_{k} V_{k}^\top) = \tr(V_{k}^\top V_{k}) = r_{k}$, $r=\prod_{\ell=1}^K r_{\ell}$, $r_{-k}=r/r_k$, respectively.
In view of Lemma \ref{lem1:stat_lowerbdd}, we have $L=Tr_{-k}\lambda^4/(4(1+\lambda^2))$. Applying Lemma \ref{lem2:stat_lowerbdd} yields the statistical lower bound
\begin{align}
\inf_{\widehat U_k}  \sup_{\cX_1,...,\cX_T\in \widetilde\sP(T,d_1,...,d_K,\lambda) } \E  \|\widehat U_k \widehat  U_k^\top  - U_k  U_k^\top \|_{\rm F}^2 \gtrsim
r_k\wedge (d_k-r_k) \wedge \frac{(\sigma^4+\sigma^2\lambda^2)d_kr_k}{\lambda^4 Tr_{-k}}
\end{align}
for all $1\le k\le K$. As $\|\widehat U_k \widehat  U_k^\top  - U_k  U_k^\top \|_{\rm F}^2 \le 2r_k \|\widehat U_k \widehat  U_k^\top  - U_k  U_k^\top \|_{\rm S}^2$, it follows that
\begin{align*}
\inf_{\widehat U_k}  \sup_{\cX_1,...,\cX_T\in \widetilde\sP(T,d_1,...,d_K,\lambda) } \E  \|\widehat U_k \widehat  U_k^\top  - U_k  U_k^\top \|_{\rm S}^2 &\ge \inf_{\widehat U_k}  \sup_{\cX_1,...,\cX_T\in \widetilde\sP(T,d_1,...,d_K,\lambda) } \E  \frac{1}{2r_k}\|\widehat U_k \widehat  U_k^\top  - U_k  U_k^\top \|_{\rm F}^2 \\
&\gtrsim
1 \wedge \frac{(\sigma^4+\sigma^2\lambda^2)d_k}{\lambda^4 Tr_{-k}} .
\end{align*}
Then we can obtain the desired statistical lower bound.
\end{proof}

\section{Proofs of Proposition \ref{prop:topup} and Corollaries} \label{proofcor}

\begin{proof}[Proofs of Proposition \ref{prop:topup}]
It directly follows from Lemma \ref{lm-pertubation}, by extending the rank one case ($r_k=1,1\le k\le K$) in \cite{chen2022rejoinder} to the general case.
\end{proof}

\begin{proposition}\label{prop:connumber}
Let $\lambda=\prod_{k=1}^K \|A_k\|_{\rm S}$. Assume that the condition numbers of $A_k^\top A_k$ ($k=1,...,K$) are bounded. Then, for all $1\le k\le K$, we have,
\begin{align*}
&\|\Theta_{k,0}\|_{\rm op} \asymp \lambda^2 \|\Phi_{k,0}\|_{\rm op}, \quad\quad \|\Theta_{k,0}^*\|_{\rm S} \asymp \lambda^2 \|\Phi_{k,0}^*\|_{\rm S}, \\
&\tau_{k,r_k} \asymp \lambda^2 \times \sigma_{r_k}\left({\rm{mat}}_1(\Phi_{k,1:h_0}) \right), \\
&\tau_{k,r_k}^* \asymp \lambda^2 \times \sigma_{r_k}\left(\Phi_{k,1:h_0}^{*\cano}/\lambda^2 \right) .
\end{align*}
\end{proposition}
\begin{proof}
If the condition numbers of $A_k^\top A_k$ ($k=1,...,K$) are bounded, all the singular values of $A_k$ are at the same order. Then Proposition \ref{prop:connumber} immediately follows.
\end{proof}

\begin{proof}[Proofs of Corollary \ref{cor:itopup1} and \ref{cor:itipup1}]
Employing Proposition \ref{prop:connumber}, under Assumption \ref{asmp:factor} and $\E [{\rm{mat}}_1(\Phi_{k,1:h_0})]$ is of rank $r_k$, we can show
$\lambda_k\asymp \lambda$ and $\tau_{k,r_k} \asymp \lambda^2$. When the ranks $r_k$ are fixed, the second part of condition \eqref{condition1n} can be written as $C_{1,K}^{\iter}{R}^{\ideal} \le \rho<1$.
Thus, for $C_1^{\topup}=1/(6R^{(0)})\le(1-\rho)/4$ and $C_{1,K}^{\iter}=\rho/{R}^{\ideal}=1/R^{(0)}$, we have
$$\rho =C_{1,K}^{\iter}{R}^{\ideal}= {R}^{\ideal}/{R}^{(0)} \asymp 1/(\max_{k}\sqrt{d_{-k}}).$$
For $m=1$, this gives the rate ${R}^{\ideal}$ by \eqref{bound:itopup:cora2}.
Then Corollary \ref{cor:itopup1} follows from the results of Theorem \ref{thm:itopup}.

Similarly, Applying Proposition \ref{prop:connumber}, under Assumption \ref{asmp:factor} and $\E [\Phi_{k,1:h_0}^{*\cano}/\lambda^2]$ is of rank $r_k$, we can obtain $\lambda_k^*\asymp \lambda$ and $\tau_{k,r_k}^* \asymp \lambda^2$. Then Corollary \ref{cor:itipup1} follows from the results of Theorem \ref{thm:itipup}.
\end{proof}

\begin{proof}[Proofs of Corollaries \ref{cor:itopup2},  \ref{cor:itipup2} and \ref{cor:tipup-itopup2}]
By Assumption \ref{asmp:strength}, \eqref{condition1n} holds when
$$
T \ge C_0 \max_{k}\left(\frac{d^{2\delta_1-\delta_0}r^{2}}{r_k}+\frac{d^{2\delta_1}rr_k}{d_k} + \frac{d_k^{1/2} \sqrt{rr_k}}{d^{1/2-\delta_1}}
+ \frac{d_{-k}^*\sqrt{rr_k}}{d^{1-\delta_1}}
+\frac{\sqrt{d_kd_{-k}^{*}}r}{d^{1-\delta_1}} \right).
$$
Because $(d_jr_j\sqrt{rr_k}+\sqrt{d_kd_jr_j}r)r_k/d + \sqrt{d_krr_k/d}r_k\le 3r^2$ for $j\neq k$ and $\delta_1\ge \delta_0$, the last three terms on the right-hand side above can be absorbed into the first, so that \eqref{eq:sample:itopup} suffices. Therefore, Theorem \ref{thm:itopup} yields Corollary \ref{cor:itopup2} by setting $\rho\asymp (R^{\ideal}+R^{\add})/R^{(0)}$ as in the proof of Corollary \ref{cor:itopup1}.

Similarly, setting $\rho\asymp \min_k(R^{*\ideal}+R^{*\add}) \lam_k^{*2}/(R^{*(0)}\| \Theta_{k,0}^* \|_{\rm S}), $ and applying Assumption \ref{asmp:strength} to Theorem \ref{thm:itipup}, leads to Corollary \ref{cor:itipup2}.

For Corollary \ref{cor:tipup-itopup2},  condition \eqref{condition1n*} holds when
$T$ is no smaller than
$$
C_0
\max_{1\le k\le K}\left(
\frac{d_kr_kr_{-k}^{2\delta_2}}{d^{1+\delta_0-2\delta_1}}
+ \frac{r_k^2r_{-k}^{2\delta_2}}{d^{1-2\delta_1}}
+ \frac{d_kr_{-k}^2r}{d^{1+\delta_0-2\delta_1}}
+ \frac{d_kr_{-k}^2r r_k}{d^{2-2\delta_1}}
+ \frac{d_{-k}^*\sqrt{rr_k}}{d^{1-\delta_1}}
+\frac{\sqrt{d_kd_{-k}^{*}}r}{d^{1-\delta_1}}
\right)
$$
due to $d_k\sqrt{r_{-k}}\sqrt{rr_k}/d^{1-\delta_1}\le d_kr_kr_{-k}/d^{1+\delta_0-2\delta_1}$.
However, the sixth term on the right-hand side above can be absorbed into the first due to $\sqrt{d_k d_jr_j}r
\le (\sqrt{d_k}r_{-k})(\sqrt{d_jr_j}r_{-j})
\le \max_k d_kr_kr_{-k}^2$ for $j\neq k$.
\end{proof}


\section{Techinical Lemmas} \label{section:lemmas}

%
%

\begin{lemma}\label{lemma:epsilonnet}
Let $d, d_j, d_*, r\le d\wedge d_j$ be positive integers, $\epsilon>0$ and
$N_{d,\epsilon} = \lfloor(1+2/\epsilon)^d\rfloor$. \\
(i) For any norm $\|\cdot\|$ in $\R^d$, there exist
$M_j\in \R^d$ with $\|M_j\|\le 1$, $j=1,\ldots,N_{d,\epsilon}$,
such that $\max_{\|M\|\le 1}\min_{1\le j\le N_{d,\epsilon}}\|M - M_j\|\le \epsilon$.
Consequently, for any linear mapping $f$ and norm $\|\cdot\|_*$,
$$
\sup_{M\in \R^d,\|M\|\le 1}\|f(M)\|_* \le 2\max_{1\le j\le N_{d,1/2}}\|f(M_j)\|_*.
$$
(ii) Given $\epsilon >0$, there exist $U_j\in \R^{d\times r}$
and $V_{j'}\in \R^{d'\times r}$ with $\|U_j\|_{\rm S}\vee\|V_{j'}\|_{\rm S}\le 1$ such that
$$
\max_{M\in \R^{d\times d'},\|M\|_{\rm S}\le 1,\text{rank}(M)\le r}\
\min_{j\le N_{dr,\epsilon/2}, j'\le N_{d'r,\epsilon/2}}\|M - U_jV_{j'}^\top\|_{\rm S}\le \epsilon.
$$
Consequently, for any linear mapping $f$ and norm $\|\cdot\|_*$ in the range of $f$,
\begin{equation}\label{lm-3-2}
\sup_{M, \widetilde M\in \R^{d\times d'}, \|M-\widetilde M\|_{\rm S}\le \epsilon
\atop{\|M\|_{\rm S}\vee\|\widetilde M\|_{\rm S}\le 1\atop
\text{rank}(M)\vee\text{rank}(\widetilde M)\le r}}
\frac{\|f(M-\widetilde M)\|_*}{\epsilon 2^{I_{r<d\wedge d'}}}
\le \sup_{\|M\|_{\rm S}\le 1\atop \text{rank}(M)\le r}\|f(M)\|_*
\le 2\max_{1\le j \le N_{dr,1/8}\atop 1\le j' \le N_{d'r,1/8}}\|f(U_jV_{j'}^\top)\|_*.
\end{equation}
(iii) Given $\epsilon >0$, there exist $U_{j,k}\in \R^{d_k\times r_k}$
and $V_{j',k}\in \R^{d'_k\times r_k}$ with $\|U_{j,k}\|_{\rm S}\vee\|V_{j',k}\|_{\rm S}\le 1$ such that
$$
\max_{M_k\in \R^{d_k\times d_k'},\|M_k\|_{\rm S}\le 1\atop \text{rank}(M_k)\le r_k, \forall k\le K}\
\min_{j_k\le N_{d_kr_k,\epsilon/2} \atop j'_k\le N_{d_k'r_k,\epsilon/2}, \forall k\le K}
\Big\|\odot_{k=2}^K M_k - \odot_{k=2}^K(U_{j_k,k}V_{j_k',k}^\top)\Big\|_{\rm op}\le \epsilon (K-1).
$$
For any linear mapping $f$ and norm $\|\cdot\|_*$ in the range of $f$,
\begin{equation}\label{lm-3-3}
\sup_{M_k, \widetilde M_k\in \R^{d_k\times d_k'},\|M_k-\widetilde M_k\|_{\rm S}\le\epsilon\atop
{\text{rank}(M_k)\vee\text{rank}(\widetilde M_k)\le r_k \atop \|M_k\|_{\rm S}\vee\|\widetilde M_k\|_{\rm S}\le 1\ \forall k\le K}}
\frac{\|f(\odot_{k=2}^KM_k-\odot_{k=2}^K\widetilde M_k)\|_*}{\epsilon(2K-2)}
\le \sup_{M_k\in \R^{d_k\times d_k'}\atop {\text{rank}(M_k)\le r_k \atop \|M_k\|_{\rm S}\le 1, \forall k}}
\Big\|f\big(\odot_{k=2}^K M_k\big)\Big\|_*
\end{equation}
and
\begin{equation}\label{lm-3-4}
\sup_{M_k\in \R^{d_k\times d_k'},\|M_k\|_{\rm S}\le 1\atop \text{rank}(M_k)\le r_k\ \forall k\le K}
\Big\|f\big(\odot_{k=2}^K M_k\big)\Big\|_*
\le 2\max_{1\le j_k \le N_{d_kr_k,1/(8K-8)}\atop 1\le j_k' \le N_{d_k'r_k,1/(8K-8)}}
\Big\|f\big(\odot_{k=2}^K U_{j_k,k}V_{j_k',k}^\top\big)\Big\|_*.
\end{equation}
\end{lemma}

\begin{proof} (i) The covering number $N_\epsilon$ follows from the standard volume comparison argument
as the $(1+\epsilon/2)$-ball under $\|\cdot\|$ and centered at the origin contains no more than
$(1+2/\epsilon)^d$ disjoint $(\epsilon/2)$-balls centered at $M_j$. The inequality follows from
the ``subtraction argument",
$$
\sup_{\|M\|\le 1}\|f(M)\|_* -\max_{1\le j\le N_{d,1/2}}\|f(M_j)\|_*
\le \sup_{\|M-M_j\|\le 1/2}\|f(M-M_j)\|_*
\le \sup_{\|M\|\le 1}\|f(M)\|_*/2.
$$
(ii) The covering numbers are given by applying (i) to both $U$ and $V$ in the decomposition $M=UV^\top$
as Lemma 7 in \citet{zhang2018tensor}.
The first inequality in \eqref{lm-3-2} follows from the fact that for $r<d\wedge d'$,
$(M-\widetilde M)/\epsilon$ is a sum of two
rank-$r$ matrices with no greater spectrum norm than 1,
and the second inequality of \eqref{lm-3-2} again follows from the subtraction argument
although we need to split $M-U_jV_{j'}^\top$ into two rank $r$ matrices to result in
an extra factor of 2. \\
(iii) The proof is nearly identical to that of part (ii). The only difference is the factor $K-1$ when
$\|\odot_{k=2}^KM_k-\odot_{k=2}^K\widetilde M_k\|_{\rm op}\le (K-1)\max_{2\le k\le K}\|M_k-\widetilde M_k\|_{\rm S}$
is applied.
\end{proof}

\begin{lemma}\label{lm-GH}
(i) Let $G\in \R^{d_1\times n}$ and $H\in \R^{d_2\times n}$ be two centered independent
Gaussian matrices such that $\E(u^\top \text{vec}(G))^2 \le \sigma^2\ \forall\ u\in \R^{d_1n}$ and
$\E(v^\top \text{vec}(H))^2\le \sigma^2\ \forall\ v\in \R^{d_2n}$. Then,
\bes\label{lm-GH-1}
\|GH^\top\|_{\rm S} \le \sigma^2\big(\sqrt{d_1d_2}+\sqrt{d_1n} + \sqrt{d_2n}\big)
+ \sigma^2x(x+2\sqrt{n}+\sqrt{d_1}+\sqrt{d_2})
\ees
with at least probability $1 - 2e^{-x^2/2}$ for all $x\ge 0$. \\
(ii) Let $G_i\in \R^{d_1\times d_2}, H_i\in \R^{d_3\times d_4}, i=1,\ldots, n$,
be independent centered Gaussian matrices
such that $\E(u^\top \text{vec}(G_i))^2 \le \sigma^2\ \forall\ u\in \R^{d_1d_2}$ and
$\E(v^\top \text{vec}(H_i))^2\le \sigma^2\ \forall\ v\in \R^{d_3d_4}$. Then,
\bes
\bigg\|\text{mat}_1\bigg(\sum_{i=1}^n G_i\otimes H_i\bigg)\bigg\|_{\rm S}
&\le& \sigma^2\big(\sqrt{d_1n}+\sqrt{d_1d_3d_4} + \sqrt{nd_2d_3d_4}\big)
\cr && + \sigma^2 x\big(x + \sqrt{n} + \sqrt{d_1} + \sqrt{d_2}+\sqrt{d_3d_4}\big)
\ees
with at least probability $1 - 2e^{-x^2/2}$ for all $x\ge 0$.
\end{lemma}

\begin{proof}
Assume $\sigma=1$ without loss of generality. Let $x\ge 0$. \\
(i) Independent of $G$ and $H$, let 
$\zeta_j\in\R^{n}$, $j=1,2$,
be independent standard Gaussian vectors.
As in \citet{chen2022factor}, 
the Sudakov-Fernique inequality provides
\bes
\E\Big[\|GH^\top\|_{\rm S}\Big|G\Big]
\le \E\Big[\max_{\|u\|_2=1}u^\top G \zeta_2\Big|G\Big] + \|G\|_{\rm S}\sqrt{d_2}.
\ees
Thus, by the Gaussian concentration inequality
\bes
\P\bigg\{\|GH^\top\|_{\rm S}
\ge \E\Big[\max_{\|u\|_2=1}u^\top G \zeta_2\Big|G\Big] + \|G\|_{\rm S}(\sqrt{d_2}+x)\bigg|G\bigg\}\le e^{-x^2/2}.
\ees
Applying the Sudakov-Fernique inequality again, we have
\bes
\E\Big[\E\Big[\max_{\|u\|_2=1}u^\top G \zeta_2\Big|G\Big] + \|G\|_{\rm S}(\sqrt{d_2}+x)\Big]
\le \sqrt{d_1n} + (\sqrt{d_1}+\sqrt{n})(\sqrt{d_2}+x).
\ees
Moreover, as the Lipschitz norm of
$\E\big[\max_{\|u\|_2=1}u^\top G \zeta_2\big|G\big] + \|G\|_{\rm S}(\sqrt{d_2}+x)$
is bounded by $\sqrt{n}+\sqrt{d_2}+x$, by the Gaussian concentration inequality
\bes
&& \E\Big[\max_{\|u\|_2=1}u^\top G \zeta_2\Big|G\Big] + \|G\|_{\rm S}(\sqrt{d_2}+x)
\cr &\le& \sqrt{d_1n} + (\sqrt{d_1}+\sqrt{n})(\sqrt{d_2}+x)
+ x\big(\sqrt{n}+\sqrt{d_2}+x\big)
\ees
holds with at least probability $1-e^{-x^2/2}$.

(ii) We treat $G=(G_1,\ldots,G_n)\in\R^{d_1\times d_2\times n}$ and
$H=(H_1,\ldots,H_n)\in\R^{d_3\times d_4\times n}$ as tensors.
Let $\xi = (\xi_1,\ldots,\xi_n)\in \R^{d_2\times n}$ be a standard Gaussian matrix
independent of $H$.
For $u\in \R^{d_1}$ and $V\in \R^{d_2\times (d_3d_4)}$,
\bes
&& \E\bigg[\bigg\|\text{mat}_1\bigg(\sum_{i=1}^n G_i\otimes H_i\bigg)\bigg\|_{\rm S}\bigg|H\bigg]
\cr & = & \E\bigg[\sup_{\|u\|_2=1,\|V\|_{\rm F}=1} u^\top\text{mat}_1(G)
\text{vec}\big(\text{mat}_3(H)V^\top\big)\bigg|H\bigg]
\cr & \le & \sqrt{d_1}\sup_{\|V\|_{\rm F}=1} \|\text{mat}_3(H)V^\top\|_{\rm F}
 + \E\bigg[\sup_{\|V\|_{\rm F}=1} (\text{vec}(\xi))^\top\text{vec}\big(\text{mat}_3(H)V^\top\big)\bigg|H\bigg]
\cr & = & \sqrt{d_1}\|\text{mat}_3(H)\|_{\rm S}
 + \E\bigg[\bigg(\sum_{j=1}^{d_2}\sum_{k=1}^{d_3d_4}
\bigg(\sum_{i=1}^n \xi_{i,j}\text{vec}(H_i)_k\bigg)^2\bigg)^{1/2} \bigg|H\bigg]
\cr & \le & \sqrt{d_1}\|\text{mat}_3(H)\|_{\rm S}  + \sqrt{d_2}\|\text{vec}(H)\|_2
\ees
By the Gaussian concentration inequality,
\bes
\P\bigg\{\bigg\|\text{mat}_1\bigg(\sum_{i=1}^n G_i\otimes H_i\bigg)\bigg\|_{\rm S}
\ge (\sqrt{d_1}+x)\|\text{mat}_3(H)\|_{\rm S}  + \sqrt{d_2}\|\text{vec}(H)\|_2 \bigg|H\bigg\}\le e^{-x^2/2}.
\ees
Moreover, as $\E\big[(\sqrt{d_1}+x)\|\text{mat}_3(H)\|_{\rm S}  + \sqrt{d_2}\|\text{vec}(H)\|_2\big]
\le (\sqrt{d_1}+x)\big(\sqrt{n}+\sqrt{d_3d_4}\big) + \sqrt{d_2nd_3d_4}$ and
the Lipschitz norm of $(\sqrt{d_1}+x)\|\text{mat}_3(H)\|_{\rm S}  + \sqrt{d_2}\|\text{vec}(H)\|_2$
is bounded by $\sqrt{d_1}+x + \sqrt{d_2}$,
\bes
&& (\sqrt{d_1}+x)\|\text{mat}_3(H)\|_{\rm S}  + \sqrt{d_2}\|\text{vec}(H)\|_2
\cr & \le & (\sqrt{d_1}+x)\big(\sqrt{n}+\sqrt{d_3d_4}\big) + \sqrt{nd_2d_3d_4}
+ x\big(\sqrt{d_1}+x + \sqrt{d_2}\big)
\ees
holds with at least probability $1-e^{-x^2/2}$.
\end{proof}

\begin{proof}[Proof of Lemma \ref{lm-pertubation}]
Because ${\widehat U}_\perp^\top M M^\top  U
= {\widehat U}_\perp^\top {\widehat M} {\widehat M}^\top U
- {\widehat U}_\perp^\top {\widehat M} \Delta^\top U- {\widehat U}_\perp^\top \Delta M^\top U$,
\bes
{\widehat U}_\perp^\top U
= \big({\widehat U}_\perp^\top {\widehat M} {\widehat M}^\top {\widehat U}_\perp{\widehat U}_\perp^\top U - {\widehat U}_\perp^\top {\widehat M} {\widehat V}_\perp{\widehat V}_\perp^\top \Delta^\top U- {\widehat U}_\perp^\top \Delta UU^\top M^\top U\big)(U^\top M M^\top  U)^{-1}
\ees
as $\lam_r>0$. It follows that
$\|{\widehat U}_\perp^\top U\|
\le \big({\widehat\lam}_{r+1}^2\|{\widehat U}_\perp^\top U\|
+{\widehat\lam}_{r+1} \|{\widehat V}_\perp \Delta^\top U\|
+ \|{\widehat U}_\perp^\top \Delta U\|\lam_r\big)/\lam_r^2$.
The conclusion \eqref{wedin+} follows by algebra,
and \eqref{wedin-2} by $\epsilon_1\vee{\widehat\lam}_{r+1}\le\|\Delta\|_{S}$.
\end{proof}

%

\end{document}